\documentclass{scrartcl}

\KOMAoptions{paper=letter}
\usepackage{fix-cm}
	\usepackage{xspace}
	\usepackage{amsmath}
	\usepackage{amsthm}
    \usepackage{amsfonts}
    \usepackage{amssymb}
	\usepackage[hidelinks]{hyperref} 

	\usepackage{nicefrac}
	\usepackage{tikz}
 
 \usetikzlibrary{positioning,arrows,shapes,decorations,calc,fit,arrows.meta} 
	\usepackage{colortbl}
	\usepackage{booktabs}
	\usepackage{cleveref}
	\usepackage{array}
	\usepackage{graphicx}
	\usepackage{xcolor}
	\usepackage{enumitem}
	\usepackage{tabularx}
	\usepackage{thm-restate}
	\usepackage{ragged2e}
    \usepackage{microtype}
    \usepackage{cleveref}
    \usepackage{mathrsfs} 
    \usepackage{caption}

\date{}

\title{
Robust Voting Rules on the Interval Domain }
\author {
    \Large Patrick Lederer\\
    \Large p.lederer@unsw.edu.au\\
    \Large UNSW Sydney
}
\usepackage[round]{natbib}

\newcounter{remarkcount}

\newcommand{\ssucc}{\mathrel{\vartriangleright}}

\newcommand{\ssucceq}{\mathrel{\trianglerighteq}}

    \def\multiset#1#2{\ensuremath{\left(\kern-.3em\left(\genfrac{}{}{0pt}{}{#1}{#2}\right)\kern-.3em\right)}}

	\theoremstyle{definition}
	\newtheorem{definition}{Definition}
	
	\theoremstyle{plain}
	\newtheorem{theorem}{Theorem}
	\newtheorem{lemma}{Lemma}

\makeatletter
\renewcommand{\paragraph}{%
  \@startsection{paragraph}{4}%
  {\z@}{2.25ex \@plus 1ex \@minus .2ex}{-1em}%
  {\normalfont\normalsize\bfseries}%
}
\makeatother

\usepackage[backgroundcolor=orange!15!white,bordercolor=orange!80!black,textsize=small]{todonotes} % Todo command
\presetkeys{todonotes}{inline}{}

\newcolumntype{L}[1]{>{\raggedright\let\newline\\\arraybackslash\hspace{0pt}}m{#1}}
\newcolumntype{C}[1]{>{\centering\let\newline\\\arraybackslash\hspace{0pt}}m{#1}}
\newcolumntype{R}[1]{>{\raggedleft\let\newline\\\arraybackslash\hspace{0pt}}m{#1}}

\newcommand{\dom}{\Lambda^*}

\begin{document}

\maketitle

\begin{abstract}
     In this paper, we study voting rules on the interval domain, where the alternatives are arranged according to an externally given strict total order and voters report intervals of this order to indicate the alternatives they support. For this setting, we introduce and characterize the class of position-threshold rules, which compute a collective position of the voters with respect to every alternative and choose the left-most alternative whose collective position exceeds its threshold value. Our characterization of these rules mainly relies on reinforcement, a well-known population consistency condition, and robustness, a new axiom that restricts how the outcome is allowed to change when a voter removes the left-most or right-most alternative from his interval. Moreover, we characterize a generalization of the median rule to the interval domain, which selects the median of the endpoints of the voters' intervals.
\end{abstract}



\section{Introduction}
A ubiquitous phenomenon in today's societies is collective decision-making: given the possibly conflicting preferences of multiple agents, a joint decision should be reached in a fair and principled way. 
Such processes of collective decision-making are formally studied in the field of social choice theory, where researchers analyze voting rules from a mathematical perspective. While this research has led to significant advances in the understanding of voting rules \citep[e.g.,][]{ASS11a,BCE+15a}, social choice theory is plagued by numerous far-reaching impossibility theorems, which demonstrate that no voting rule can satisfy several desirable properties at the same time \citep[e.g.,][]{Arro51a,Gibb73a,Satt75a,Moul88b}. 

One of the most successful escape routes to such impossibility theorems is to impose more structure on the voters' preferences. 
In particular, the notion of \emph{single-peaked preferences}, which goes back to \citet{Blac48a}, has turned out fruitful. The idea of single-peaked preferences is that there is a strict total order $\ssucc$ over the alternatives and that the voters' preferences are structurally aligned with this order. Specifically, the preference relation of every voter specifies an ideal alternative $p$ such that alternatives become less preferred when moving further away from $p$ with respect to $\ssucc$.
For single-peaked preferences, it is known that most impossibility theorems vanish and there is large consensus on which voting rules to use \citep[e.g.,][]{BoJo83a,Spru95a,Chin97a,EhSt08a,Weym11a}: the median rule and, more generally, the phantom median rules introduced by \citet{Moul80a} have superior axiomatic properties.
Intuitively, the median rule sorts the voters with respect to their top-ranked alternatives according to $\ssucc$ and then returns the favorite alternative of the median voter. Moreover, phantom median rules generalize this method by computing the median rule for the $n$ original voters and $n-1$ phantom voters who report a fixed single-peaked preference relation. 

The appeal of these phantom median rules lies in the fact that they are the only voting rules on the domain of single-peaked preferences that satisfy anonymity (i.e., all voters are treated equally), unanimity (i.e., an alternative is guaranteed to be chosen if it is the favorite alternative of all voters), and strategyproofness (i.e., voters cannot benefit by lying about their true preferences) \citep{Moul80a,Weym11a}. This combination of axioms is remarkable because no voting rule satisfies all three properties on the domain of all preference relations \citep{Gibb73a,Satt75a}. Moreover, phantom median rules satisfy numerous further desirable properties such as tops-onlyness (i.e., voters only need to report their favorite alternative), reinforcement (i.e., when combining two elections with the same winner, the winner remains the same), and participation (i.e., voters cannot benefit by abstaining) \citep[][]{Moul84c,JLPV24a}. 

Despite the success of phantom median rules on the domain of single-peaked preferences, rather little is known about voting rules when slightly modifying the setting. In particular, we are interested in the case that voters indicate the alternatives they like by reporting intervals, i.e., sets of consecutive alternatives with respect to~$\ssucc$. We refer to this problem as voting on the interval domain and we believe that this domain naturally arises when there are many similar alternatives that can be ordered with respect to an external order $\rhd$. For example, when voting about the start time of a group dinner, it seems plausible that the voters prefer to report a range of start times instead of a single minute-specific time. As another, more high-stakes example, in political elections with many candidates, voters may wish to, e.g., vote for all left-wing candidates instead of singling out one such candidate to reduce their effort for casting a vote. Motivated by such settings, we will suggest and axiomatically characterize a new class of voting rules for the interval domain. 

While we assume that voters report intervals, we have not yet discussed the interpretation of these intervals in depth. For instance, are voters indifferent between alternatives in their interval or do intervals contain all alternatives that exceed a voter-specific utility threshold? We believe that our results are largely independent from such questions: because we will focus on natural consistency conditions, our results appear reasonable as long as we can interpret a voter's interval as the set of alternatives he likes or supports. Moreover, from a technical side, we do not need any information other than the voters' intervals, i.e., our results are not bound to a particular model for inferring the voters' intervals. Nevertheless, to allow for an unambiguous interpretation of our axioms and to relate our work to the literature, we will assume that voters have weakly single-peaked preferences and that the voters' intervals correspond to their sets of favorite alternatives. 

Such weakly single-peaked preferences generalize single-peaked preferences by allowing for ties: alternatives must only be weakly less preferred when moving away from the voter's ideal alternative. 
Importantly, the voters' sets of favorite alternatives are known to form intervals for weakly single-peaked preferences \citep[e.g.,][]{ElLa15a,Pupp16a}, thus giving a direct connection to the interval domain. 
Moreover, weakly single-peaked preferences are well-suited to capture the voters' preferences in our previous examples and have been repeatedly studied before \citep[e.g.,][]{Moul84c,Berg98a,AuBa00a,BeMo09a}. Most of these papers focus on strategyproofness and aim to generalize Moulin's characterization of phantom median rules. However, an appealing closed-form characterization of strategyproof voting rules for weakly single-peaked preferences remains elusive as indifferences are challenging to handle with strategyproofness. While we will focus in this paper on axioms other than strategyproofness and employ a variable-electorate setting, our results will also give new insights for designing appealing strategyproof rules for the weakly single-peaked domain.

\paragraph{Contribution.} In this paper, we will study voting rules on the interval domain and introduce the class of position-threshold rules. Intuitively, these rules determine for every alternative a collective position, which quantifies the voters' relative position regarding this alternative, and choose the left-most alternative  whose collective position exceeds its threshold value. In more detail, position-threshold rules for $m$ alternatives are defined by a weight vector $\alpha\in [0,1]^m$ and a threshold vector $\theta\in(0,1)^m$. The weight vector $\alpha$ is used to quantify the relative position of the voters to the alternatives: a voter's relative position to an alternative $x_i$ is $0$ 
if all alternatives in his interval are right of $x_i$, $1$ if all alternatives in his interval are left of or equal to $x_i$, and $\alpha_i$ otherwise. Then, a position-threshold rule computes the collective position of an alternative as the average of the voters' individual positions to this alternative, and it returns the left-most alternative whose collective position exceeds its threshold value. While it may not be clear from this description, position-threshold rules generalize phantom median rules because phantom median rules can also be formulated via collective positions and threshold values. 

As our main contribution, we will provide an axiomatic characterization of position-threshold rules in this paper. Specifically, we will show that position-threshold rules are the only voting rules on the interval domain that satisfy five conditions called robustness, reinforcement, right-biased continuity, anonymity, and unanimity (\Cref{thm:characterization}).

The two central conditions for this characterization are robustness and reinforcement. Robustness formalizes that if a voter removes an alternative from his interval, the outcome can only change to reflect this modification or not at all. More precisely, robustness postulates that if a voter removes his left-most (resp. right-most) alternative from his interval, the outcome changes from the voter's old left-most (resp. right-most) alternative to his new left-most (resp. right-most) alternative or it remains unchanged. While this axiom is new, it is conceptually related to many prominent conditions such as localizedness \citep{Gibb77a}, Maskin-monotonicity \citep{Mask99a}, and uncompromisingness \citep{BoJo83a}. We believe that robustness is desirable as it ensures that the outcome smoothly changes along $\rhd$ when voters modify their intervals and that the outcome is independent of, e.g., intervals that do not contain the winner. Moreover, we prove that robustness implies strategyproofness when the voters' intervals correspond to their sets of favorite alternatives in weakly single-peaked preference relations. 

Our second main condition, reinforcement, demands that if an alternative is chosen in two disjoint elections, it should also be chosen in a combined election. Notably, this condition requires us to work with a variable-electorate setting, whereas most prior works on (weakly) single-peaked preferences focus on a fixed electorate \citep[e.g.,][]{Moul80a,Moul84a,Berg98a}. While this means that our results are logically incomparable to these prior works, the assumption of reinforcement allows for a clean characterization of position-threshold rules. Moreover, we note that the variable-electorate setting as well as variants of reinforcement lie at the heart of numerous influential works in social choice theory \citep[e.g.,][]{Smit73a,Youn75a,YoLe78a,Fish78d,Bran13a,LaSk21a,BrPe19a,Lede23a}. 

In addition to reinforcement and robustness, we use a third non-standard axiom called right-biased continuity, which is a variant of a continuity condition that is commonly studied for scoring rules \citep[e.g.,][]{Smit73a,Youn75a,Myer95b,LaSk21a,Lede23a}. Intuitively, continuity requires that the effect of a group of voters on the outcome can be marginalized by cloning the other voters sufficiently often. However, since this condition is not compatible with resolute voting rules, we weaken it: the effect of a group of voters $G$ on the outcome can only be fully marginalized by cloning the other voters if the outcome chosen for the remaining voters is right to the outcome chosen for $G$ with respect to $\rhd$. Otherwise, we can only ensure that the intervals of the remaining voters are respected in a weaker sense. Finally, anonymity and unanimity are two standard conditions, which respectively state that all voters should be treated equally and that we need to choose an alternative if all voters unanimously and exclusively report this alternative. 

Based on our characterization of position-threshold rules, we further aim to extend the median rule to the interval domain. To this end, we note that the median rule is the only phantom median rule that guarantees to select an alternative that is top-ranked by a strict majority of the voters whenever such an alternative exists. In the context of intervals, we call this condition the majority criterion and formalize it by requiring that an alternative is chosen if it is uniquely reported by more than half of the voters. We show that there is only a single position-threshold rule that satisfies the majority criterion and the natural condition of strong unanimity (if the intersection of the intervals of all voters is non-empty, an alternative in this intersection needs to be chosen): the endpoint-median rule (\Cref{thm:AWS1/2}). Intuitively, this rule replaces the interval of each voter with two singleton ballots corresponding to the left-most and right-most alternative in the interval, and then executes the median rule. In combination with \Cref{thm:characterization}, we derive that the endpoint-median rule is the only voting rule on the interval domain that satisfies anonymity, strong unanimity, the majority criterion, robustness, reinforcement, and right-biased continuity.

\paragraph{Related work.} Our paper is closely related to the study of strategyproof voting rules on the domain of single-peaked preference relations, which has garnered significant attention. In particular, Moulin's characterization of phantom median rules \citep{Moul80a} has inspired a large body of follow-up works, including extensions to multi-dimensional variants of single-peakedness \citep{BoJo83a,Zhou91b,BGS93a}, non-anonymous variants of phantom median rules \citep{Chin97a}, and randomized versions of this result \citep{EPS02a,PRSS14a,PyUn15a}. For a comprehensive overview of early studies on strategyproof voting rules for single-peaked preferences, we refer the reader to the surveys by \citet{Spru95a} and \citet{Weym11a}.
More recent research \citep[e.g.,][]{CSS13a,Reff15a,CSZ16a,ChMa18a,ChZe21a} focuses on strategyproof voting rules for somewhat technical extensions of the single-peaked domain, such as semi-single-peaked preferences. In addition to demonstrating the existence of non-dictatorial and strategyproof voting rules on these domains, these works show that, under suitable side conditions, the corresponding domains are necessary for the existence of such rules. 

In contrast to the aforementioned papers, we assume that the voters report intervals instead of a single favorites alternative or a rankings of the alternatives. To the best of our knowledge, comparable settings have only been studied by \citet{Moul84c}, \citet{Berg98a}, and \citet{BeMo09a}, who investigate voting rules for variants of weakly single-peaked preference relations.
In more detail, \citet{Moul84c} characterizes the class of generalized Condorcet winner rules, which roughly compute the median rule with respect to the endpoints of the voters' intervals and some additional parameters, based on two axioms similar to Arrow's independence of irrelevant alternatives. By contrast, \citet{Berg98a} and \citet{BeMo09a} focus on strategyproof voting rules for the single-plateaued domain but do not provide a closed-form characterization of such rules.
Our results differ from these works as we focus on robustness and reinforcement and characterize the class of position-threshold rules based on these axioms.

Furthermore, our paper is related to the problem of facility location, where a public-good facility needs to be placed on the real line based on the voter's preferences over the possible positions \citep[e.g.,][]{PrTe13a,FFG16a,CFLL+21a}. In particular, in facility location it is typically assumed that the voters report their ideal positions for the facility and that their cost for a location is its distance to their ideal position. Put differently, facility location can be seen as voting on the real line with single-peaked preferences. The goal of facility location is to identify strategyproof voting rules that optimize some objective such as the utilitarian or egalitarian social welfare, a problem for which Moulin's (\citeyear{Moul80a}) characterization has proven invaluable. Since recent works on facility location also investigate scenarios where the voters report intervals instead of a single location \citep{ELZ22a,ZZML23a}, we believe that our results can also provide valuable insights for this setting.

Finally, voting on the interval domain is loosely connected to the problem of interval aggregation \citep[e.g.,][]{FaCo11a,KlPr20a,ENT22a}, where multiple input intervals need to be aggregated into an output interval. However, as we aim to choose a single alternative based on the voters' intervals, we end up with rather different axioms than those considered in interval aggregation.

\section{The Model}

We will use a variable-electorate framework in this paper and thus let $\mathbb{N}=\{1,2,\dots\}$ denote an infinite set of voters and $A=\{x_1, \dots, x_m\}$ a finite set of $m\geq 2$ alternatives. Intuitively, $\mathbb{N}$ is the set of all possible voters and a concrete electorate $N$ is a finite and non-empty subset of $\mathbb{N}$. The set of all electorates is therefore defined by $\mathcal{F}(\mathbb{N})=\{N\subseteq\mathbb{N}\colon N\text{ is non-empty and finite}\}$. 
The central assumption in this paper is that there is an externally given strict total order $\ssucc$ over the alternatives. Throughout the paper, we will assume that $\ssucc$ is given by $x_1\ssucc x_2\ssucc \dots \ssucc x_m$, and we define the relation $\ssucceq$ by $x_i\ssucceq x_j$ if and only if $x_i=x_j$ or $x_i\ssucc x_j$ for all $x_i,x_j\in A$. 

Given an electorate $N\in\mathcal{F}(\mathbb{N})$, the voters are asked to report intervals of $\ssucc$ to indicate the alternatives they like. 
Formally, a set of alternatives $I$ is an \emph{interval} (with respect to $\ssucc$) if $x_i\in I$ and $x_k\in I$ implies $x_j\in I$ for all alternatives $x_i,x_j,x_k\in A$ with $x_i\ssucc x_j\ssucc x_k$. Since intervals are fully specified by their endpoints, we define $[x_i,x_k]=\{x_j\in A\colon x_i\ssucceq x_j\ssucceq x_k\}$ as the interval from $x_i$ to $x_k$. The set of all intervals, or the \emph{interval domain}, is given by $\Lambda=\{[x_i,x_j]\subseteq A\colon x_i\ssucceq x_j\}$. Next, an \emph{interval profile} $\mathcal{I}=(I_{i_1}, \dots, I_{i_n})$ for a given electorate $N=\{i_1,\dots, i_n\}$ contains the interval of every voter $i\in N$, i.e., it is a function from $N$ to $\Lambda$. The set of all interval profiles for a fixed electorate $N$ is defined by $\Lambda^N$ and the set of all possible interval profiles is $\Lambda^*=\bigcup_{N\in\mathcal{F}(\mathbb{N})} \Lambda^N$. We denote by $N_{\mathcal I}$ the set of voters that report an interval in a profile $\mathcal I$ and by $n_\mathcal{I}=|N_\mathcal{I}|$ the size of this set.

To give a clear economic interpretation to the interval domain, we will assume that the interval of each voter corresponds to his favorite alternatives in a weakly single-peaked preference relation. We note, however, that this assumption is technically not necessary for our main results, and we believe that our axioms may also be of interest for other interpretations of the interval domain.\footnote{We see at least three more plausible interpretations of the interval domain for our results. Firstly, one can interpret intervals as the set of approved candidates for the candidate-interval domain, which is set of weakly single-peaked preferences that are additionally dichotomous \citep[e.g.,][]{ElLa15a}. Secondly, voters may infer their interval from single-peaked preferences based on a threshold model, similar to those studied for approval voting \citep[e.g.,][]{Niem84a,BrFi07c,Endr13a}. Specifically, the idea here is that voters have single-peaked preferences and infer their interval by reporting all alternatives that exceed a (voter-specific) utility threshold. Thirdly, intervals may naturally arise in the context of partial preferences, where voters have not fully explored the alternatives. In this case, a voter's interval corresponds to the alternatives that he view as the most desirable ones, but there may not be any preference between alternatives not in the interval.}
To introduce weakly single-peaked preference relations, we first define \emph{preference relations $\succsim$} as complete and transitive binary relations over the alternatives, where $x\succsim y$ means that the respective voter weakly prefers $x$ to $y$. As usual, a preference relation is \emph{strict} if it does not contain ties. Further, we denote by $T({\succsim})=\{x\in A\colon \forall y\in A:x\succsim y\}$ the set of most preferred alternatives in a preference relation $\succsim$ and observe that $|T(\succsim)|=1$ if $\succsim$ is strict. Analogous to interval profiles, a \emph{preference profile} $R$ maps every voter in a given electorate to a preference relation. 

Now, the idea of weak single-peakedness is that each voter has an ideal alternative $x$ and that voters weakly prefer alternatives that are closer to $x$ with respect to $\rhd$. More formally, a preference relation $\succsim$ is \emph{weakly single-peaked} (with respect to $\rhd$) if there is an alternative $x$ such that $x\succsim y\succsim z$ for all $y,z\in A$ with $x\rhd y\rhd z$ or $z\rhd y\rhd x$. Moreover, a preference relation is called \emph{single-peaked} if it is weakly single-peaked and strict. We define by $\mathcal{R}_\rhd$ and $\mathcal{P}_\rhd$ the set of weakly single-peaked and single-peaked preference relation, respectively. It is well-known that the set of a voters' favorite alternatives $T(\succsim)$ is an interval if $\succsim$ is weakly single-peaked \citep{ElLa15a,Pupp16a}. We will thus interpret each voter's interval $I_i$ as his set of most preferred alternatives in his weakly single-peaked preference relation $\succsim_i$, i.e., $I_i=T(\succsim_i)$. Lastly, we say a preference profile is (weakly) single-peaked if every voter has a (weakly) single-peaked preference relation, and we define by $\mathcal{P}_\rhd^N$ and $\mathcal{R}_\rhd^N$ (resp. $\mathcal{P}_\rhd^*$ and $\mathcal{R}_\rhd^*$) the sets of all (weakly) single-peaked preference profiles for a fixed electorate $N$ (resp. for all electorates).

The goal in this paper is to select a single winning alternative based on interval profiles. We will thus study \emph{voting rules}, which are functions that map every interval profile $\mathcal{I}\in\dom$ to a single alternative $x\in A$. Since we interpret intervals as the sets of the voters' favorite alternatives in weakly single-peaked preference relations, we may equivalently view voting rules as functions that map every weakly single-peaked profile $R\in \mathcal{R}_\rhd^*$ to a single winning alternative $x\in A$ and that only depend on the voters' favorite alternatives. This latter property is commonly known as plateau-onlyness \citep[e.g.,][]{Berg98a,BeMo09a}.

\subsection{Robustness}\label{subsec:robustness}

We will next introduce the central axiom for our analysis called robustness. The rough idea of this axiom is that if a voter removes his left-most or right-most alternative from his interval, the outcome is only allowed to change to reflect this modification. In more detail, robustness requires that, if a voter removes the left-most (resp. right-most) alternative from his interval, then the winner cannot change at all or the winner changes from the old left-most (resp. right-most) alternative to the new left-most (resp. right-most) alternative. To formalize this, we let $\mathcal I^{i\downarrow x}$ denote the profile derived from another profile $\mathcal{I}$ by removing alternative $x$ from voter $i$'s interval. We note that $\mathcal{I}^{i\downarrow x}$ is a valid interval profile only if $|I_i|\geq 2$ and $x$ is the left-most or right-most alternative in $I_i$. 

\begin{definition}[Robustness]\label{def:robustness}
A voting rule $f$ is \emph{robust} if, for all interval profiles $\mathcal I\in \dom$, voters $i\in N_{\mathcal I}$, and alternatives $x_\ell$, $x_r$ such that $I_i=[x_\ell, x_r]$ and $x_\ell\ssucc x_r$, it holds that 
    \begin{enumerate}[label=(\roman*), leftmargin = *]
    \item $f(\mathcal{I})=f(\mathcal{I}^{i\downarrow x_\ell})$, or $f(\mathcal{I})=x_\ell$ and $f(\mathcal{I}^{i\downarrow x_\ell})=x_{\ell+1}$, and 
    \item $f(\mathcal{I})=f(\mathcal{I}^{i\downarrow x_r})$, or $f(\mathcal{I})=x_r$ and $f(\mathcal{I}^{i\downarrow x_r})=x_{r-1}$.
\end{enumerate}
\end{definition}

Robustness can equivalently be formulated in terms of adding an alternative to a voter's interval: if we, e.g., add a new left-most alternative to a voter's interval, the winner cannot change at all or it changes from the voter's old left-most alternative to his new left-most alternative. Similar invariance notions have been studied before \citep[e.g.,][]{Gibb77a,Saij87a,Mask99a,MuSa17a}, with the most prominent examples being localizedness and Maskin-monotonicity. We thus believe that robustness is a desirable property as it prohibits that the outcome changes in an unexpected way. 

To provide further motivation for robustness, we will next show this axiom is closely related to well-known incentive compatibility conditions, namely strategyproofness and uncompromisingness, when assuming that the voters infer their intervals from weakly single-peaked preference relations. Specifically, we will prove that robustness implies strategyproofness and is equivalent to a strengthening of uncompromisingness when treating voting rules on the interval domain as plateau-only voting rules on the domain of weakly single-peaked preference relations. To define these incentive compatibility conditions, we say two profiles $\mathcal I$ and $\mathcal I'$ are \emph{$i$-variants} if they only differ in the interval of voter $i$, i.e., $N_\mathcal{I}=N_{\mathcal{I'}}$ and there is a voter $i\in N_{\mathcal I}$ such that $I_j=I_j'$ for all $j\in N_{\mathcal I}\setminus \{i\}$.

\paragraph{Strategyproofness.} Strategyproofness requires that voters cannot benefit by lying about their true preferences. Since we assume that voters have weakly single-peaked preference relations $\succsim_i$, we formalize this by requiring that it is always in the best interest of the voters to truthfully reveal their set of favorite alternatives $T(\succsim_i)$. Hence, a voting rule $f$ on $\dom$ is \emph{strategyproof} if $f(\mathcal{I})\succsim_i f(\mathcal{I}')$ for all voters $i$, $i$-variants $\mathcal{I},\mathcal{I}'\in \dom$, and weakly single-peaked preference relation $\succsim_i$ such that $I_i=T(\succsim_i)$. 

\paragraph{Uncompromisingness.} Uncompromisingness is commonly studied in the context of single-peaked preference relations \citep[e.g.,][]{BoJo83a,Spru95a,EPS02a} and informally requires that the outcome is not allowed to change as long as the manipulator's favorite alternative remains on the same side of the current winner with respect to $\rhd$. In the context of single-peaked preferences, this condition turns out to be equivalent to strategyproofness \citep{BoJo83a}.
Moreover, \citet{Berg98a} extended this notion to voting rules on the interval domain (or, equivalently, plateau-only voting rules on the weakly single-peaked domain) by requiring that the outcome is not allowed to change if the manipulator's full interval stays on the same side of the current winner. More formally, a voting rule $f$ on $\dom$ is \emph{uncompromising} if $f(\mathcal I)=f(\mathcal I')$ for all voters $i$, $i$-variants $\mathcal I,\mathcal I'\in\dom$, and intervals $I_i=[x_\ell,x_r]$ and $I_i'=[x_\ell', x_r']$ such that \emph{(i)} $x_r\rhd f(\mathcal I)$ and $x_{r}'\ssucceq f(\mathcal I)$ or \emph{(ii)} $f(\mathcal I)\ssucc x_\ell$ and $f(\mathcal I)\ssucceq x_\ell$. 
We, however, observe that this uncompromisingness axiom has no consequences if $f(\mathcal I)\in [x_\ell, x_r]$.
We thus introduce \emph{strong uncompromisingness}, which requires three more conditions than uncompromisingness: it must also hold that $f(\mathcal I)=f(\mathcal I')$ 
if \emph{(iii)} $x_\ell\ssucc f(\mathcal I)\ssucc x_r$ and $x_\ell'\ssucceq f(\mathcal I)\ssucceq x_r'$, 
\emph{(iv)} $x_\ell=f(\mathcal I)\rhd x_r$ and $x_\ell'=f(\mathcal I)\ssucceq x_r'$, or 
\emph{(v)} $x_\ell\rhd f(\mathcal I)= x_r$ and $x_\ell'\ssucceq f(\mathcal I)= x_r'$. Put differently, these conditions ensure that the the winning alternative $f(\mathcal I)$ is not allowed to change unless the relative position of the endpoints of the manipulator's interval to  $f(\mathcal I)$ change. 
\medskip

As we show next, robustness is equivalent to strong uncompromisingness and it implies strategyproofness. We note here that our definition of strategyproofness is based on the assumption that voters have weakly single-peaked preference relations, so it is not necessary to refer to these preferences in the following proposition.

\begin{restatable}{proposition}{sincerity}\label{prop:sincerity}
    A voting rule on $\dom$ is robust if and only if it is strongly uncompromising. Further, every robust voting on $\dom$ rule is strategyproof.
\end{restatable}
\begin{proof}
    Fix a voting rule $f$ on the interval domain, a profile $\mathcal I$, and a voter $i\in N_{\mathcal I}$. Further, we denote voter $i$'s interval in $\mathcal I$ by $I_i=[x_\ell, x_r]$. We break down this proposition into three implications: strong uncompromisingness implies robustness, robustness implies strong uncompromisingness, and robustness implies strategyproofness.
    
    \paragraph{Strong uncompromisingness implies robustness.} We first assume that $f$ is strongly uncompromising and show that it is robust. To this end, assume that $x_\ell\neq x_r$ and consider the profile $\mathcal I^{i\downarrow x_\ell}$ where voter $i$ removes $x_\ell$ from his interval. We note that the argument for $\mathcal I^{i\downarrow x_r}$ is symmetric. Now, if $f(\mathcal I)\neq x_\ell$, strong uncompromisingness implies that $f(\mathcal I^{i\downarrow x_\ell})=f(\mathcal I)$ and robustness holds. In more detail, if $f(\mathcal I)\neq x_\ell$, then either \emph{(i)} $f(\mathcal I)\rhd x_\ell$, \emph{(ii)} $x_\ell \rhd f(\mathcal I)\rhd x_r$, \emph{(iii)} $x_r\rhd f(\mathcal I)$, or  \emph{(iv)} $x_\ell\rhd f(\mathcal I)=x_r$. In all four cases, the premises of strong uncompromisingness are satisfied as voter $i$ only removes $x_\ell$ from his interval, so this axiom directly implies our claim. Next, assume for contradiction that $f(\mathcal I)=x_\ell$ and $f(\mathcal I^{i\downarrow x_\ell})\not\in \{x_\ell, x_{\ell+1}\}$. In this case, it either holds that  \emph{(i)} $f(\mathcal I^{i\downarrow x_\ell})\rhd x_{\ell}$,  \emph{(ii)} $x_{\ell+1}\rhd f(\mathcal I^{i\downarrow x_\ell})\rhd x_r$,  \emph{(iii)} $x_{\ell+1}\rhd f(\mathcal I^{I\downarrow x_\ell})=x_r$, or \emph{(iv)} $x_r\rhd f(\mathcal I^{i\downarrow x_\ell})$. In all of these cases, strong uncompromisingness from $\mathcal I^{i\downarrow x_\ell}$ to $\mathcal I$ implies that $f(\mathcal I)=f(\mathcal I^{i\downarrow x_\ell})$. However, this contradicts that $f(\mathcal I)=x_\ell$ and $f(\mathcal I^{i\downarrow x_\ell})\not\in\{x_\ell, x_{\ell+1}\}$. Hence, if $f(\mathcal I)=x_\ell$, then $f(\mathcal I^{i\downarrow x_\ell})\in\{x_\ell, x_{\ell+1}\}$ and $f$ is robust. 

    \paragraph{Robustness implies strong uncompromisingness.} For the other direction, we suppose that $f$ is robust and show that $f$ is strongly uncompromising. To this end, fix a second profile $\mathcal I'$ that differs from $\mathcal I$ only in the interval of voter $i$ and let $I_i'=[x_\ell', x_r']$. We proceed with a case distinction on $f(\mathcal I)$. First, suppose that $f(\mathcal I)\rhd x_\ell$. In this case, strong uncompromisingness only applies if $f(\mathcal I)\ssucceq x_\ell'$, so we assume this for $I_i'$. We can transform $I_i$ to $I_i'$ as follows: if $f(\mathcal I)\not\in I_i'$, we first expand $I_i$ to contain all candidates right of $f(\mathcal I)$ and subsequently remove all candidates not in $I_i'$. By first adding and then removing alternatives one after another and applying robustness along the way, it follows that the outcome is not allowed to change at any step, so $f(\mathcal I)=f(\mathcal I')$. In more detail, we note that no modification involves the current winner $f(\mathcal I)$, so robustness does not allow the outcome to change. On the other hand, if $f(\mathcal I)\in I_i'$, then $f(\mathcal I)=x_\ell'$ and we can expand voter $i$'s interval first to $[f(\mathcal I), x_m]$ and then remove alternatives from the right to derive $I_i'$. Again adding and removing alternatives one at time shows that the outcome cannot change due to robustness, because we do not delete $x_\ell'$ or add an alternative left of it.
    Hence, strong uncompromisingness is satisfied in the case, and a symmetric analysis applies when $x_r\rhd f(\mathcal I)$. 

    Next, suppose that $x_\ell=f(\mathcal I)\rhd x_r$. In this case, strong uncompromisingness only applies when $x_\ell=x_\ell'$. Hence, we can immediately transform $I_i$ to $I_i'$ by adding alternatives right of $x_r$ (if $x_r\rhd x_r'$) or deleting alternatives left of $x_r$ (if $x_r'\rhd x_r$). By sequentially applying these modifications, we infer from robustness that the outcome is not allowed to change, as voter $i$ never removes the current winner from his interval or adds alternatives directly next to it with respect to $\rhd$. The case that $x_\ell\rhd f(\mathcal I)= x_r$ is again symmetric.

    As the last case, we suppose that $x_\ell \rhd f(\mathcal I)\rhd x_r$. In this case, we can assume that $x_\ell'\ssucceq f(\mathcal I)\ssucceq x_r'$ since strong uncompromisingness is otherwise void. In this case, we can first delete the alternatives from $I_i\setminus I_i'$ from $I_i$ and then add those in $I_i'\setminus I_i$ to transform $I_i$ to $I_i'$. Moreover, by applying these modifications one after another, it follows again that the outcome is not allowed to change. To make this more precise, let's assume that $x_\ell'\rhd x_\ell$ and $x_r'\rhd x_r$. Then, we can first one after another delete the alternatives $x_t,\dots, x_r$ (where $x_t$ is the alternative directly right of $x_r'$) from $I_i$. Since none of these alternatives is the current winner, robustness implies that the winner remains unchanged. Next, we can one after another add the alternatives $x_{\ell+1},\dots, x_{\ell}'$ to $I_i$. Since the current winner is strictly left of $x_\ell$, robustness implies again that the outcome is not allowed to change. We note that there are three more cases (depending on whether $x_\ell'\rhd x_\ell$ or $x_\ell\rhd x_\ell'$ and $x_r'\rhd x_r$ or $x_r\rhd x_r'$), all of which follow from symmetric arguments. Hence, it holds that $f(\mathcal I)=f(\mathcal I')$ and strong uncompromisingness is satisfied. 
    
    \paragraph{Robustness implies strategyproofness.} Finally, we assume that $f$ is robust and show that it satisfies strategyproofness. Assume for contradiction that this is not the case, which means that there is a voter $i$, two $i$-variants $\mathcal I$ and $\mathcal I'$, and a weakly single-peaked preference relation $\succsim_i$ such that $f(\mathcal I)\not\succsim_i f(\mathcal I)$ and $I_i=T(\succsim_i)$. This implies that $f(\mathcal I)\not\in T(\succsim_i)$ because  $x\succsim_i y$ for all $x\in T(\succsim_i)$ and $y\in A$. Let $I_i=T(\succsim_i)=[x_\ell, x_r]$ and $I_i'=[x_\ell',x_r']$. Further, we assume without loss of generality that $f(\mathcal I)\rhd x_\ell$ as the case that $x_r\rhd f(\mathcal I)$ is symmetric. If also $f(\mathcal I)\rhd x_\ell'$, it holds that $f(\mathcal I)=f(\mathcal I')$ because robustness implies strong uncompromisingness. Hence, we may assume that $x_\ell'\rhd f(\mathcal I)$. In this case, we first consider the intermediate profile $\mathcal I^1$ where voter $i$ reports $[f(\mathcal I),x_r]$. Strong uncompromisingness and thus also robustness shows that $f(\mathcal I)=f(\mathcal I^1)$. Next, we expand voter $i$'s interval to $[x_\ell', x_r]$ by adding further alternatives to the left, resulting in the profile $\mathcal I^2$. By robustness, the winning alternative can only move to the left, i.e., it holds that $f(\mathcal I^2)\ssucceq f(\mathcal I^1)$. Lastly, if $x_r'\rhd x_r$, we delete the corresponding alternatives from voter $i$'s interval to derive $I_i'$ and robustness shows again that the winning alternative can only move to the left. On the other hand, if $x_r\ssucceq x_r'$, we add the alternatives in $[x_{r+1}, x_r']$ one after another to voter $i$'s interval. Since $f(\mathcal I^2)\rhd f(\mathcal I)\rhd x_\ell\ssucceq x_r$, robustness again shows that this step cannot affect the outcome. We hence conclude that $f(\mathcal I')\ssucceq f(\mathcal I)\ssucc x$ for all $x\in T(\succsim_i)$. Lastly, since $\succsim_i$ is weakly single-peaked, this implies that $f(\mathcal I)\succsim f(\mathcal I')$, which contradicts our assumption.
    \end{proof}

\subsection{Further axioms}

In this section, we introduce three standard axioms, namely unanimity, anonymity, reinforcement, as well as a variant of a continuity axiom adopted to our setting. Variants of these axioms feature prominently in the analysis of various types of scoring rules \citep[e.g.,][]{Smit73a,Youn75a,YoLe78a,Fish78d,Bran13a,LaSk21a,BrPe19a,Lede23a}.

\paragraph{Unanimity.} A basic requirement of voting rules is that, if all voters agree on a favorite alternative, this alternative must be selected. This is formalized by \emph{unanimity}, which requires of a voting rule $f$ on $\dom$ that $f(\mathcal{I})=x_j$ for all interval profiles $\mathcal I\in\dom$ and alternatives $x_j\in A$ such that $I_i=\{x_j\}$ for all $i\in N_{\mathcal I}$.

\paragraph{Anonymity.} Anonymity postulates that the selected outcome is invariant under renaming the voters. More formally, a voting rule $f$ is \emph{anonymous} if $f(\mathcal{I})=f(\tau(\mathcal{I}))$ for all interval profiles $\mathcal{I}\in\dom$ and bijections $\tau:\mathbb{N}\rightarrow\mathbb{N}$. Here, $\mathcal{I}'=\tau(\mathcal{I})$ denotes the profile defined by $N_{\mathcal{I}'}=\{\tau(i)\colon i\in N_{\mathcal{I}}\}$ and $I_{\tau(i)}'=I_i$ for all $i\in N_{\mathcal{I}}$.

\paragraph{Reinforcement.} The idea of reinforcement is that, if the same alternative is chosen for two disjoint elections, this alternative should also be chosen for a combined election. The reason for this is that if an alternative is the ``best'' outcome for two disjoint profiles, it should also be the ``best'' outcome for the combined profile. To formalize this, we define by $\mathcal{I}''=\mathcal{I}+\mathcal{I}'$ the profile derived from two voter-disjoint profiles $\mathcal{I},\mathcal{I}'\in\dom$ by setting $N_{\mathcal{I}''}=N_{\mathcal{I}}\cup N_{\mathcal{I}'}$, $I_i''=I_i$ for all $i\in N_{\mathcal{I}}$, and $I''_i=I_i'$ for all $i\in N_{\mathcal{I}'}$. Now, a voting rule $f$ is \emph{reinforcing} if $f(\mathcal{I}+\mathcal{I}')=f(\mathcal{I})$ for all interval profiles $\mathcal{I},\mathcal{I}'\in\dom$ such that $N_{\mathcal{I}}\cap N_{\mathcal{I}'}=\emptyset$ and $f(\mathcal{I})=f(\mathcal{I}')$. 

\paragraph{Right-biased continuity.} As our last axiom, we will introduce a variant of an axiom known as continuity or overwhelming-majority property \citep{Youn75a,Myer95b}. The rough idea of this axiom is that, if we are given two profiles, we can marginalize the effect of one profile on the outcome by sufficiently often cloning the other profile. More formally, continuity typically requires that, for all profiles $\mathcal{I},\mathcal{I}'\in\dom$, there is $\lambda\in\mathbb{N}$ such that $f(\lambda \mathcal{I}+\mathcal{I}')\subseteq f(\mathcal{I})$ (where $\lambda \mathcal{I}$ denotes the profile that consists of $\lambda$ copies of~$\mathcal{I}$). 
However, this formulation is incompatible with our model because resolute voting rules require tie-breaking, which is not continuous. We will thus weaken this axiom and, moreover, use it to specify the tie-breaking used in our voting rules. In more detail, we say that a voting rule $f$ on $\dom$ satisfies \emph{right-biased continuity} if it holds for all profiles $\mathcal{I},\mathcal{I}'\in\dom$ with $N_{\mathcal{I}}\cap N_{\mathcal{I}'}=\emptyset$ that \emph{(i)} if $f(\mathcal{I}')\ssucceq f(\mathcal{I})$, there is $\lambda\in\mathbb{N}$ such that $f(\lambda \mathcal{I}+\mathcal{I}')=f(\mathcal I)$ and \emph{(ii)} if $f(\mathcal{I})\ssucc f(\mathcal{I}')$, there is an alternative $x_j\in \bigcup_{i\in N_{\mathcal{I}}} I_i$ and $\lambda\in\mathbb{N}$ such that $f(\mathcal{I})\ssucceq f(\lambda \mathcal{I}+\mathcal{I}')\ssucceq x_j$. Less formally, right-biased continuity ensures full continuity if $f(\mathcal I)$ is right of $f(\mathcal I')$ and otherwise only guarantees that we cannot completely ignore $\mathcal{I}$ when duplicating this profile sufficiently often. We note, however, that the second condition becomes close to full continuity when the set $\bigcup_{i\in N_{\mathcal{I}}} I_i$ is small. 

\subsection{Position-Threshold Rules}\label{subsec:singleRobust}

We will now introduce the central class of voting rules in this paper, which we call position-threshold rules. Since these rules generalize Moulin's phantom median rules (\citeyear{Moul80a}), we will first discuss these rules.

\paragraph{Phantom median rules.} Phantom median rules are typically defined on the domain of single-peaked preference profiles for a fixed electorate $N$ and work as follows: given a single-peaked profile $R\in\mathcal{P}_{\ssucc}^N$, we first add $|N|-1$ phantom voters with fixed single-peaked preference relations, then order all $2|N|-1$ voters with respect to their favorite alternatives according to $\ssucc$, and finally return the favorite alternative of the $|N|$-th voter in this list. To make this more formal, let $p_i$ denote the number of phantom voters that top-rank alternative $x_i$ and let $q_i(R)$ be the number of regular voters that top-rank alternative $x_i$ in the profile~$R$. Then, a phantom median rule chooses for every profile $R$ the alternative $x_i$ with minimal index $i$ such that $\sum_{j=1}^i p_j + q_j(R)\geq |N|$. 

We will next reformulate this definition and hence introduce the \emph{(individual) peak position function $\pi_\mathit{SP}({\succsim_i}, x_j)$}. This function states the relative position of a voter $i$ with respect to each alternative $x_j$: $\pi_\mathit{SP}({\succsim_i}, x_j)=1$ if voter $i$'s favorite alternative $x_i$ is weakly left of $x_j$ (i.e., $x_i\ssucceq x_j)$ and $\pi_\mathit{SP}(\succsim_i,x_j)=0$ otherwise. Moreover, we define the \emph{(collective) peak position function $\Pi_\mathit{SP}$} of an alternative $x_j$ in a single-peaked profile $R$ by $\Pi_\mathit{SP}(R, x_j)=\sum_{i\in N_R} \pi_\mathit{SP}{(\succsim_i}, x_j)$, where $N_R$ denotes the set of voters that report a preference relation in $R$. 
Put differently, the collective peak position function $\Pi_\mathit{SP}$ counts how many voters in $R$ report a favorite alternative that is weakly left of $x_j$. 
Next, we define by $\max_{\ssucc} X$ the left-most alternative $x$ in a given set $X$, i.e., $x=\max_{\ssucc} X$ satisfies that $x\in X$ and $x\ssucc y$ for all $y\in X\setminus \{x\}$. 
Then, a voting rule $f$ is a phantom median rule if there are integers $p_1,\dots, p_m\in\mathbb{N}\cup\{0\}$ such that $\sum_{j=1}^m p_j=|N|-1$ and $f(R)=\max_{\ssucc} \{x_i\in A\colon \Pi_\mathit{SP}(R, x_i)\geq |N|-\sum_{j=1}^i p_j\}$ for all profiles $R\in\mathcal{P}_{\ssucc}^N$. 

Finally, to extend the definition of phantom median rules from a fixed electorate to all feasible electorates, we will replace the values $p_i$ with a \emph{threshold vector} $\theta$, which is a vector in $(0,1)^m$ such that $\theta_1\geq\theta_2\geq \dots\geq\theta_m$. The intuition is that for every electorate $N$, the value $\theta_i\cdot |N|$ is equal to $|N|-\sum_{j=1}^i p_j$, i.e., $\theta_i$ states the fraction of the voters that need to report an alternative left of $x_i$ to make $x_i$ an eligible outcome.
We are now ready to state our final definition of phantom median rules.

\begin{definition}[Phantom median rules]\label{def:PMR}
A voting rule $f$ on $\mathcal{P}_{\ssucc}^*$ is a \emph{phantom median rule} if there is a threshold vector $\theta\in(0,1)^m$ with $\theta_1\geq\dots\geq\theta_m$ such that $f(R)=\max_{\ssucc} \{x_i\in A\colon \Pi_\mathit{SP}(R, x_i)\geq \theta_i|N_R|\}$ for all profiles $R\in\mathcal{P}^*_{\ssucc}$.    
\end{definition}

In this definition, the value $\theta_m$ is irrelevant since $\Pi_\mathit{SP}(R,x_m)=|N_R|$ for all single-peaked profiles $R$. Moreover, the restriction that $\theta_i>0$ for $i\in \{1,\dots, m-1\}$ is necessary to ensure unanimity as $\Pi_\mathit{SP}(R,x_i)\geq 0$ for all $i\in \{1,\dots, m\}$, and the condition that $\theta_i<1$ for all $i\in \{1,\dots, m-1\}$ is necessary to satisfy right-biased continuity. 

\paragraph{Position-threshold rules.} Position-threshold rules aim to extend phantom median rules to the interval domain by generalizing the peak position function $\pi_\mathit{SP}$ to intervals. The central problem for this is that a voter's position with respect to some alternatives is unclear if his interval contains more than one alternative. For instance, if a voter reports $[x_1, x_2, x_3]$, his relative position with respect to $x_1$ and $x_2$ is ambiguous. In this paper, we will solve this issue by using a \emph{weight vector} $\alpha\in [0,1]^m$ which quantifies the relative position of every voter with respect to the alternatives in his interval. In more detail, if an alternative $x_k$ is in the interval $I_i$ of voter $i$ (but it is not the right-most alternative in~$I_i$), the relative position of voter $i$ with respect to $x_k$ is $\alpha_k$. By contrast, just as for the peak position function, a voter's relative position with respect to $x_k$ is $1$ if every alternative in his interval is weakly left of $x_k$ and $0$ if every alternative in his interval is strictly right of $x_k$. We formalize this idea with \emph{individual position functions $\pi_\alpha:\Lambda\times A\rightarrow \mathbb{R}$}, which depend on a weight vector $\alpha\in [0,1]^m$ and are defined as follows:
\begin{align*}
    \pi_\alpha([x_i, x_j], x_k)=\begin{cases}
        0 \qquad\qquad\qquad &\text{ if } x_k\ssucc x_i\\
        \alpha_k\qquad\qquad\qquad &\text{ if } x_i\ssucceq x_k \ssucc_i x_j\\
        1 \qquad\qquad\qquad & \text{ if } x_j\ssucceq x_k.
    \end{cases}
\end{align*}

We note that every individual position function $\pi_\alpha$ generalizes the individual peak position function $\pi_\mathit{SP}$ to the interval domain because $\pi_{SP}({\succsim}, x_j)=\pi_\alpha(T(\succsim), x_j)$ for all single-peaked preference relations ${\succsim}\in \mathcal{P}_{\ssucc}$, alternatives $x_j\in A$, and weight vectors $\alpha$. 

Given a weight vector $\alpha$, we define the \emph{(collective) position function $\Pi_\alpha$} by $\Pi_\alpha(\mathcal I,x_j)=\sum_{i\in N_{\mathcal I}} \pi_\alpha(I_i,x_j)$ for all alternatives $x_j$ and interval profiles $\mathcal I$. Based on a weight vector $\alpha$ and a threshold vector $\theta$, we can now define position-threshold rules: these rules choose the alternative with the smallest index whose collective position exceeds its threshold. However, while all such rules are well-defined, we need to impose additional constraints on the weight vector $\alpha$ to guarantee robustness. Specifically, we say that a weight vector $\alpha$ is \emph{compatible} with a threshold vector $\theta$ if, for all alternatives $x_i\in A\setminus \{x_m\}$ and profiles $\mathcal I\in\dom$, it holds that $\Pi_\alpha(\mathcal I, x_{i})\geq \theta_{i} n_{\mathcal I}$ implies $\Pi_\alpha(\mathcal I, x_{i+1})\geq \theta_{i+1} n_{\mathcal I}$. Put differently, this condition states that, if the collective position of an alternative $x_i$ exceeds its threshold value, the same should be true for every alternative right of $x_i$. We note that this condition again mimics the behavior of phantom median rules because the collective peak position function always satisfies this condition. Moreover, we give a characterization of compatible weight and threshold vectors in \Cref{lem:robust}, which, e.g., shows that a weight vector $\alpha$ is compatible with every threshold vector if $\alpha_1\leq \alpha_2\leq \dots\leq\alpha_m$. These are arguably the most natural weight vectors because $\alpha_{i+1}\geq\alpha_{i}$ intuitively captures that a voter is at least as much left of $x_{i+1}$ as of $x_{i}$.

Finally, the class of position-threshold rules consists precisely of the rules induced by a pair of compatible weight and threshold vectors.

\begin{definition}[Position-threshold rules]
    A voting rule $f$ on $\dom$ is a \emph{position-threshold rule} if there are a threshold vector $\theta\in (0,1)^m$ with $\theta_1\geq\dots\geq\theta_m$ and a weight vector $\alpha\in[0,1]^m$ such that $f(\mathcal I)=\max_{\ssucc}\{x_i\in A\colon \Pi_{\alpha}(\mathcal I, x_i)\geq \theta_i n_{\mathcal I}\}$ for all profiles $\mathcal I\in\dom$. 
\end{definition}

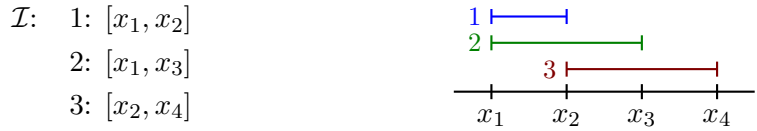
\begin{figure}
    \centering
    \begin{tabular}{lll}
       $\mathcal I$:   & $1$: $[x_1,x_2]$\\
         & $2$: $[x_1,x_3]$ \\
         & $3$: $[x_2,x_4]$ 
    \end{tabular}
    \hspace{3cm}
    \begin{tikzpicture}[baseline=0.3cm]
        \draw[thick] (0.5,0) -- (4.5,0); % Extended line to indicate an axis with an arrow

    \node[left] at (1,1) {\small\color{blue} $1$};
    \draw[blue, thick] (1,1) -- (2, 1);
    \draw[blue, thick] (1,0.9) -- (1, 1.1);
    \draw[blue, thick] (2,0.9) -- (2, 1.1);

    \node[left] at (1,0.65) {\small\color{green!50!black} $2$};
    \draw[green!50!black, thick] (1,0.65) -- (3, 0.65);
    \draw[green!50!black, thick] (1,0.55) -- (1, 0.75);
    \draw[green!50!black, thick] (3,0.55) -- (3, 0.75);

    \node[left] at (2,0.3) {\small\color{red!50!black} $3$};
    \draw[red!50!black, thick] (2,0.3) -- (4, 0.3);
    \draw[red!50!black, thick] (2,0.2) -- (2, 0.4);
    \draw[red!50!black, thick] (4,0.2) -- (4, 0.4);
    
    % Draw ticks and label points
    \foreach \x in {1,2,3,4} {
        \draw[thick] (\x,0.1) -- (\x,-0.1); % Small vertical ticks
        \node[below] at (\x,-0.1) {$x_{\x}$}; % Label points
    }
    \end{tikzpicture}
    \caption{Example of position-threshold rules. The interval profile $\mathcal I$ contains $3$ voters and $4$ alternatives, and it is graphically represented on the right. First, let $f_1$ denote the position-threshold rule given by the threshold vector $\theta=(\frac{1}{2}, \frac{1}{2}, \frac{1}{2}, \frac{1}{2})$ and the weight vector $\alpha=(1,1,1,1)$. This rule selects the median of the voters' left endpoints and it holds that $f_1(\mathcal I)=x_1$ since $\Pi_{\alpha}(\mathcal I, x_1)=2>\frac{1}{2} n_{\mathcal I}$. By contrast, for the position-threshold rule $f_2$ defined by the same threshold vector $\theta$ and the weight vector $\beta=(\frac{1}{2}, \frac{1}{2}, \frac{1}{2}, \frac{1}{2})$, it holds that $\Pi_{\beta}(\mathcal I, x_1)=1<\frac{1}{2} n_\mathcal{I}$ and $\Pi_{\beta}(\mathcal I, x_2)=2\geq \frac{1}{2} n_\mathcal{I}$, so $f_2(\mathcal I)=x_2$. 
    }
    \label{fig:example}
\end{figure}

To provide further intuition for these rules, we discuss the roles of the threshold vector $\theta$ and the weight vector $\alpha$ in more detail. Moreover, we display in \Cref{fig:example} an example illustrating how position-threshold rules work. First, the threshold vector $\theta$ can be interpreted just as for phantom median rules: for every $i\in\{1,\dots, m-1\}$, $\theta_i$ states the fraction of the phantom voters that report an alternative right of $x_i$. For instance, if $\theta_i=\frac{3}{4}$ for all $i\in \{1,\dots, m\}$, we may imagine that $\frac{1}{4}n_{\mathcal I}$ phantom voters report $\{x_1\}$ and $\frac{3}{4}n_{\mathcal I}$ phantom voters report $\{x_m\}$. As a second example, if $\theta_i=\frac{m-i}{m}$ for all $i\in \{1,\dots, m-1\}$ and $\theta_m=\frac{1}{m}$, each alternative $x_i$ is reported by a $\frac{1}{m}$ share of the phantom voters. We note that, even though our voters are allowed to report intervals, it suffices that the phantom voters report single alternatives. 

Next, the weight vector can be interpreted as a way of decomposing intervals into singleton ballots: if a voter reports an interval $[x_\ell,x_r]$, we may replace this voter with $\alpha_\ell$ voters reporting $\{x_\ell\}$, $\alpha_{i}-\alpha_{i-1}$ voters reporting $\{x_i\}$ for all $x_i\in [x_{\ell+1},x_{r-1}]$, and $1-\alpha_r$ voters reporting $\{x_r\}$. 
It can be shown that this transformation preserves the collective position of every alternative, so it does not affect the outcome of a position-threshold rule. For instance, for the weight vector $\alpha=(1,\dots, 1)$, this means that every interval $I$ is represented by a voter who only reports the left-most alternative of $I$. As a second example, if $\alpha=(\frac{1}{2},\dots,\frac{1}{2}$), every interval $[x_i, x_j]$ is represented by $\frac{1}{2}$ voters reporting $\{x_i\}$ and $\frac{1}{2}$ voters reporting~$\{x_j\}$. 
Since position-threshold rules reduce to phantom median rules when all voters report singleton intervals, this gives a direct relation between position-threshold rules and phantom median rules.

\section{Results}

We are now ready to state our results: we discuss our characterization of position-threshold rules in \Cref{subsec:singlewinner}, and we present in \Cref{subsec:endponintmedian} a characterization of a particular position-threshold rule called the endpoint-median rule.

\subsection{A Characterization of Position-Threshold Rules}\label{subsec:singlewinner}

In this section, we discuss our main result, the characterization of position-threshold rules. In more detail, we will prove that position-threshold rules are the only voting rules on $\dom$ that satisfy robustness, anonymity, unanimity, reinforcement, and right-biased continuity. We note that it is an easy consequence of Moulin's work (\citeyear{Moul80a}) that phantom median rules, as defined in \Cref{def:PMR}, are the only single-winner voting rules on the domain of single-peaked preference relations that satisfy the given axioms when suitably adapting robustness. Hence, our characterization can be seen as a variant of Moulin's result for the interval domain. Since the proof of the subsequent theorem is lengthy, we will only provide a proof sketch in the main body and defer the full proof to the appendix.

\begin{restatable}{theorem}{robustCharacterization}\label{thm:characterization}
    A voting rule on $\dom$ is robust, anonymous, unanimous, reinforcing, and right-biased continuous if and only if it is a position-threshold rule. 
\end{restatable}
\begin{proof}[Proof Sketch]
    $(\impliedby)$
    For the direction from right to left, we fix a position-threshold rule $f$ and let $\theta$ and $\alpha$ denote its threshold and weight vectors. Note that, by definition, $\theta$ and $\alpha$ are compatible. First, it follows immediately from the definition of position-threshold rules that $f$ is anonymous and unanimous. 
    Next, to show that $f$ is reinforcing, let $\mathcal{I}$ and $\mathcal{I'}$ denote two voter-disjoint profiles such that $f(\mathcal I)=f(\mathcal I')=x_i$ for some alternatives $x_i\in A$. This means for $\mathcal{\hat I}\in \{\mathcal I, \mathcal I'\}$ that $\Pi_{\alpha}(\mathcal{\hat I}, x_i)\geq \theta_i|N_{\mathcal{\hat I}}|$ and $\Pi_{\alpha}(\mathcal{\hat I}, x_h)< \theta_h|N_{\mathcal{\hat I}}|$ for all $x_h\in A$ with $x_h\ssucc x_i$. In turn, it follows that $\Pi_\alpha(\mathcal I+\mathcal I',x_i)=\Pi_\alpha(\mathcal I, x_i)+\Pi_\alpha(\mathcal I',x_i)\geq \theta_i n_{\mathcal I}+\theta_i n_{\mathcal I'}=\theta_i n_{\mathcal I+\mathcal I'}$ and $\Pi_\alpha(\mathcal I+\mathcal I',x_h)=\Pi_\alpha(\mathcal I, x_h)+\Pi_\alpha(\mathcal I',x_h)<\theta_h n_{\mathcal I+\mathcal I'}$ for all $x_h\in A$ with $x_h\ssucc x_i$. This shows that $f(\mathcal I+\mathcal I')=x_i$, so reinforcement is satisfied.

    For robustness and right-biased continuity, we recall that $\Pi_{\alpha}(\mathcal I, x_i)\geq \theta_in_\mathcal{I}$ implies that $\Pi_{\alpha}(\mathcal I, x_{i+1})\geq \theta_{i+1}n_\mathcal{I}$ for all profiles $\mathcal I\in\dom$ and alternatives $x_i\in A\setminus \{x_m\}$ since $\alpha$ and $\theta$ are compatible. Now, to prove that $f$ is robust, let $\mathcal I$ be an interval profile, $i\in N_{\mathcal I}$ a voter, and $x_\ell, x_r\in A$ the two alternatives such that $x_\ell\ssucc x_r$ and $I_i=[x_\ell, x_r]$. We will show that $f(\mathcal I)=f(\mathcal I^{i\downarrow x_\ell})$ or $f(\mathcal I)=x_{\ell}$ and $f(\mathcal I^{i\downarrow x_\ell})=x_{\ell+1}$ and observe that similar arguments also apply for $\mathcal I^{i\downarrow x_r}$. 
    By the definition of $\Pi_\alpha$, it holds that $\Pi_{\alpha}(\mathcal I, x_j)=\Pi_{\alpha}(\mathcal{I}^{i\downarrow x_\ell}, x_j)$ for all alternatives $x_j\in A\setminus \{x_\ell\}$ and $\Pi_{\alpha}(\mathcal I, x_\ell)\geq \Pi_{\alpha}(\mathcal{I}^{i\downarrow x_\ell}, x_\ell)$. If $f(\mathcal I)\neq x_\ell$, this implies that $f(\mathcal I)=f(\mathcal{I}^{i\downarrow x_\ell})$. By contrast, if $f(\mathcal I)=x_\ell$, it follows that $\Pi_{\alpha}(\mathcal I^{i\downarrow x_\ell}, x_j)=\Pi_{\alpha}(\mathcal I, x_j)<\theta_jn_\mathcal{I}$ for all alternatives $x_j$ with $x_j\ssucc x_\ell$ and the compatibility of $\alpha$ and $\theta$ implies that $\Pi_{\alpha}(\mathcal I^{i\downarrow x_\ell}, x_{\ell+1})=\Pi_{\alpha}(\mathcal I, x_{\ell+1})\geq\theta_{\ell+1}n_\mathcal{I}$ as $\Pi_{\alpha}(\mathcal I, x_{\ell})\geq\theta_{\ell}n_\mathcal{I}$. In combination, this proves that $f(\mathcal{I}^{i\downarrow x_\ell})\in \{x_\ell, x_{\ell+1}\}$ and $f$ thus is robust.

    Finally, to prove that $f$ satisfies right-biased continuity, we consider two profiles $\mathcal I$ and $\mathcal I'$ and we make a case distinction with respect to the relative position of $x_i=f(\mathcal I)$ and $x_j=f(\mathcal I')$. For instance, if $x_j\ssucc x_i$, we infer that $\Pi_\alpha(\mathcal I, x_h)<\theta_h n_\mathcal{I}$ for all $x_h$ with $x_h\ssucc x_i$. By copying $\mathcal I$ sufficiently often, we can make the absolute difference in this inequality arbitrarily big, so there is $\lambda\in\mathbb{N}$ such that $\Pi_\alpha(\lambda \mathcal I+\mathcal{I}', x_h)<\theta_h |N_{\lambda \mathcal I+\mathcal{I}'}|$ for all $x_h$ with $x_h\ssucc x_i$. Furthermore, it holds that $\Pi_\alpha(\mathcal I, x_i)\geq \theta_in_\mathcal{I}$ as $f(\mathcal I)=x_i$ and that $\Pi_\alpha(\mathcal I', x_i)\geq \theta_i n_{\mathcal I'}$ as $f(\mathcal I')=x_j$ and $x_j\ssucc x_i$. This implies that $\Pi_\alpha(\lambda \mathcal I+\mathcal I', x_i)\geq \theta_i  n_{\lambda \mathcal I+\mathcal{I}'}$, so $f(\lambda \mathcal I+\mathcal I')=x_i$ and right-biased continuity is satisfied in this case. 
    \medskip

    $(\implies)$ For the direction from left to right, we assume that $f$ is a voting rule on $\dom$ that satisfies anonymity, unanimity, robustness, reinforcement, and right-biased continuity. As the first step, we investigate $f$ on the subset $\mathcal{D}_1^*$ of $\dom$ where voters only report a single alternative. This domain is related to the domain of single-peaked preferences by associating the reported alternative with the top-ranked alternative of a single-peaked preference relation. We hence show that $f$ induces a voting rule $f'$ on the domain of single-peaked preference profiles that satisfies anonymity, unanimity, and strategyproofness. In turn, the characterization of \citet{Moul80a} shows that $f'$ is a phantom median rule. By using the connection between $f$ and $f'$, we derive that there is a threshold vector $\theta\in (0,1)^m$ such that $\theta_1\geq \theta_2\geq \dots\geq \theta_m$ and $f(\mathcal I)=\max_{\ssucc}\{x_i\in A\colon \Pi_\alpha(\mathcal I, x_i)\geq \theta_in_\mathcal{I}\}$ for all profiles $\mathcal I\in\mathcal{D}_1^*$. (Note that the choice of the weight vector $\alpha$ does not matter here as all voters only report a single alternative for profiles in $\mathcal{D}_1^*$.)

   Next, we will apply the geometric techniques originally developed by \citet{Youn75a} in the context of scoring rules. To this end, we define $q=|\Lambda|$ as the number of intervals over $\ssucc$ and assume that the intervals in $\Lambda$ are enumerated in an arbitrary order $I_1,\dots, I_q$. Based on this enumeration, we can represent each interval profile $\mathcal I$ by a vector $v\in\mathbb{N}^q_0\setminus \{0\}$, where the entry $v_i$ states how often the interval $I_i$ is reported in $\mathcal I$. Moreover, since $f$ is anonymous, it can be computed based on the vectors $v$, i.e., we may interpret $f$ as function from $\mathbb{N}^{q}\setminus \{0\}$ to $A$. Next, we show based on reinforcement that $f$ can be extended to a function $\hat g:\mathbb{Q}^{q}_{\geq 0}\setminus \{0\}\rightarrow A$ while preserving its desirable properties. We then define the sets $Q_i=\{v\in\mathbb{Q}^{q}_{\geq 0}\setminus \{0\}\colon \hat g(v)=x_i\}$ and note that these sets are $\mathbb{Q}$-convex (i.e., for all $v,v'\in Q_i$ and $\lambda\in [0,1]\cap \mathbb{Q}$, it holds that $\lambda v+(1-\lambda)v'\in Q_i$) as $\hat g$ is reinforcing. This implies that the closure of $Q_i$ with respect to $\mathbb{R}^{q}$, denoted by $\bar Q_i$, is convex. Further, we show that the interiors of these sets are non-empty and pairwise disjoint. The separating hyperplane theorem for convex sets hence implies that, for all distinct $x_i, x_j\in A$, there is a non-zero vector $u^{i,j}\in\mathbb{R}$ such that $vu^{i,j}\geq 0$ for all $v\in\bar Q_i$ and $vu^{i,j}\leq 0$ for all $v\in \bar Q_j$ ($vu^{i,j}$ denotes here the standard scalar product between vectors). Moreover, we prove that the vectors $u^{i,j}$ fully describe the sets $\bar Q_i$ because $\bar Q_i=\{v\in \mathbb{R}^{q}_{\geq 0}\colon\forall x_j\in A\setminus \{x_i\}\colon vu^{i,j}\geq 0\}$ for all $x_i\in A$. 

    As the next step, we investigate the vectors $u^{i,j}$ in more detail. In particular, we will prove that $vu^{i,i+1}\geq 0$ implies $vu^{i+1, i+2}\geq 0$ for all $i\in \{1,\dots, m-2\}$ and $v\in\mathbb{R}^q_{\geq 0}$. Based on this insight, we derive a simplified representation of the sets $\bar Q_i$: it holds that $\bar Q_1=\{v\in \mathbb{R}^q\colon vu^{1,2}\geq 0\}$, $\bar Q_i=\{v\in \mathbb{R}^q\colon vu^{i-1,i}\leq 0\land vu^{i,i+1}\geq 0\}$ for all $i\in \{2,\dots, m-1\}$, and $\bar Q_m=\{v\in \mathbb{R}^q\colon vu^{m-1,m}\leq 0\}$. Thus, it suffices to study the vectors $u^{i,i+1}$ for all $i\in \{1,\dots, m-1\}$. To do so, we denote by $u^{i,j}_X$ the entry in $u^{i,j}$ that corresponds to the interval $X$ and we recall that $\theta$ is the threshold vector of $f$ for the domain $\mathcal{D}_1^*$. We show that we can scale each vector $u^{i,i+1}$ such that \emph{(i)} $u^{i,i+1}_X=1-\theta_i$ for all $X\in\Lambda$ with $X\subseteq [x_1,x_i]$ \emph{(ii)} $u^{i, i+1}_X=-\theta_i$ for all $X\in\Lambda$ with $X\subseteq [x_{i+1}, x_m]$, and \emph{(iii)} $u^{i,i+1}_X=u^{i,i+1}_{\{x_i,x_{i+1}\}}$ for all $X\in \Lambda$ with $\{x_i, x_{i+1}\}\subseteq X$. We next derive a weight vector $\alpha$ from the vectors $u^{i,i+1}$ by defining $\alpha_i=u^{i, i+1}_{\{x_i,x_{i+1}\}}+\theta_i$ for all $i\in \{i,\dots, m-1\}$ and $\alpha_m=\alpha_{m-1}$. Moreover, we define the individual position function $\pi_\alpha$ based on this weight vector and we show that $\sum_{X\in\Lambda} v_X \pi_\alpha(X,x_i)=vu^{i,j}+\theta_i\sum_{X\in \Lambda} v_X$. By combining our insights, we thus conclude that $\bar Q_i=\{v\in \mathbb{R}^q\colon \sum_{X\in\Lambda} v_X \pi_\alpha(X, x_{i-1})\leq \theta_{i-1}\sum_{X\in\Lambda} v_X \land \sum_{X\in\Lambda} v_X\pi_\alpha(X, x_i)\geq \theta_i\sum_{X\in \Lambda} v_X\}$. 
    
    Now, fix a profile $\mathcal I$ and let $v$ denote the vector such that $v_X$ states how often the interval $X$ is reported in $\mathcal I$. It is easy to check that $\sum_{X\in\Lambda} v_X \pi_\alpha(X, x_i)$ is equal to $\Pi_\alpha(\mathcal I,x_i)$. We hence conclude for every profile $\mathcal I\in\dom$ with corresponding vector $v$ that
    \begin{align*}
    f(\mathcal I)=\hat g(v)
    \in  \{x_i\in A\colon v\in \bar Q_i\}
    =\{x_i\in A\colon \Pi_{\alpha}(\mathcal I,x_{i-1})\leq \theta_{i-1}n_\mathcal{I}\land \Pi_{\alpha}(\mathcal I,x_{i})\geq \theta_{i}n_\mathcal{I}\}.
    \end{align*}

    Based on right-based continuity, we next show that $f$ picks the left-most alternatives in this set, i.e., $f(\mathcal I)=\max_{\ssucc} \{x_i\in A\colon \Pi_{\alpha}(\mathcal I,x_{i-1})\leq \theta_{i-1}n_\mathcal{I}\land \Pi_{\alpha}(\mathcal I,x_{i})\geq \theta_{i}n_\mathcal{I}\}=\max_{\ssucc} \{x_i\in A\colon \Pi_{\alpha}(\mathcal I,x_{i})\geq \theta_{i}n_\mathcal{I}\}$. Finally, we infer based on robustness that $\theta$ and $\alpha$ are compatible, which proves that $f$ is the position-threshold rule defined by $\theta$ and $\alpha$. 
\end{proof}

\begin{remark}
    All axioms of \Cref{thm:characterization} are necessary for our characterization: if we only drop anonymity, position-threshold rules that weight ``even'' voters $i\in 2\mathbb{N}$ twice and ``odd'' voters $i\in 2\mathbb{N}+1$ only once satisfy all given axioms. If we only omit right-biased continuity, we can, for instance, define position-threshold rules that select the left-most alternative $x_i$ such that $\Pi_\alpha(\mathcal I, x_{i})> \theta_{i} n_{\mathcal I}$. When omitting unanimity, every constant voting rule satisfies the given axioms. When omitting robustness or weakening robustness to strategyproofness, one can extend the class of position-threshold rules by allowing for arbitrary individual position functions (i.e., no consistency between different intervals is required anymore). Finally, if we only drop reinforcement, position-threshold rules that use different weight and threshold vectors depending on $n_\mathcal{I}$ satisfy all remaining axioms. For instance, the following rule satisfies all conditions except for reinforcement and it is no position-threshold rule: if $\lceil\log_2 n_{\mathcal I}\rceil$ is odd, we apply the median rule with respect to the left endpoints of the voters' intervals, and if it is even, we apply the median rule with respect to the right endpoints of the voters' intervals. 
\end{remark}

\begin{remark}
    The variable-electorate framework is required for \Cref{thm:characterization}. To see this, we define the (profile-dependent) weight vector $\alpha(\mathcal I)$ by $\alpha_1(\mathcal I)=\frac{1}{2}-\frac{|i\in N_{\mathcal I}\colon x_1\not\in I_i\}|}{2n_\mathcal{I}}$, and $\alpha_j(\mathcal I)=1$ for all $j\in \{2,\dots, m\}$ and we let $\theta=(\frac{1}{2},\dots,\frac{1}{2})$. Now, consider the following rule $f$ given by $f(\mathcal I)=\max_{\rhd}\{x_i\in A\colon \Pi_{\alpha(\mathcal I)}(\mathcal I, x_i)\geq \theta_i n_{\mathcal I}\}$, where $\Pi_{\alpha(\mathcal I)}$ is the collective position function given by the vector $\alpha(\mathcal I)$. We first note that $f$ is robust: if the number of voters $i$ with $x_1\not\in I_i$ does not change, this is true as $f$ behaves like a position-threshold rule. By contrast, if some voter changes his interval by removing $x_1$, the collective position of $x_1$ only decreases while the collective position of all other alternatives remains the same. This means that the winner either changes from $x_1$ to $x_2$ or not at all, so robustness is satisfied. Further, it can be verified that $f$ is anonymous, unanimous, and right-biased continuous. However, even when the electorate $N$ is fixed, this rule is no position-threshold rule. 
    In more detail, we suppose there are $n=100$ voters and we assume for contradiction that $f$ is a position-threshold rule. This means that $f$ is defined by a threshold vector $\phi$ and a weight vector $\beta$. First, consider the profile $\mathcal I$ where $50$ voters report $\{x_1\}$ and $50$ voters report $\{x_2\}$. It holds that $f(\mathcal I)=x_1$ since the weight vector does not matter if all voters report a single alternative. This implies that $\phi_1\leq \frac{\Pi_\beta(\mathcal I, x_1)}{100}=\frac{1}{2}$. Next, consider the profile $\mathcal I'$ where all $100$ voters report $\{x_1,x_2\}$. Using the definition of $f$, we derive that $f(\mathcal I')=x_1$, so $\beta_1\geq \phi_1$. Finally, consider the profile $\mathcal I''$ where $50$ voters report $\{x_1,x_2\}$, $25$ voters report $\{x_1\}$, and $25$ voters report $\{x_2\}$. By the definition of $f$, it follows that $f(\mathcal I'')=x_2$ because $\pi_{\alpha(\mathcal I'')}(\{x_1,x_2\}, x_1)<\frac{1}{2}$. However, $\Pi_\beta(\mathcal I'', x_1)=25+50\beta_1\geq 100\phi_1$, so $f$ cannot be a position-threshold rule. 
\end{remark}

\begin{remark}
We can simplify \Cref{thm:characterization} when replacing unanimity with weak efficiency, which requires of a voting rule $f$ that an alternative can only be chosen if it is contained in the interval of at least one voter. In particular, this axiom implies for the threshold vector $\theta$ of a position-threshold rule that $\theta_i=\theta_j$ for all $i,j\in \{1,\dots, m-1\}$. Indeed, if this was not true, there is an index $i\in \{1,\dots, m-2\}$ such that $\theta_i>\theta_{i+1}$ because $\theta$ is non-increasing. Further, we can find two integers $\ell_1$ and $\ell_2$ such that $\theta_i>\frac{\ell_1}{\ell_1+\ell_2}>\theta_{i+1}$. By again using that $\theta$ is non-decreasing, we derive now that any position-threshold rule $f$ defined by $\theta$ (and an arbitrary compatible weight vector $\alpha$) satisfies that $f(\mathcal I)=x_{i+1}$ for the profile $\mathcal I$ where $\ell_1$ voters report $\{x_1\}$ and $\ell_2$ voters report $\{x_m\}$. However, this violates weak efficiency as no voter reports $x_{i+1}$, so we have indeed that $\theta_i=\theta_j$ for all $i,j\in \{1,\dots, m-1\}$. In turn, the compatibility of $\alpha$ and $\theta$ together with Statement (2) in \Cref{lem:robust} requires that $\alpha_1\leq \alpha_2\leq \dots \leq \alpha_m$. Since it is easy to verify that all position-threshold rules that satisfy these additional constraints on the threshold vector are weakly efficient, this yields an attractive variant of \Cref{thm:characterization}.
\end{remark}

\begin{remark}
    \citet{Berg98a} characterizes the set of strategyproof voting rules on the interval domain (with underlying weakly single-peaked preference relations) for a fixed electorate. In more detail, this author shows that every strategyproof voting rule can be described via a two-step procedure: first, we use for every voter $i\in N$ a tie-breaking function $h_i$, which maps the current interval profile $\mathcal I$ to a single alternative $x=h_i(\mathcal I)\in I_i$. Moreover, for each voter $i$, the function $h_i$ is required to be strategyproof for all voters but $i$ (with respect to the underlying weakly single-peaked preferences). As the second step, we apply a minmax rule (i.e., a non-anonymous variant of a phantom median rule) to the profile where each voter $i$ reports $h_i(\mathcal I)$. Since robustness implies strategyproofness, every position-threshold rule $f$ can be represented as such a tie-breaking minmax rule for a fixed electorate. In particular, the minmax rule of $f$ is the phantom median rule to which $f$ reduces when voters report single alternatives. Further, for the tie-breaking rule, let $I_i=[x_\ell, x_r]$ for an arbitrary voter $i$. Then, we can set $h_i(\mathcal I)=f(\mathcal I)$ if $f(\mathcal I)\in I_i$, $h_i(\mathcal I)=x_\ell$ if $f(\mathcal I)\rhd x_\ell$, $h_i(\mathcal I)=x_r$ if $x_r\rhd f(\mathcal I)$. A similar construction has been used by \citet{Berg98a} for his class of extended median voting schemes. 
\end{remark}

\subsection{Characterization of the Endpoint-Median Rule}\label{subsec:endponintmedian}

A natural follow-up question to \Cref{thm:characterization} is which position-threshold rule to use in practice. We will give one possible answer to this question by characterizing the \emph{endpoint-median rule $f_\mathit{EM}$}, which is the position-threshold rule defined by the weight vector $\alpha=(\frac{1}{2},\dots,\frac{1}{2})$ and the threshold vector $\theta=(\frac{1}{2},\cdots,\frac{1}{2})$. We note that the endpoint-median rule has a much simpler formulation when avoiding the formalism of position-threshold rules: we substitute the interval $[x_\ell,x_r]$ of every voter with the intervals $\{x_\ell\}$ and $\{x_r\}$ and compute the median rule on this simplified profile. For our characterization of the endpoint-median rule, we will rely on the following two axioms.

\paragraph{Majority criterion.} The majority criterion states that an alternative should be chosen if it is uniquely reported by a strict majority of the voters. More formally, a voting rule $f$ on $\dom$ satisfies the \emph{majority criterion} if $f(\mathcal I)=x_j$ for all profiles $\mathcal I\in\dom$ and alternatives $x_j\in A$ such that $|\{i\in N_{\mathcal I}\colon I_i=\{x_j\}\}|>\frac{n_\mathcal{I}}{2}$. We note that this axiom can be seen as a mild variant of Condorcet-consistency for the interval domain. 

\paragraph{Strong unanimity.} Strong unanimity strengthens unanimity by requiring that, if an alternative is reported by all voters, it should be chosen even if some voters approve additional alternatives. Since multiple alternatives can be reported by all voters in such cases, \emph{strong unanimity} formally postulates of a voting rule $f$ on $\dom$ that $f(\mathcal I)\in \bigcap_{i\in N_{\mathcal I}} I_i$ for all interval profiles $\mathcal I\in \dom$ with $\bigcap_{i\in N_{\mathcal I}} I_i\neq \emptyset$.\medskip

We are now ready to prove our characterization of the endpoint-median rule. 

\begin{theorem}\label{thm:AWS1/2}
    The endpoint-median rule is the only position-threshold rule that satisfies the majority criterion and strong unanimity. 
\end{theorem}
\begin{proof} 
We will prove both directions of the theorem separately. \medskip

$(\implies)$ We start by showing that $f_\mathit{EM}$ satisfies our two axioms and thus define $\alpha$ and $\theta$ as the weight and threshold vector of this rule. First, we analyze the majority criterion. For this, consider an interval profile $\mathcal I$ and an alternative $x_j$ such that more than $\frac{n_\mathcal{I}}{2}$ voters in $\mathcal I$ report $\{x_j\}$. It is straightforward to check that $\Pi_\alpha(\mathcal I, x_j)>\frac{n_\mathcal{I}}{2}=\theta_j n_\mathcal{I}$. Moreover, since $\pi_\alpha(\{x_j\}, x_h)=0$ for all alternatives $x_h$ with $x_h\ssucc x_j$, it holds for all such alternatives that $\Pi_\alpha(\mathcal I, x_h)<\frac{n_\mathcal{I}}{2}= \theta_h n_\mathcal{I}$. It thus follows that $f_\mathit{EM}(\mathcal I)=x_j$, so the endpoint-median rule satisfies the majority criterion. 

For proving that $f_\mathit{EM}$ also satisfies strong unanimity, let $\mathcal I$ denote a profile such that $\bigcap_{i\in N_\mathcal{I}} I_i\neq \emptyset$. Moreover, we define $x_j$ as the left-most alternative in $\bigcap_{i\in N_\mathcal{I}} I_i$, i.e., $x_j=\max_{\ssucc} \bigcap_{i\in N_\mathcal{I}} I_i$. By the definition of the weight vector $\alpha$, we have that $\pi_\alpha(I,x_j)\geq \frac{1}{2}$ for every interval $I\in\Lambda$ with $x_j\in I$. Since $x_j\in I_i$ for all $i\in N_\mathcal{I}$, it follows that $\Pi_\alpha(\mathcal I, x_j)\geq \frac{1}{2}n_\mathcal{I}=\theta_jn_\mathcal{I}$. Next, let $x_h$ denote an alternative with $x_h\ssucc x_j$. Since $x_j\in I_i$ for all voters $i\in N_{\mathcal I}$, $x_h$ is not the right-most approved alternative of any voter. We infer from this insight that $\pi(I_i, x_h)\leq \frac{1}{2}$ for all voters $i\in N_{\mathcal I}$. Further, since $x_j$ is the left-most alternative that is reported by all voters, there is a voter $i$ with $x_h\not \in I_i$ and $x_j\in I_i$, which implies that $\pi_\alpha(I_i,x_h)=0$. Combining these insights shows that $\Pi_\alpha(\mathcal I, x_h)<\frac{1}{2}n_\mathcal{I}=\theta_h n_\mathcal{I}$. Since this analysis holds for all alternatives $x_h$ with $x_h\ssucc x_j$, we derive that $f_\mathit{EM}(\mathcal I)=x_j$, so the endpoint-median rule satisfies strong unanimity.\medskip

$(\impliedby)$ Let $f$ denote a position-threshold rule that satisfies the majority criterion and strong unanimity. Moreover, we let $\alpha$ and $\theta$ denote the weight and threshold vector of $f$. We will prove that $\alpha$ and $\theta$ are the vectors of the endpoint-median rule. Hence, we first show that $\theta_i=\frac{1}{2}$ for all alternatives $x_i\in A$. To this end, we observe that $\theta_m$ has no influence on the outcome of $f$, so we can always define $\theta_m=\frac{1}{2}$. Next, we assume for contradiction that $\theta_i<\frac{1}{2}$ for some $i\in \{1,\dots, m-1\}$. In this case, let $w_1,w_2\in\mathbb{N}$ denote two integers such that $\theta_i<\frac{w_1}{w_1+w_2}<\frac{1}{2}$. Such integers exist because every rational value $q\in\mathbb{Q}\cap (0,1)$ can be written as $q=\frac{w_1'}{w_1'+w_2'}$ for two integers $w_1', w_2'\in\mathbb{N}$. Now, consider the profile $\mathcal I$ where $w_1$ voters report $\{x_i\}$ and $w_2$ voters report $\{x_m\}$. It holds that $f(\mathcal I)=x_i$ because $\Pi_{\alpha}(\mathcal I, x_h)=0$ for all $x_h\in A$ with $x_h\ssucc x_i$ and $\Pi_{\alpha}(\mathcal I, x_h)=w_1>\theta_in_\mathcal{I}$. However, $\frac{w_1}{w_1+w_2}<\frac{1}{2}$ implies that $w_1<w_2$, so a strict majority of the voters report $\{x_m\}$. The majority criterion thus postulates that $f(\mathcal I)=x_m$, which contradicts that $f(\mathcal I)=x_i$. As the second case, suppose that $\theta_i>\frac{1}{2}$ for some $i\in \{1,\dots, m-1\}$. In this case, we can find two integers $w_1,w_2\in\mathbb{N}$ such that $\frac{1}{2}<\frac{w_1}{w_1+w_2}<\theta_i$ and we consider again the profile $\mathcal I$ where $w_1$ voters report $\{x_i\}$ and $w_2$ voters report $\{x_m\}$. This time, it can be checked that $x_i\ssucc f(\mathcal I)$ as $\Pi(\mathcal I, x_h)=0$ for all $x_h\in A$ with $x_h\ssucc x_i$ and $\Pi(\mathcal I, x_i)=w_1<\theta_i n_\mathcal{I}$. However, since $\frac{1}{2}<\frac{w_1}{w_1+w_2}$, the majority criterion postulates that $f(\mathcal I)=x_i$, so we have again a contradiction. Thus, it follows that $\theta_i=\frac{1}{2}$ for all $i\in \{1,\dots, m-1\}$. 

Next, we will show that $\alpha_i=\frac{1}{2}$ for all $i\in \{1,\dots, m-1\}$ and note that $\alpha_m$ is irrelevant for the definition of $f$. We use again a case distinction for this and first suppose that $\alpha_i<\frac{1}{2}$ for some $i\in \{1,\dots, m-1\}$. In this case, we define $\delta=\frac{1}{2}-\alpha_i$ and choose an integer $t\in\mathbb{N}$ such $t\delta>1$. Now, consider the profile $\mathcal I$ where $t$ voters report $\{x_i, x_{i+1}\}$ and a single voter reports $\{x_i\}$. It holds that $\Pi_\alpha(\mathcal I,x_i)=1+t\alpha_i=1+\frac{1}{2} t -\delta t<\frac{1}{2} n_{\mathcal I}$ because $\pi_\alpha(\{x_i,x_{i+1}\}, x_i)=\alpha_i=\frac{1}{2}-\delta$. Since $\theta_i=\frac{1}{2}$, this means that $f(\mathcal I)\neq x_i$. However, $x_i$ is the only alternative that is reported by all voters in $\mathcal I$, so strong unanimity requires that $f(\mathcal I)=x_i$, a contradiction. For the second case, we suppose that $\alpha_i>\frac{1}{2}$ and we define $\delta=\alpha_i-\frac{1}{2}$. Now, we choose $t\in\mathbb{N}$ such that $t\delta>\frac{1}{2}$ and we consider this time the profile $\mathcal I$ where $t$ voters report $\{x_i, x_{i+1}\}$ and a single voter reports $\{x_{i+1}\}$. Analogous to the last case, it can be checked that $\Pi_\alpha(\mathcal I, x_i)=t\alpha_i=\frac{1}{2} t + \delta t>\frac{1}{2}n_\mathcal{I}$. This means that $f(\mathcal I)\ssucceq x_i$. However, $x_{i+1}$ is the only alternative that is reported by all voters in $\mathcal I$, so strong unanimity requires that $f(\mathcal I)=x_{i+1}$. Because we have a contradiction in both cases, we conclude that the assumption that $\alpha_i\neq \frac{1}{2}$ is wrong, i.e., it holds for all $i\in \{1,\dots, m-1\}$ that $\alpha_i= \frac{1}{2}$. This shows that $\alpha$ and $\theta$ are the weight and threshold vectors of the endpoint-median rule, so $f=f_\mathit{EM}$.
\end{proof}

\begin{remark}
On the domain of single-peaked preference relations, the median rule is the only phantom median rule that satisfies the majority criterion. Therefore, we interpret \Cref{thm:AWS1/2} as demonstrating that the endpoint-median rule is the ``correct'' extension of the median rule to the interval domain. Moreover, in combination, \Cref{thm:characterization,thm:AWS1/2} show that the endpoint-median rule is the only voting rule on $\dom$ that preserves all desirable properties of the median rule because $f_\mathit{EM}$ is the unique voting rule on the interval domain that satisfies anonymity, strong unanimity, the majority criterion, robustness, reinforcement, and right-biased continuity.
\end{remark}

\begin{remark}
    The endpoint-median rule can also be characterized as the most neutral position-threshold rule. In more detail, this rule is the only position-threshold rule that is shift-symmetric (if we move the interval of every voter one position to the right, the winner will also move one position to the right) and orientation-symmetric (if we exchange $x_i$ with $x_{m+1-i}$ for all $i\in \{1,\dots,m\}$ in the interval of every voter, the winner will change from $x_j$ to $x_{m+1-j}$ unless there is an alternative $x_i$ with $\Pi_\alpha(\mathcal I, x_i)=\theta_i n_{\mathcal I}$). This again mirrors the behavior of the median rule for single-peaked preference relations because this is the only phantom median rule that satisfies these conditions. 
\end{remark}

\section{Conclusion}

In this paper, we study voting rules for the interval domain, where voters report subintervals of a set of linearly ordered alternatives to indicate their preferences. As our main contribution, we propose and characterize the class of position-threshold rules, which generalize Moulin's phantom median rules \citep{Moul80a} to the interval domain. In essence, position-threshold rules compute for each alternative a collective position, which quantifies the voters' relative positions with respect to this alternative, and choose the left-most alternative whose collective position exceeds its threshold value. As our main result, we characterize these rules based on robustness (which demands that small changes in a voter's interval lead to analogous changes in the outcome or no changes at all), reinforcement (which demands that, if an alternative is chosen for two disjoint elections, it is also chosen when combining these elections), and three mild auxiliary conditions called anonymity, unanimity, and right-biased continuity. 

Moreover, we propose and characterize the endpoint-median rule, which replaces the interval of each voter with two singleton ballots corresponding to the endpoints of the interval and then computes the median rule. In more detail, we show that this rule is the only position-threshold rule that satisfies the majority criterion (an alternative is guaranteed to be chosen if it is uniquely reported by more than half of the voters) and strong unanimity (if some alternatives are reported by all voters, one such alternative is chosen). Since the median rule is the only phantom median rule that satisfies these conditions, our result suggests that the endpoint-median rule is the  ``correct'' extension of the median rule to the interval domain.

We note that our paper offers several directions for future work. Firstly, we believe that it is interesting to further analyze the axiomatic properties of position-threshold rules. This may help to identify new desirable voting rules on the interval domain or to strengthen the argument for the endpoint-median rule. Moreover, it may be fruitful to analyze position-threshold rules also in the context of facility location on the real line. An interesting question regarding this is, e.g., how much social welfare position-threshold rules can guarantee. Finally, it seems worthwhile to study voting rules on the interval domain that fail robustness but, e.g., satisfy strategyproofness to give a more complete pictures about the possibilities arising from this domain. 

\section*{Acknowledgements}

I thank Evghenii Beriozchin for helpful discussions and Felix Brandt for valuable feedback. This work was funded by the NSF-CSIRO grant on ``Fair Sequential Collective Decision-Making'' (RG230833).

\newpage
\appendix

\section{Proof of \Cref{thm:characterization}}

In this appendix, we will prove \Cref{thm:characterization}: position-threshold rules are the only voting rules on $\dom$ that satisfy anonymity, unanimity, robustness, reinforcement, and right-biased continuity. More specifically, we will show in \Cref{subsec:properties} that position-threshold rules satisfy all desired properties and in \Cref{subsec:characterizationproof} that these properties indeed characterize position-threshold rules. 

\subsection{Axiomatics of Position--Threshold Rules}\label{subsec:properties}

We start by showing that all position-threshold rules satisfy the five axioms of \Cref{thm:characterization}. To this end, we first recall that position-threshold rules are defined by a weight vector $\alpha$ and threshold $\theta$ vector that are compatible. This means that, for all profiles $\mathcal I$ and alternatives $x_i\in A\setminus \{x_m\}$, it holds that $\Pi_\alpha(\mathcal I, x_i)\geq n_\mathcal{I} x_i$ implies that $\Pi_\alpha(\mathcal I, x_{i+1})\geq n_\mathcal{I} x_{i+1}$. As an auxiliary tool, we first characterize the weight and threshold vectors that are compatible and show that the compatibility condition is necessary for position-threshold rules to satisfy robustness.
\begin{lemma}\label{lem:robust}
    Let $\theta\in(0,1)^m$ denote a threshold vector and $\alpha\in[0,1]^m$ a weight vector. The following statements are equivalent: 
    \begin{enumerate}[label=(\arabic*)]
    \item The vectors $\alpha$ and $\theta$ are compatible.
    \item It holds that $\alpha_{i+1}-\alpha_i\geq (\theta_{i+1}-\theta_i)\max(\frac{\alpha_i}{\theta_i}, \frac{1-\alpha_i}{1-\theta_i})$ for all $i\in \{1,\dots, m-2\}$.
    \item The voting rule $f$ given by $f(\mathcal I)=\max_{\ssucc} \{x_i\in A\colon \Pi_\alpha(\mathcal I, x_i)\geq \theta_i n_{\mathcal I}\}$ is robust. 
    \end{enumerate}
\end{lemma}
\begin{proof}
Fix a threshold vector $\theta$ and a weight vector $\alpha$, and let $f$ be defined as in the lemma. We will break the lemma down into three implications. Firstly, we show that the inequality given in Statement (2) indeed implies that $\alpha$ and $\theta$ are compatible. Secondly, we will prove that if $\alpha$ and $\theta$ are compatible, then $f$ is robust. Lastly, we will prove that, if $f$ is robust, the inequality in Statement (2) has to hold. By the transitivity of logical implications, this proves the lemma.\bigskip

\textbf{Claim 1: (2) implies (1).} Suppose that $\alpha_{i+1}-\alpha_i\geq (\theta_{i+1}-\theta_i)\max(\frac{\alpha_i}{\theta_i}, \frac{1-\alpha_i}{1-\theta_i})$ for all $i\in \{1,\dots, m-2\}$ and consider an arbitrary interval profile $\mathcal I\in\dom$ and an alternative $x_i\in A\setminus \{x_m\}$ such that $\Pi_\alpha(\mathcal I, x_i)\geq \theta_i n_\mathcal{I}$. We will show that $\Pi_\alpha(\mathcal I, x_{i+1})\geq \theta_{i+1}n_\mathcal{I}$, which proves that $\alpha$ and $\theta$ are compatible. For this, we first observe that this implication holds trivially for $x_{m-1}$ since $\Pi_\alpha(\mathcal I, x_m)=n_\mathcal{I}$ for all profiles $\mathcal I$. We thus assume that $x_i\neq x_{m-1}$ and we partition the voters in three sets regarding their position with respect to $x_i$: 
$L=\{j\in N_{\mathcal I}\colon I_j\subseteq[x_1, x_i]\}$ are the voters that are weakly left of $x_i$, $R=\{j\in N_{\mathcal I}\colon I_j\subseteq[x_{i+1}, x_m]\}$ are the voters that are fully right of $x_i$, and $Z=N\setminus (L\cup R)$ is the set of voters who report an interval $I_j$ with $\{x_i,x_{i+1}\}\subseteq I_j$.
    By the definition of these sets, we have that $\pi_\alpha(I_j,x_i)=1$ for all $j\in L$, $\pi_\alpha(I_j,x_i)=0$ for all $j\in R$, and $\pi_\alpha(I_j,x_i)=\alpha_i$ for all $j\in Z$. 
    Moreover, it holds that $\pi_\alpha(I_j,x_{i+1})=1$ for all $j\in L$, $\pi_\alpha(I_j, x_{i+1})\geq \alpha_{i+1}$ for all $j\in Z$, and $\pi_\alpha(I_j,x_{i+1})\geq 0$ for all $j\in R$. 
    Now, we define $\ell_{\mathcal I}=\frac{|L|}{n_\mathcal{I}}$, $r_{\mathcal I}=\frac{|R|}{n_\mathcal{I}}$, and $z_{\mathcal I}=\frac{|Z|}{n_\mathcal{I}}$. Since $\Pi_\alpha(\mathcal I, x_i)\geq \theta_in_\mathcal{I}$, it holds that $\ell_{\mathcal I} + z_{\mathcal I} \alpha_i \geq \theta_i$ and we aim to show that $\ell_{\mathcal I} + z_{\mathcal I}\alpha_{i+1} \geq \theta_{i+1}$ because this implies that $\Pi_\alpha(\mathcal I, x_{i+1})\geq \theta_{i+1}n_\mathcal{I}$. First, if $\alpha_{i+1}\geq \alpha_i$, this follows immediately as $\theta_i\geq \theta_{i+1}$. We hence assume that $\alpha_{i+1}<\alpha_i$. Since $\alpha_{i+1}\geq 0$ by the definition of a weight vector, this further implies that $\alpha_i>0$. 
    Next, we proceed with a case distinction with respect to whether $\frac{\alpha_i}{\theta_i}\geq \frac{1-\alpha_i}{1-\theta_i}$.\medskip

    \emph{Case 1:} We first assume that $\frac{\alpha_i}{\theta_i}\geq \frac{1-\alpha_i}{1-\theta_i}$. In this case, we will minimize the term $\ell+z \alpha_{i+1}$ subject to $\ell+z\alpha_i \geq \theta_i$, $\ell\geq 0$, and $z\geq 0$. To this end, we observe that we can set $\ell=0$: if $\ell>0$, we define $z'=z+\frac{\ell}{\alpha_i}$ and $\ell'=0$. It is easy to see that our constraints are still satisfied. Moreover, since $\alpha_{i+1}<\alpha_i$, it holds that 
    \begin{align*}
        \ell+z \alpha_{i+1}-\ell'-z'\alpha_{i+1}&=\ell+z\alpha_{i+1}-\left(z+\frac{\ell}{\alpha_i}\right)\alpha_{i+1}
        = \ell\left(1-\frac{\alpha_{i+1}}{\alpha_i}\right)
        >0. 
    \end{align*}

    Hence, we need to set $\ell=0$ to minimize $\ell + z\alpha_{i+1}$. In turn, to satisfy that $\ell+z\alpha_i\geq \theta_i$, we set $z=\frac{\theta_i}{\alpha_i}$, i.e., as the minimal value such that $z\alpha_i\geq \theta_i$ is satisfied. As a consequence of this analysis, the optimal value our linear program is $\frac{\theta_i \alpha_{i+1}}{\alpha_i}$. Since  $\ell_{\mathcal I}$ and $z_{\mathcal I}$ are a feasible solution to this linear program, it follows that $\ell_{\mathcal I} + z_\mathcal{I}\alpha_{i+1}\geq \frac{\theta_i \alpha_{i+1}}{\alpha_i}$. Finally, because $\frac{\alpha_i}{\theta_i}\geq \frac{1- \alpha_i}{1-\theta_i}$ by assumption, we have that $\alpha_{i+1}-\alpha_i\geq (\theta_{i+1}-\theta_i)\max(\frac{\alpha_i}{\theta_i}, \frac{1-\alpha_i}{1-\theta_i})=(\theta_{i+1}-\theta_i) \frac{\alpha_i}{\theta_i}$. We now conclude that $\Pi_\alpha(\mathcal I, x_{i+1})\geq \theta_{i+1}n_\mathcal{I}$ as
    \begin{align*}
        \frac{\Pi_\alpha(\mathcal I, x_{i+1})}{n_\mathcal{I}}\geq \ell_\mathcal{I}+z_\mathcal{I}\alpha_{i+1}\geq\frac{\theta_i \alpha_{i+1}}{\alpha_i}=\theta_i + (\alpha_{i+1}-\alpha_i)\frac{\theta_i}{\alpha_{i}}\geq \theta_i + (\theta_{i+1}-\theta_i) =\theta_{i+1}.
    \end{align*}

\emph{Case 2:} For our second case, we suppose that $\frac{\alpha_i}{\theta_i}< \frac{1-\alpha_i}{1-\theta_i}$. Equivalently, this means that $\theta_i> \alpha_i$, which implies that $1-\alpha_i>0$ as $\theta_i<1$. We will use a similar approach as in the first case and minimize the term $\ell+z\alpha_{i+1}$ subject to $\ell+z\alpha_i \geq \theta_i$, $\ell\geq 0$, $z\geq 0$, and $\ell+z\leq 1$. Just as in Case 1, it can be shown that the objective value of this linear program is minimal if $\ell$ is chosen to be minimal. However, since $\theta_i> \alpha_i$ and we require that $\ell+z\leq 1$, we cannot set $\ell$ to $0$ without violating that $\ell+z\alpha_i \geq \theta_i$. Instead, by using a similar reasoning as in the first case, one can compute that the values of $\ell$ and $z$ that minimize $\ell+z\alpha_{i+1}$ subject to our constraints are $\ell=1-\frac{1-\theta_i}{1-\alpha_i}$ and $z=\frac{1-\theta_i}{1-\alpha_i}$. For these values, our objective value is 
\begin{align*}
    1-\frac{1-\theta_i}{1- \alpha_i}+\frac{1-\theta_i}{1-\alpha_i} \alpha_{i+1}&=1-(1- \alpha_i)\frac{1-\theta_i}{1-\alpha_i} + \frac{1-\theta_i}{1- \alpha_i} (\alpha_{i+1}-\alpha_i)\\
    &=\theta_i + \frac{1-\theta_i}{1- \alpha_i} (\alpha_{i+1}-\alpha_i).
\end{align*}

Since $\ell_{\mathcal I}$ and $z_\mathcal{I}$ form a feasible solution to our linear program, we infer that $\frac{\Pi_{\alpha}(\mathcal I, x_{i+1})}{n_{\mathcal I}}$ is lower-bounded by this expression. Finally, since $\frac{ \alpha_i}{\theta_i}< \frac{1-\alpha_i}{1-\theta_i}$, we derive that $\alpha_{i+1}-\alpha_i\geq (\theta_{i+1}-\theta_i)\max(\frac{\alpha_i}{\theta_i}, \frac{1-\alpha_i}{1-\theta_i})=(\theta_{i+1}-\theta_i) \frac{1-\alpha_i}{1-\theta_i}$. Combined with our previous insight, this means that 
\begin{align*}
    \frac{\Pi_\alpha(\mathcal I, x_{i+1})}{n_\mathcal{I}}
    \geq \theta_i + \frac{1-\theta_i}{1-\alpha_i} (\alpha_{i+1}-\alpha_i)
    \geq \theta_i+(\theta_{i+1}-\theta_i)\cdot \frac{1-\alpha_i}{1-\theta_i} \cdot \frac{1-\theta_i}{1-\alpha_i}
    =\theta_{i+1}.
\end{align*}\medskip

\textbf{Claim 2: (1) implies (3).} We will now show that $f$ is robust if $\alpha$ and $\theta$ are compatible. To this end, consider an arbitrary interval profile $\mathcal I$, let $i\in N_{\mathcal I}$ denote a voter, and let $x_\ell$ and $x_r$ denote alternatives such that $x_\ell\ssucc x_r$ and $I_i=[x_\ell, x_r]$. 
First, we analyze the profile $\mathcal{I}^{i\downarrow x_\ell}$ and note that $\Pi_\alpha(\mathcal{I}^{i\downarrow x_\ell}, x_\ell)\leq \Pi_\alpha(\mathcal{I}, x_\ell)$ because $\pi_\alpha(\mathcal I, x_\ell)=\alpha_i\geq 0=\pi_\alpha(\mathcal{I}^{i\downarrow x_\ell}, x_\ell)$. 
Moreover, $\Pi_\alpha(\mathcal{I}^{i\downarrow x_\ell}, x_h)= \Pi_\alpha(\mathcal{I}, x_h)$ for all $x_h\in A\setminus \{x_{\ell}\}$ because the relative position of no voter changed with respect to $x_h$. 
Now, if $f(\mathcal I)=x_j$ for some alternative $x_j\neq x_\ell$, it is easy to show that $f(\mathcal{I}^{i\downarrow x_\ell})=x_j$, too. 
In more detail, if $x_j\ssucc x_\ell$, then the value $\Pi_\alpha(\mathcal I, x_\ell)$ does not matter as $\Pi_\alpha(\mathcal I^{i\downarrow x_\ell}, x_j)=\Pi_\alpha(\mathcal I, x_j)\geq \theta_j n_\mathcal{I}$ and $x_j\ssucc x_\ell$. 
On the other hand, if $x_\ell\ssucc x_j$, then $\Pi_\alpha(\mathcal{I}^{i\downarrow x_\ell}, x_\ell)\leq \Pi_\alpha(\mathcal{I}, x_\ell)< \theta_\ell n_\mathcal{I}$ and the outcome again remains the same. Finally, if $f(\mathcal I)=x_\ell$, then $\Pi_\alpha(\mathcal I^{i\downarrow x_\ell}, x_h)=\Pi_\alpha(\mathcal I^{i\downarrow x_\ell}, x_h)<\theta_hn_\mathcal{I}$ for all $x_h\ssucc x_\ell$, so $x_\ell\ssucceq f(\mathcal I^{i\downarrow x_\ell})$. 
Moreover, due to the compatibility of $\alpha$ and $\theta$, it holds that $\Pi_\alpha(\mathcal{I}^{i\downarrow x_\ell}, x_{\ell+1})=\Pi_\alpha(\mathcal{I}, x_{\ell+1})\geq \theta_{\ell+1}n_\mathcal{I}$ because $f(\mathcal I)=x_\ell$ implies that $\Pi_\alpha(\mathcal I, x_\ell)\geq \theta_\ell n_\mathcal{I}$. This means that $f(\mathcal{I}^{i\downarrow x_\ell})\ssucceq x_{\ell+1}$, so $f(\mathcal{I}^{i\downarrow x_\ell})\in \{x_\ell, x_{\ell+1}\}$ and robustness is satisfied.

Next, consider the profile $\mathcal{I}^{i\downarrow x_r}$, for which $\pi_\alpha(\mathcal{I}^{i\downarrow x_r}, x_{r-1})=1\geq \pi_\alpha(\mathcal{I}, x_{r-1})$ and $\pi_\alpha(\mathcal{I}^{i\downarrow x_r}, x_{r})=1= \pi_\alpha(\mathcal{I}, x_{r})$. This means that $\Pi_\alpha(\mathcal{I}^{i\downarrow x_r}, x_{r-1})\geq \Pi_{\alpha}(\mathcal{I}, x_{r-1})$ and $\Pi_\alpha(\mathcal{I}^{i\downarrow x_r}, x_{h})= \Pi_{\alpha}(\mathcal{I}, x_{h})$ for all $x_h\in A\setminus \{x_{r-1}\}$. We first assume that $f(\mathcal{I})=x_j\neq x_{r}$. If $x_j\ssucc x_{r}$, then $f(\mathcal{I}^{i\downarrow x_r})=f(\mathcal{I})$ because $\Pi_\alpha(\mathcal{I}^{i\downarrow x_r}, x_{h})= \Pi_{\alpha}(\mathcal{I}, x_{h})$ for all $x_h$ with $x_h\ssucc x_{r-1}$ and $\Pi_\alpha(\mathcal{I}^{i\downarrow x_r}, x_{r-1})\geq \Pi_{\alpha}(\mathcal{I}, x_{r-1})$. On the other hand, if $x_{r}\ssucc x_j$, we infer that $\Pi_\alpha(\mathcal{I}^{i\downarrow x_r}, x_r)=\Pi_\alpha(\mathcal I, x_r)<\theta_r n_\mathcal{I}$. By the contrapositive of our compatibility condition, this means that $\Pi_\alpha(\mathcal{I}^{i\downarrow x_r}, x_{r-1})<\theta_{r-1}n_\mathcal{I}$. Hence, it is now easy to derive that $f(\mathcal I)=f(\mathcal{I}^{i\downarrow x_r})$ since $\Pi_\alpha(\mathcal{I}^{i\downarrow x_r}, x_{h})= \Pi_{\alpha}(\mathcal{I}, x_{h})$ for all alternatives $x_h\in A\setminus \{x_{r-1}\}$. Finally, assume that $f(\mathcal I)=x_r$. This means that $\Pi_\alpha(\mathcal{I}^{i\downarrow x_r}, x_{h})= \Pi_{\alpha}(\mathcal{I}, x_{h})<\theta_hn_\mathcal{I}$ for all $x_h$ with $x_h\ssucc x_{r-1}$ and $\Pi_\alpha(\mathcal{I}^{i\downarrow x_r}, x_{r})= \Pi_{\alpha}(\mathcal{I}, x_{r})\geq \theta_r n_\mathcal{I}$. It follows that $f(\mathcal I^{i\downarrow x_r})\in \{x_{r-1}, x_r\}$ and robustness holds again. \bigskip

\textbf{Claim 3: (3) implies (2).} Lastly, we will show the robustness of $f$ implies the inequalities of Statement (2). We will prove the contrapositive of this claim: if Statement (2) fails, then $f$ is not robust. Thus, we assume that there is an index $i\in \{1,\dots, m-2\}$ with $\alpha_{i+1}-\alpha_i<(\theta_{i+1}-\theta_{i})\max(\frac{\alpha_i}{\theta_{i}}, \frac{1-\alpha_i}{1-\theta_{i}})$. We proceed with a case distinction regarding $\max(\frac{\alpha_i}{\theta_{i}}, \frac{1-\alpha_i}{1-\theta_{i}})$.\medskip

\emph{Case 1:} First assume that $\frac{\alpha_i}{\theta_{i}}\geq \frac{1-\alpha_i}{1-\theta_{i}}$. Equivalently, this assumption means that $\alpha_i\geq \theta_i$, so it follows that $\alpha_i>0$. We define $\delta=(\theta_{i+1}-\theta_{i})\cdot \frac{\alpha_i}{\theta_{i}}-(\alpha_{i+1}-\alpha_i)$ and we choose a value $\epsilon>0$ such that $\delta\frac{\theta_i}{\alpha_i}>\epsilon \alpha_{i+1}$. Moreover, let $v\in \mathbb{Q}\cap (0,1]$ denote a rational value such that $\frac{\theta_i}{\alpha_i}\leq v\leq \frac{\theta_i}{\alpha_i} + \epsilon$. Finally, we let $w_1,w_2\in\mathbb{N}_0$ denote two integers such that $v=\frac{w_1}{w_1+w_2}$ and consider the interval profile $\mathcal{I}$ where $w_1$ voters report $\{x_i, x_{i+1}, x_{i+2}\}$ and $w_2$ voters report $\{x_m\}$. It can be computed that 
\begin{align*}
    \Pi_\alpha(\mathcal I, x_i)&=w_1 \cdot \alpha_i = v \cdot n_\mathcal{I}\cdot  \alpha_i\geq \frac{\theta_i}{\alpha_i}\cdot n_\mathcal{I}\cdot \alpha_i= \theta_i n_\mathcal{I}.
\end{align*}

Since $\Pi_\alpha(\mathcal I,x_h)=0$ for all $x_h$ with $x_h\ssucc x_i$, we conclude that $f(\mathcal I)=x_i$. 

Next, let $\mathcal I'$ denote the profile where $w_1$ voters report $\{x_{i+1}, x_{i+2}\}$ and $w_2$ voters report $\{x_m\}$. Robustness from $\mathcal I$ to $\mathcal I'$ postulates that $f(\mathcal I')\in \{x_i, x_{i+1}\}$. Moreover, it holds that $\Pi_\alpha(\mathcal I', x_i)=0$ as no voter approves an alternative left of $x_{i+1}$ in $\mathcal I'$, so $f(\mathcal I')\neq x_i$.
However, the subsequent computations show that $\Pi_\alpha(\mathcal I',x_{i+1})=\Pi_\alpha(\mathcal I,x_{i+1})<\theta_{i+1}n_\mathcal{I}$. This means that $f(\mathcal I')\neq x_{i+1}$ and robustness is violated.
\begin{align*}
    \Pi_\alpha(\mathcal I', x_{i+1})
    &= v \cdot n_\mathcal{I}\cdot  \alpha_{i+1}\\
    &\leq \left (\frac{\theta_i}{\alpha_i}+\epsilon \right)\cdot n_\mathcal{I} \cdot \alpha_{i+1}\\
    &= \theta_i \cdot n_\mathcal{I}+(\alpha_{i+1}-\alpha_i)\cdot\frac{\theta_i n_\mathcal{I}}{\alpha_i}+\epsilon \cdot n_\mathcal{I}\cdot \alpha_{i+1}\\
    &= \theta_i \cdot n_\mathcal{I}+\left((\theta_{i+1}-\theta_i)\cdot \frac{\alpha_i}{\theta_i}-\delta\right)\cdot\frac{\theta_i n_\mathcal{I}}{\alpha_i}+\epsilon \cdot n_\mathcal{I}\cdot \alpha_{i+1}\\
    &=n_\mathcal{I}\left(\theta_i+\theta_{i+1}-\theta_i-\delta\frac{\theta_i}{\alpha_i}+\epsilon\cdot\alpha_{i+1}\right)\\
    &<\theta_{i+1}n_\mathcal{I}.
\end{align*}

Here, the second line uses the definition of $v$, the third one rearranges the terms, and the fourth one applies the definition of $\delta$. The fifth line is again simple calculus, and the last inequality follows because $\delta\frac{\theta_i}{\alpha_i}>\epsilon\alpha_{i+1}$. 
\medskip

\emph{Case 2:} As the second case, we assume that $\frac{\alpha_i}{\theta_{i}}< \frac{1- \alpha_i}{1-\theta_{i}}$. This is equivalent to $\theta_i>\alpha_i$, so we derive that $0<1-\theta_i<1-\alpha_i$. Next, we define $\delta=(\theta_{i+1}-\theta_{i})\cdot \frac{1-\alpha_i}{1-\theta_i}-(\alpha_{i+1}-\alpha_i)$ and we choose $\epsilon>0$ such that $\delta\frac{1-\theta_i}{1-\alpha_i}>\epsilon(1-\alpha_{i+1})$. Moreover, we observe that $\frac{1-\theta_i}{1-\alpha_i}>0$ since $0<\theta_i<1$ and $1-\alpha_i>0$, and that $\frac{1-\theta_i}{1-\alpha_i}<1$ since $0<1-\theta_i<1-\alpha_i$. Hence, there is a rational value $v\in\mathbb{Q}\cap (0,1)$ with $\frac{1-\theta_i}{1- \alpha_i}-\epsilon\leq v\leq\frac{1-\theta_i}{1-\alpha_i}$. Finally, let $w_1,w_2\in\mathbb{N}$ denote two integers such that $v=\frac{w_1}{w_1+w_2}$ and consider the profile $\mathcal I$ where $w_1$ voters report $\{x_i, x_{i+1}, x_{i+2}\}$ and $w_2$ voters report $\{x_i\}$. We first compute that 
\begin{align*}
    \Pi_\alpha(\mathcal I, x_i)
    =w_2+w_1\alpha_i
    =n_\mathcal{I}(1-(1-\alpha_i)v)
    \geq n_\mathcal{I}\left(1-(1-\alpha_i)\frac{1-\theta_i}{1-\alpha_i}\right)
    =\theta_i n_\mathcal{I}.
\end{align*}

Here, the inequity in the third step follows because $v\leq\frac{1-\theta_i}{1-\alpha_i}$ and $(1-\alpha_i)> 0$. Since no voter reports an alternative left of $x_i$, this shows that $f(\mathcal I)=x_i$. 

Next, let $\mathcal I'$ denote the profile where $w_1$ voters report $\{x_{i+1}, x_{i+2}\}$ and $w_2$ voters report $\{x_{i+1}\}$. Repeatedly applying robustness from $\mathcal I$ to $\mathcal I'$ shows that $f(\mathcal I')\in \{x_i, x_{i+1}\}$. In particular, for the voters deviating from $\{x_i\}$ to $\{x_{i+1}\}$, we can make an intermediate step by expanding their intervals to $\{x_i, x_{i+1}\}$. On the other hand, we derive that $f(\mathcal I')\neq x_i$ because no voter reports an alternative $x_h$ with $x_h\ssucceq x_i$. Finally, we show next that $\Pi_\alpha(\mathcal I', x_{i+1})<\theta_{i+1}n_\mathcal{I'}$, which proves that $f(\mathcal I)\neq x_{i+1}$. Thus, $f$ fails robustness. 
\begin{align*}
    \Pi_\alpha(\mathcal I', x_{i+1})
    &=w_2+w_1\alpha_{i+1}\\
    &=n_\mathcal{I}(1-(1-\alpha_{i+1})v)\\
    &\leq n_\mathcal{I}\left(1-(1-\alpha_{i+1}) \left(\frac{1-\theta_i}{1-\alpha_i}-\epsilon\right)\right)\\
    &= n_\mathcal{I}\left(1-(1-\alpha_{i})\cdot \frac{1-\theta_i}{1-\alpha_i}+(\alpha_{i+1}-\alpha_i)\cdot \frac{1-\theta_i}{1-\alpha_i} + (1-\alpha_{i+1})\epsilon \right)\\
    &=n_\mathcal{I}\left(\theta_i+\left((\theta_{i+1}-\theta_i)\cdot\frac{1-\alpha_i}{1-\theta_i}-\delta\right)\cdot \frac{1-\theta_i}{1-\alpha_i} + (1-\alpha_{i+1})\epsilon \right)\\
    &=n_\mathcal{I}\left(\theta_{i+1}-\delta\cdot \frac{1-\theta_i}{1-\alpha_i}+(1-\alpha_{i+1})\epsilon\right)\\
    &<\theta_{i+1}n_\mathcal{I}.
\end{align*}

The first line uses the definition of $\Pi_\alpha$, the second the definition of $w_1$ and $w_2$, and the third inequality that $\frac{1-\theta_i}{1- \alpha_i}-\epsilon\leq v$. Next, we rearrange our formula and substitute the definition of $\delta$ in the fifth line. The remaining two lines follow from simple calculus and the definition of $\epsilon$. 
\end{proof}

Based on \Cref{lem:robust}, we will now show that position-threshold rules indeed satisfy all required axioms. 

\begin{lemma}\label{lem:allaxioms}
    Every position-threshold rule satisfies anonymity, unanimity, robustness, reinforcement, and right-biased continuity. 
\end{lemma}
\begin{proof}
     Fix a threshold vector $\theta\in(0,1)^m$ and a compatible weight vector $\alpha\in [0,1]^m$ and let $f$ denote the position-threshold rule induced by these vectors. We first note that $\Pi_\alpha$ is anonymous, so $f$ also satisfies this property. Moreover, if $I_i=\{x_j\}$ for all voters in a profile $\mathcal{I}$, then $\Pi_\alpha(\mathcal I, x_j)=n_\mathcal{I}\geq \theta_j n_\mathcal I$ and $\Pi_\alpha(\mathcal I, x_h)=0<\theta_hn_\mathcal{I}$ for all $x_h\ssucc x_j$. This means that $f(\mathcal{I})=x_j$, so $f$ is unanimous. Next, \Cref{lem:robust} shows that $f$ is robust since $\alpha$ and $\theta$ are compatible.
    
    As the fourth axiom, we will prove that $f$ is reinforcing. For this, let $\mathcal I^1$ and $\mathcal I^2$ denote two profiles in $\dom$ such that $f(\mathcal I^1)=f(\mathcal I^2)=x_i$ for some alternative $x_i\in A$ and $N_{\mathcal I^1}\cap N_{\mathcal I^2}=\emptyset$. By definition of $f$, 
    it holds for $\mathcal{I}\in \{\mathcal{I}^1, \mathcal{I}^2\}$ that $\Pi_\alpha(\mathcal{I},x_i)\geq \theta_in_\mathcal{I}$ 
    and $\Pi_\alpha(\mathcal{I},x_h)< \theta_{h}n_\mathcal{I}$ for all $x_h\in A$ with $x_h\ssucc x_i$. 
    Moreover, we have for all $x\in A$ that $\Pi_\alpha(\mathcal{I}^1+\mathcal{I}^2,x)=\Pi_\alpha(\mathcal{I}^1,x)+\Pi_{\alpha}(\mathcal{I}^2,x)$. 
    Hence, it follows that $\Pi_\alpha(\mathcal{I}^1+\mathcal{I}^2,x_h)<\theta_{h}n_{\mathcal{I}^1}+\theta_{h}n_{\mathcal{I}^2}=\theta_{h}n_{\mathcal{I}^1+\mathcal{I}^2}$ for all $x_h$ with $x_h\ssucc x_i$ and $\Pi_\alpha(\mathcal{I}^1+\mathcal{I}^2,x_i)\geq \theta_{i}n_{\mathcal{I}^1}+\theta_{i}n_{\mathcal{I}^2}=\theta_{i}n_{\mathcal{I}^1+\mathcal{I}^2}$. This means that $f(\mathcal{I}^1+\mathcal{I}^2)=x_i$ and $f$ thus is reinforcing.

    Finally, we will prove that $f$ satisfies right-biased continuity. For this, we consider two profiles $\mathcal{I}^1, \mathcal{I}^2\in\dom$. First, if $f(\mathcal{I}^2)= f(\mathcal{I}^1)$, it follows by reinforcement that $f(\mathcal{I}^1+\mathcal{I}^2)=f(\mathcal{I}^1)$ and right-biased continuity is satisfied. Next, we assume that $f(\mathcal{I}^2)\ssucc f(\mathcal{I}^1)$ and we will show that there is $\lambda\in\mathbb{N}$ such that $f(\lambda \mathcal{I}^1+\mathcal{I}^2)=f(\mathcal{I}^1)$. By the definition of $f$, we derive that $\Pi_\alpha(\mathcal{I}^1,x_h)<\theta_{h}n_{\mathcal{I}^1}$ for all $x_h\in A$ with $x_h\ssucc x_i$ and $\Pi_\alpha(\mathcal{I}^1,x_i)\geq\theta_i|N_{\mathcal{I}^1}|$. 
    Moreover, it holds that $\Pi_\alpha(\mathcal{I}^2, x_j)\geq \theta_jn_{\mathcal{I}^2}$ for the alternative $x_j=f(\mathcal{I}^2)$. Next, because $\alpha$ and $\theta$ are compatible, $\Pi_\alpha(\mathcal I, x_h)\geq \theta_hn_\mathcal{I}$ implies $\Pi_\alpha(\mathcal I, x_{h+1})\geq \theta_{h+1}n_\mathcal{I}$ for all interval profiles $\mathcal I$ and all alternatives $x_h\in A\setminus \{x_m\}$. Based on this insight, we infer that $\Pi_\alpha(\mathcal{I}^2, x_i)\geq \theta_in_{\mathcal{I}^2}$ as $x_j\ssucc x_i$. 
    Now, let $\delta_h=\theta_{h}n_{\mathcal{I}^1}-\Pi_\alpha(\mathcal{I}^1,x_{h})$ for all $h\in \{1,\dots, i-1\}$ and let $\lambda\in\mathbb{N}$ denote an integer such that $\lambda\delta_h>\Pi_\alpha(\mathcal{I}^2, x_h)$ for all such $h$. 
    By the definition of $\lambda$, we derive for all $x_h$ with $x_h\ssucc x_{i}$ that
    \begin{align*}
        \Pi_\alpha(\lambda \mathcal{I}^1 + \mathcal{I}^2, x_h)
        &= \lambda\Pi_\alpha(\mathcal{I}^1, x_h) + \Pi_\alpha(\mathcal{I}^2, x_h)\\
        &=\lambda(\theta_h n_{\mathcal{I}^1}-\delta_h) + \Pi_\alpha(\mathcal{I}^2, x_h)\\
        &<\theta_h n_{\lambda \mathcal{I}^1+\mathcal{I}^2}
    \end{align*}

    This implies that $x_i\ssucceq f(\lambda \mathcal{I}^1 + \mathcal{I}^2)$. On the other hand, it is holds that
    \begin{align*}
        \Pi_\alpha(\lambda \mathcal{I}^1 + \mathcal{I}^2, x_i)&=\lambda\Pi_\alpha(\mathcal{I}^1,x_i)+\Pi_\alpha(\mathcal{I}^2,x_i)\\
        &\geq \lambda \theta_i n_{\mathcal{I}^1} + \theta_i n_{\mathcal{I}^2}\\
        &=\theta_in_{\lambda \mathcal{I}^1+\mathcal{I}^2}. 
    \end{align*}

    Hence, $f(\lambda \mathcal{I}^1 + \mathcal{I}^2)=x_i$ and right-biased continuity is satisfied. 
    
    For the second case, suppose that $f(\mathcal{I}^1)\ssucc f(\mathcal{I}^2)$ and let $x_r$ denote the right-most alternative that is reported by some voter in $\mathcal{I}^1$. We moreover let $x_i=f(\mathcal{I}^1)$, and we will show that there is $\lambda\in\mathbb{N}$ such that $x_i\ssucceq f(\lambda \mathcal{I}^1 + \mathcal{I}^2)\ssucceq x_r$. To this end, we first observe that $\Pi_\alpha(\mathcal{I}, x_j)<\theta_j n_{\mathcal I}$ for all $j<i$ and $\mathcal{I}\in \{\mathcal{I}^1, \mathcal{I}^2\}$, so analogous arguments as before show that $x_i\ssucceq f(\lambda \mathcal{I}^1 + \mathcal{I}^2)$ for all $\lambda\in\mathbb{N}$. Next, we choose $\lambda$ such that $\theta_r\leq \frac{\lambda n_{\mathcal{I}^1}}{\lambda n_{\mathcal{I}^1}+n_{\mathcal{I}^2}}$. We note that such a $\lambda$ exists as $\theta_r<1$ and $\frac{\lambda n_{\mathcal{I}^1}}{\lambda n_{\mathcal{I}^1}+n_{\mathcal{I}^2}}$ converges to $1$ as $\lambda$ increases. By the choice of $x_r$, we have that $\Pi_{\alpha}(\mathcal{I}^1, x_r)=n_{\mathcal{I}^1}$.
    Hence, we compute that     
    \begin{align*}
        \Pi_\alpha(\lambda\mathcal{I}^1 + \mathcal{I}^2, x_r) & = \lambda \Pi_\alpha(\mathcal{I}^1, x_r) + \Pi_\alpha(\mathcal{I}^2, x_r)\\
        &\geq \lambda n_{\mathcal{I}^1}\\
        &\geq \theta_r (\lambda n_{\mathcal{I}^1} + n_{\mathcal{I}^2})\\
        &= \theta_r n_{\lambda \mathcal{I}^1+\mathcal{I}^2}.
    \end{align*}

    This proves that $f(\lambda\mathcal{I}^1 + \mathcal{I}^2)\ssucceq x_r$ and thus completes the proof that $f$ satisfies right-biased continuity.
\end{proof}

\subsection{Derivation of Weight and Threshold Vectors}\label{subsec:characterizationproof}

We will next show that every voting rule on $\dom$ that satisfies anonymity, unanimity, robustness, reinforcement, and right-biased continuity is a position-threshold rule. To this end, we suppose throughout this section that $f$ is a voting rule on $\dom$ that satisfies all considered axioms, and we aim to represent $f$ as a position-threshold rule by deriving its weight and threshold vectors.

As a first step, we will show that $f$ coincides with a phantom median rule on the domain $\mathcal{D}_1^N$ where all voters of a fixed electorate $N$ report a single alternative. More formally, the domain $\mathcal{D}_1^N$ is the subset of $\Lambda^N$ given by $\mathcal{D}_1^N=\{\mathcal I\in \Lambda^N\colon \forall i\in N\colon |I_i|=1\}$. 
We will prove our claim by showing that $f$ induces a voting rule on the domain of single-peaked preferences $\mathcal{P}_{\ssucc}^N$ that satisfies anonymity, unanimity, and strategyproofness. By the characterization of \citet{Moul80a}, we then infer that $f$ is a phantom median rule for $\mathcal{P}^N_{\ssucc}$ (see also \citet{BoJo83a} or \citet{Weym11a} as Moulin uses slightly stronger axioms than we do), which will imply the desired representation of $f$ on $\mathcal{D}_1^N$. To make our proof precise, we next present the definitions of anonymity, unanimity, and strategyproofness for the domain of single-peaked preferences $\mathcal{P}_{\ssucc}^N$. We say that a voting rule $f$ on $\mathcal{P}_{\ssucc}^N$ is
\begin{itemize}[leftmargin=*, itemsep=2pt]
    \item \emph{anonymous} if $f(\pi(R))=f(R)$ for all preference profiles $\mathcal{P}_{\ssucc}^N$ and permutations $N\rightarrow N$. 
    \item \emph{unanimous} if $f(R)=x_i$ for all preference profiles $R\in \mathcal{P}_{\ssucc}^N$ and alternatives $x_i\in A$ such that all voters in $R$ report $x_i$ as their favorite alternative.
    \item \emph{strategyproof} if $f(R)\succsim f(R')$ for all profiles $R,R'\in\mathcal{R}_{\ssucc}^N$ and voters $i\in N$ such that ${\succsim_j}={\succsim_j'}$ for all $j\in N\setminus \{i\}$. 
\end{itemize}

Then, the characterization of \citet{Moul80a} states that a voting rule $f$ on $\mathcal{P}_{\ssucc}^N$ satisfies anonymity, unanimity, and strategyproofness if and only if it is a phantom median rule, i.e., there is a threshold vector $\theta\in (0,1)^m$ such that $\theta_1\geq \dots \geq \theta_m$ and $f(R)=\max_{\ssucc}\{x_i\in A\colon \Pi_\mathit{SP}(R, x_i)\geq\theta_i |N_R|\}$ for all profiles $R\in\mathcal{P}_{\ssucc}^N$.\footnote{
The standard way to state Moulin's result is that a voting rule $f$ on $\mathcal{P}_{\ssucc}^N$ satisfies anonymity, unanimity, and strategyproofness if and only if there are $|N|-1$ phantom voters who report fixed single-peaked preference relations and that $f$ chooses the top-ranked alternative of the median voter with respect to our $|N|$ original voters and the $|N|-1$ phantom voters. To arrive at our representation, we define $p_i$ as the number of phantom voters that top-rank alternative $x_i$. It can be checked that, for all profiles $R\in\mathcal{P}_{\ssucc}^N$, it holds that $f(R)=\max_{\ssucc}\{x_i\in A\colon \Pi_\mathit{SP}(R,x_i)\geq \theta_i |N_R|\}$ for the threshold vector $\theta$ given by $\theta_k=\frac{1+2\sum_{i=k+1}^m p_i}{2|N|}$ for all $k\in \{1,\dots, m\}$.
}
We are now ready to show our first lemma. For this lemma, we extend the definition of the individual and collective peak position functions to interval profiles $\mathcal I\in \mathcal{D}_1^N$ by defining $\pi_\mathit{SP}(\{x_i\}, x_j)=1$ if $x_i\ssucceq x_j$ and $\pi_\mathit{SP}(\{x_i\}, x_j)=0$ if $x_j\ssucc x_i$ for all $x_i, x_j\in A$, and $\Pi_\mathit{SP}(\mathcal I, x_j)=\sum_{i\in N_{\mathcal I}} \pi_\mathit{SP}(I_i, x_j)$ for all $x_j\in A$ and $\mathcal I\in \mathcal{D}_1^N$. 

\begin{lemma}\label{lem:MoulinD1}
    For every electorate $N\in\mathcal{F}(\mathbb{N})$, there is a threshold vector $\theta=(\theta_1,\dots, \theta_m)\in (0,1)^m$ such that $\theta_1\geq \dots \geq \theta_m$ and $f(\mathcal I)=\max_{\ssucc}\{x_i\in A\colon \Pi_\mathit{SP}(\mathcal I, x_i)\geq \theta_i n_{\mathcal I}\}$ for all interval profiles $\mathcal I\in\mathcal{D}_{1}^N$. 
\end{lemma}
\begin{proof}
    We will show the lemma by reducing $f$ to a voting rule $f'$ on $\mathcal{P}_{\ssucc}^N$ that is strategyproof, anonymous, and unanimous. In turn, the characterization of \citet{Moul80a} shows that there is a vector $\theta\in (0,1)^m$ such that $\theta_1\geq \dots \geq \theta_m$ and $f'(R)=\max_{\ssucc}\{x_i\in A\colon \Pi_\mathit{SP}(R, x_i)\geq \theta_i|N_R|\}$ for all profiles $R\in\mathcal{P}_{\ssucc}^N$. 
    We thus define the interval profile $\mathcal{I}(R)$ given a single-peaked profile $R\in\mathcal{P}_{\ssucc}^N$ by ${I}(R)_i=T(\succsim_i)$ for all $i\in N$, i.e., the interval of each voter only contains his favorite alternative in $R$. Next, we set $f'(R)=f(\mathcal{I}(R))$ for all profiles $R\in\mathcal{P}_{\ssucc}^N$. We first note that $f'$ is anonymous and unanimous as $f$ satisfies these axioms. 
    
    We hence focus on showing that $f'$ is strategyproof. For this, we consider two preference profiles $R,R'\in\mathcal{P}^N_{\ssucc}$ and a voter $i\in N$ such that ${\succsim_j}={\succsim_j'}$ for all $j\in N\setminus \{i\}$. We will show that $f'(R)\succsim_i f'(R')$. To this end, let $x_i$ denote voter $i$'s favorite alternative in $R$ and $x_i'$ denote his favorite alternative in $R'$. First, if $f'(R)= x_i$, voter $i$ cannot manipulate as his favorite alternative is chosen in $R$. Without loss of generality, we will hence assume that $x_i\ssucc f'(R)$. Now, if $x_i'\ssucceq x_i$, it follows from the robustness of $f$ that $f'(R)=f'(R')$. In more detail, let $\mathcal I^*\in\Lambda^N$ denote the interval profile such that $I_i^*=[x_i', x_i]$ and $I_j^*=I(R)_j$ for all $j\in N\setminus \{i\}$. We can transform $\mathcal{I}(R)$ to $\mathcal I^*$ by one after another adding alternatives left of $x_i$ to voter $i$'s interval. Since $x_i\ssucc f'(R)=f(\mathcal I(R))$, robustness implies for all of these steps that the outcome does not change. Hence, we have that $f(\mathcal I^*)=f(\mathcal I(R))$. Finally, we can transform $\mathcal I^*$ to $\mathcal I(R')$ by one after another deleting the alternatives right of $x_i'$ from voter $i$'s interval. Since $f(\mathcal I^*)\not\in I_i^*$, robustness implies that the outcome is again not allowed to change, so we conclude that $f'(R)=f(\mathcal I^*)=f'(R')$. 
    
    Next, if $x_i\ssucc x_i'\ssucceq f'(R)$, we can use an analogous argument based on the interval $I_i^*=[x_i, x_i']$ as robustness still implies that the winner is not allowed to change. 
    Finally, if $f'(R)\ssucc x_i'$, it follows from robustness that $f'(R)\ssucceq f'(R')$. To see this, we first consider the profile $\mathcal I^1$ where voter $i$ reports $[x_i, f'(R)]$ and all other voters $j\in N\setminus \{i\}$ report $I(R)_j$. Repeatedly applying robustness shows that $f(\mathcal I(R))=f(\mathcal I^1)$. Next, let $\mathcal I^2$ denote the profile derived from $\mathcal I^1$ by assigning voter $i$ the interval $[x_i, x_i']$. We can transform $\mathcal I^1$ to $\mathcal I^2$ by adding one after another the alternatives right of $f'(R)$ to voter $i$'s interval. Robustness implies that the winner can only move to the right, i.e., that $f(\mathcal I^1)\ssucceq f(\mathcal I^2)$. 
    As the last step, we transform $\mathcal I^2$ to $\mathcal I(R')$ by one after another deleting alternatives left of $x_i'$ from voter $i$'s interval. Robustness implies for such actions again that the winner can only move to the right, so $f(\mathcal I^2)\ssucceq f(\mathcal I(R'))$. By chaining these insights and using the definition of $f'$, it follows that $f'(R)\ssucceq f'(R')$. Finally, since $x_i\ssucc f'(R)\ssucceq f'(R')$, the single-peakedness of $\succsim_i$ implies that $f'(R)\succsim_i f'(R')$, so $f'$ is indeed strategyproof.
    
    Moulin's characterization now shows that there is a threshold vector $\theta\in (0,1)^m$ such that $\theta_1\geq \dots \geq \theta_m$ and $f'(R)=\max_{\ssucc}\{x_i\in A\colon \Pi_\mathit{SP}(R, x_i)\geq \theta_i |N_R|\}$ for all profiles $R\in\mathcal{P}^N_{\ssucc}$. Due to the relation between $f$ and $f'$, this further implies that $f(\mathcal I)=\max_{\ssucc}\{x_i\in A\colon \Pi_\mathit{SP}(\mathcal I, x_i)\geq \theta_i n_{\mathcal I}\}$ for all interval profiles $\mathcal I\in\mathcal{D}_{1}^N$.
\end{proof}

Next, we will generalize \Cref{lem:MoulinD1} from a fixed electorate $N$ to the domain of all electorates. To this end, we set $\mathcal{D}_1^*=\bigcup_{N\in\mathcal{F}(\mathbb{N})} \mathcal{D}_1^N$. 

\begin{lemma}\label{lem:generalizedmedian}
There is a threshold vector $\theta\in (0,1)^m$ such that $\theta_1\geq \dots \geq \theta_m$ and $f(\mathcal I)=\max_{\ssucc}\{x_i\in A\colon \Pi_\mathit{SP}(\mathcal I, x_i)\geq \theta_i n_{\mathcal I}\}$ for all interval profiles $\mathcal I\in\mathcal{D}_1^*$. 
\end{lemma}
\begin{proof}
    To prove this lemma, we denote by $N(z)$ an arbitrary electorate with $z$ voters. Because of the anonymity of $f$, the choice of $N(z)$ does not matter. By \Cref{lem:MoulinD1}, there is for every electorate $N(z)$ a threshold vector $\theta^z\in(0,1)^m$ such that $\theta_1^z\geq \dots \geq \theta_m^z$ and $f(\mathcal I)=\max_{\ssucc}\{x_i\in A\colon \Pi_\mathit{SP}(\mathcal I, x_i)\geq \theta_i^z n_{\mathcal I}\}$ for all profiles $\mathcal I\in\mathcal{D}_1^{N(z)}$. 
    We note, however, that these vectors are not unique: instead of $\theta^z$, we can represent $f$ on $\mathcal{D}_{1}^{N(z)}$ by every vector $q\in (0,1]^m$ such that $q_m=q_{m-1}$ and $q_i\in (\frac{v^z_i}{z},\frac{v^z_i+1}{z}]$ for all $i\in \{1\dots, m-1\}$, where $v^z_i\in\mathbb{N}_0$ is chosen such that $\frac{v^z_i}{z}<\theta_i^z\leq\frac{v^z_i+1}{z}$. The reason for this is that $\Pi_\mathit{SP}(\mathcal I,x_i)\in \{0,\dots, z\}$ for every alternative $x_i\in A$ and profile $\mathcal{I}\in \mathcal{D}_1^{N(z)}$. Also, the exact choice of $q_m$ has no influence as $\Pi_\mathit{SP}(\mathcal I, x_m)=n_{\mathcal I}$ for all $\mathcal I\in\mathcal{D}_1^{N(z)}$. We hence define the interval $I^z_i=(\frac{v^z_i}{z},\frac{v^z_i+1}{z}]$ for all $z\in\mathbb{N}$ and $i\in \{1,\dots, m-1\}$ and emphasize that $f(\mathcal I)=\max_{\ssucc} \{x_i\in A\colon \Pi_\mathit{SP}(\mathcal I, x_i)\geq q_i n_\mathcal{I}\}$ for all $\mathcal I\in \mathcal{D}_1^{N(z)}$  and every vector $q$ with $q_i\in I_i^z$ for all $i\in \{1,\dots, m-1\}$ and $q_m=q_{m-1}$. 

    We next define $\bar I^z_i=\bigcap_{s\in \{1,\dots, z\}} I^s_i$ as the intersection of the first $z$ intervals $I_i^s$ for some alternative $x_i\in A\setminus \{x_m\}$, and we will show that $\bar I^z_i\neq\emptyset$ for all $z\in\mathbb{N}$. Assume for contradiction that there are $z\in\mathbb{N}$ and $i\in \{1,\dots, m-1\}$ such that $\bar I_i^z=\emptyset$, and moreover suppose that $z$ is chosen minimal, i.e., $\bar I^{z-1}_i\neq\emptyset$ but $\bar I_i^z=\emptyset$. Since $\bar I^{z-1}_i$ is the non-empty intersection of intervals that are all closed to the right, it is itself an interval that is closed to the right. Next, we denote every interval $I^s_i$ by $I^s_i=(\ell^s_i, r^s_i]$ and define $\bar r^{z-1}_i=\min_{s\in \{1,\dots, z-1\}} r^s_i$ and $\bar \ell^{z-1}_i=\max_{s\in \{1,\dots, z-1\}} \ell^s_i$. It holds that $\bar I_i^{z-1}=(\bar \ell^{z-1}_i,\bar r^{z-1}_i]$ since every point left of $\bar \ell^{z-1}_i$ and right of $\bar r^{z-1}_i$ is not included in some interval $I^s_i$. Because $\bar I_i^{z}=\emptyset$, it either holds that $r^z_i\leq \bar \ell^{z-1}_i$ or $\bar r^{z-1}_i\leq\ell^z_i$. We subsequently assume that $r^z_i\leq \bar \ell^{z-1}_i$ as both cases are symmetric. Now, let $s\in \{1,\dots, z-1\}$ denote the index of an interval $I^s_i$ with $r^s_i=\bar r^{z-1}_i$, which means that $\bar r^{z-1}_i=\frac{c}{s}$ for some $c\in \{1,\dots, s\}$. 
    
    We consider the profile $\mathcal I$ where $s\cdot \bar r^{z-1}_i$ voters report $x_i$ and $s\cdot (1-\bar r^{z-1}_i)$ voters report $x_{i+1}$. By the definition of $I^s_i$, we have that $f(\mathcal{I})=\max_{\ssucc}\{x_j\in A\colon \Pi_\mathit{SP}(\mathcal I, x_j)\geq \theta^s_j \cdot s\}=x_i$ since $\Pi_\mathit{SP}(\mathcal{I},x_j)=0$ for all $x_j\in A$ with $x_j\ssucc x_i$ and $\Pi_\mathit{SP}(\mathcal{I},x_i)=s\cdot \bar r^{z-1}_i\geq \theta^s_i n_{\mathcal I}$. 
    Moreover, by reinforcement, it also holds that $f(z\mathcal{I})=x_i$ for the profile $z\mathcal{I}$ that consists of $z$ copies of $\mathcal{I}$. Next, consider the profile $\mathcal{I}'$, which consists of $z\cdot \ell^z_i$ voters reporting $x_i$ and $z\cdot (1-\ell^z_i)$ voters reporting $x_{i+1}$ (note that $z\cdot \ell^z_i$ and $z\cdot (1-\ell^z_i)$ are integers since $\ell^z_i=\frac{c}{z}$ for some $c\in \{0,\dots,z-1\}$). Using the definition of $\ell^z_i$, it follows that $f(\mathcal{I}')=\max_{\ssucc}\{x_j\in A\colon \Pi_\mathit{SP}(\mathcal I', x_j)\geq \theta^z_j \cdot z\}=x_{i+1}$ since $\Pi(\mathcal{I}',x_i)=z\cdot \ell^z_i<\theta^z_i\cdot n_{\mathcal I}$ and $\Pi(\mathcal{I}',x_i)=z\geq \theta^z_{i+1} n_{\mathcal I}$. By reinforcement, it follows for the profile $s\mathcal{I}'$, which consists of $s$ copies of $\mathcal{I}'$, that $f(s\mathcal{I}')=x_{i+1}$. Finally, there is also a threshold vector $\theta^{sz}$ such that $f(\mathcal{\hat{I}})=\max_{\ssucc} \{x_j\in A\colon \Pi_\mathit{SP}(\mathcal{\hat{I}}, x_j)\geq \theta^{sz}_j \cdot s\cdot z\}$ for all profiles $\mathcal{I}\in\mathcal{D}_{1}^{N({sz})}$. Now, since $f(s\mathcal{I}')=x_{i+1}$ and there are $z\cdot s\cdot \ell^z_i$ voters reporting $x_i$ in $z\mathcal{I}$, we infer that $\theta^{sz}_i> \frac{\Pi_\mathit{SP}(s\mathcal{I}',x_i)}{sz}=\ell^z_i$. Analogously, it holds that $\theta^{sz}_i\leq \bar r^{z-1}_i$ since $s\cdot z\cdot \bar r^{z-1}_i$ voters report $x_i$ in $z\mathcal{I}$ and $f(z\mathcal{I})=x_{i}$. However, this means that $\theta_i^{sz}\leq \bar r^{z-1}_i\leq \ell_i^s<\theta_i^{sz}$. This contradiction proves that our assumption that $\bar I^z_i=\emptyset$ is wrong. 

    We now define the threshold vector $\theta$. To this end, we observe that $\bar I^{z+1}_i\subseteq \bar I^z_i$ for all $z\in\mathbb{N}$ and $i\in \{1,\dots, m-1\}$. Finally, since $\bar r_i^z-\bar \ell_i^z\leq r^z-\ell^z=\frac{1}{z}$ for all $z\in\mathbb{N}$, the series $\bar r^1_i, \bar r^2_i, \dots$ (i.e., the right endpoints of the intervals $\bar I^1_i, \bar I^2_i,\dots$) is guaranteed to converge for all $i\in \{1,\dots, m-1\}$. We hence define the vector $\theta$ by $\theta_i=\lim_{z\rightarrow \infty} \bar r^z_i$ for all $i\in \{1,\dots, m-1\}$ and $\theta_m=\theta_{m-1}$. 
    We will first show that $\theta_{i}\geq \theta_{i+1}$ for all $i\in \{1,\dots, m-1\}$. For $i=m-1$, this is clear from the definition. For $i<m-1$, we have by definition that $\theta_{i}^z\geq \theta_{i+1}^z$ for all $z\in\mathbb{N}$. This implies that $r^z_i\geq r^z_{i+1}$ for all $z\in\mathbb{N}$ and consequently also that $\bar r_i^z\geq \bar r_{i+1}^z$. This shows that $\theta_i=\lim_{z\rightarrow\infty} \bar r_i^z\geq \lim_{z\rightarrow\infty} \bar r_{i+1}^z=\theta_{i+1}$. 
    
    Next, we will prove that $\theta_i\in (0,1)$ for all $i\in \{1,\dots, m-1\}$. To this end, assume for contradiction that $\theta_i=0$ or $\theta_i=1$ for some $i\in \{1,\dots, m-1\}$. We first consider the case that $\theta_i=0$ for some alternative $x_i$ and we assume that $x_i$ is the alternative with minimal index such that $\theta_i=0$, i.e., $\theta_j>0$ for all $j\in \{1,\dots, i-1\}$. Since $\theta_i=0$, we infer that $\bar \ell^z_i=0$ for all $z\in\mathbb{N}$ because $\bar \ell^z_i< \bar r^{z}_i$ for all $z\in\mathbb{N}$. In particular, this means that $f(\mathcal I)=x_i$ for all profiles $\mathcal{I}\in\mathcal{D}_{1}^*$ where $\{x_i\}$ is reported by one voter and all other voters report $\{x_m\}$. Now, consider the profile $\mathcal I$ that consists of one voter reporting $\{x_m\}$, and the profile $\mathcal I'$ that consists of one voter reporting $\{x_i\}$. By unanimity, we have that $f(\mathcal I)=x_m$ and $f(\mathcal I')=x_i$. Hence, right-biased continuity requires that there is a $\lambda\in\mathbb{N}$ such that $f(\lambda \mathcal I+\mathcal I')=\{x_m\}$. However, this contradicts with our previous insight, so the assumption that $\theta_i=0$ must have been wrong. Next, consider the case that $\theta_i=1$. This is only possible if $\bar r^z_i=1$ for all $z\in\mathbb{N}$, so $f(\mathcal I)=x_i$ requires for all profiles $\mathcal{I}\in \mathcal{D}_1^*$ that no voter reports $\{x_m\}$. Now, consider again the profiles $\mathcal I$ and $\mathcal I'$ where one voter reports $\{x_m\}$ and one voter reports $\{x_i\}$, respectively. By right-biased continuity, there must be a $\lambda\in\mathbb{N}$ such that $f(\lambda \mathcal I'+\mathcal I)\ssucceq x_i$. However, this contradicts with the our previous observation, so we conclude that $\theta_i\neq 1$. 

    Finally, we will verify that $f(\mathcal I)=\max_{\ssucc} \{x_i\in A\colon \Pi_\mathit{SP}(\mathcal{I}, x_j)\geq \theta_j n_\mathcal{I}\}$ for all profiles $\mathcal{I}\in\mathcal{D}_{1}^*$. To this end, fix a set of voters $N(z)$. By the definition of $\theta$, it holds that $\theta_i\in [\bar \ell^z_i, \bar r^z_i]\subseteq [\ell^z_i, r^z_i]$ for all $i\in \{1,\dots, m-1\}$. If $\theta_i\in (\ell^z_i,r^z_i]$ for all $i\in \{1,\dots, m-1\}$, then $f(\mathcal I)=\max_{\ssucc}\{x_i\in A\colon \Pi_\mathit{SP}(\mathcal I, x_i)\geq \theta_i n_{\mathcal I}\}$ as all values in $(\ell_i^z, r^z_i]$ define the same rule. We hence will show that $\theta_i\neq \ell^z_i$ for all $i\in \{1,\dots, m-1\}$ and assume for contradiction that $\theta_i=\ell^z_i$ for some $i\in \{1,\dots, m-1\}$. We first note that this implies that $\ell^z_i\geq \ell^{s}_i$ for all $s\in \mathbb{N}$. Indeed, if there was an index $s$ such that $\ell^z_i< \ell^{s}_i$, it would be the case that $\ell^s_i\leq \bar \ell^t_i<\bar r^t_i$ for all $t\in\mathbb{N}$ with $t\geq s$. However, this means that $\theta_i\geq \ell^s_i>\ell^z_i$, which contradicts our assumption. 
    Now, consider the profile $\mathcal I$ with $z\cdot \ell^z_i$ voters reporting $\{x_i\}$ and $z\cdot (1-\ell^z_i)$ voters reporting $\{x_m\}$. By the definition of $\ell^z_i$, we have that $f(\mathcal I)\neq x_i$. Moreover, because $\theta_j^z> 0$ for all $j\in \{1,\dots, m\}$, we conclude that $x_i\ssucc f(\mathcal I)$. Next, let $\mathcal I'$ denote the profile where a single voter reports $\{x_1\}$. By unanimity, we have that $f(\mathcal I')=x_1$. Finally, by right-biased continuity, we infer that there is a value $\lambda\in\mathbb{N}$ such that $f(\lambda \mathcal I+\mathcal I')=f(\mathcal I)$. However, for each $\lambda\in\mathbb{N}$, it holds that $\Pi_\mathit{SP}(\lambda \mathcal I'+\mathcal I, x_i)=1+\lambda \cdot z\cdot \ell^z_i>(1+\lambda \cdot z) \ell^z_i\geq (1+\lambda \cdot z)\ell_i^{1+\lambda\cdot z}$. By the definition of $\ell^{1+\lambda z}_i$, this means that $f(\lambda\mathcal I'+\mathcal I)\ssucceq x_i$ for all $\lambda\in\mathbb{N}$, which contradicts right-biased continuity. Hence, we conclude that $\theta_i\neq \ell^z_i$, which completes the proof of this lemma. 
\end{proof}

As it will turn out, the threshold vector $\theta$ derived in \Cref{lem:generalizedmedian} is the threshold vector that defines $f$. We will hence focus next on the weight vector of $f$, for which we will rely on reinforcement. In more detail, to employ the full power of reinforcement, we will change the domain of $f$. To this end, we define $q=|\Lambda|$ as the number of intervals in $\Lambda$ and we arbitarily enumerate the intervals by ${I}^1,\dots, {I}^{q}$. This allows us to represent each interval profile $\mathcal I$ by a vector $v\in\mathbb{N}^q_0\setminus \{0\}$: the $i$-th entry of $v$ states how often the interval ${I}^i$ is reported. For the ease of notation, we write $v(\mathcal I)$ for the vector corresponding to the profile $\mathcal I$, and $v_I$ for the entry in $v$ corresponding to an interval $I$. Because $f$ is anonymous, there is a function $g$ from $\mathbb{N}^q\setminus \{0\}$ to $X$ such that $f(\mathcal I)=g(v(\mathcal I))$ for all profiles $\mathcal{I}\in\dom$. Moreover, $g$ inherits all desirable properties of $f$. We will next generalize $g$ to a function $\hat g:\mathbb{Q}^q_{\geq 0}\setminus \{0\}\rightarrow X$ while preserving the desirable properties of $f$. In particular, we will show that $\hat g$ extends $f$ (i.e., $f(\mathcal I)=\hat g(v(\mathcal I))$ for all profiles $\mathcal{I}\in\dom$) and satisfies reinforcement (i.e., $g(v+v')=g(v)$ for all $v,v'\in\mathbb{Q}^q_{\geq 0}\setminus \{0\}$ with $g(v)=g(v')$).

\begin{lemma}\label{lem:representation}
    There is a functions $\hat g:\mathbb{Q}^q_{\geq 0}\setminus \{0\}\rightarrow C$ that extends $f$ and satisfies reinforcement.
\end{lemma}
\begin{proof}
    Let $g$ denote the function from $\mathbb{N}^q\setminus \{0\}\rightarrow A$ that computes $f$ based on anonymized profiles. We first note that $g$ satisfies both conditions of the lemma as it is only a different representation of $f$. Next, we define the function $\hat g$ by $\hat g(v)=g(\lambda v)$ for all $v\in\mathbb{Q}^q_{\geq 0}\setminus \{0\}$, where $\lambda\in\mathbb{N}$ is an arbitrary scalar such that $\lambda v\in\mathbb{N}^q_0\setminus \{0\}$. 

    We first show that $\hat g$ is well-defined despite not fully specifying the parameter $\lambda$. To this end, let $v\in\mathbb{Q}^q_{\geq 0}$ denote an arbitrary vector and let $\lambda_1,\lambda_2\in\mathbb{N}$ denote two scalars such that $\lambda_1v,\lambda_2v\in\mathbb{N}^q_0\setminus \{0\}$. We will show that $g(\lambda_1v)=g(\lambda_2 v)$ as this implies that $\hat g$ is well-defined. For this, we note that reinforcement implies that $g(\lambda_1 v)=g(\lambda_1 \lambda_2 v)$ and that $g(\lambda_2 v)=g(\lambda_1 \lambda_2 v)$. Hence, $g(\lambda_1v)=g(\lambda_2 v)$ as desired. Moreover, this proves that $\hat g(v(\mathcal I))=g(1\cdot v(\mathcal I))=f(\mathcal I)$ for all profiles $\mathcal{I}\in\dom$, so $\hat g$ indeed extends $f$.

    Next, we show that $\hat g$ is reinforcing. To this end, consider two vectors $v^1,v^2\in\mathbb{Q}^q_{\geq 0}\setminus \{0\}$ and let $\lambda_1,\lambda_2\in\mathbb{N}$ denote scalars such that $\lambda_1v^1,\lambda_2v^2\in\mathbb{N}^q_0\setminus \{0\}$. We suppose that $\hat g(v^1)=\hat g(v^2)$ as there is otherwise nothing to show. By the definition of $\hat g$ and the reinforcement of $g$, it holds that $\hat g(v^1)=g(\lambda_1 v^1)=g(\lambda_1\lambda_2 v^1)$ and $\hat g(v^2)=g(\lambda_1 v^2)=g(\lambda_1\lambda_2 v^2)$. Because of the reinforcement of $g$, we next conclude that $\hat g(v^1+v^2)=g(\lambda_1\lambda_2 (v^1+v^2))=g(\lambda_1\lambda_2 v^1)=\hat g(v^1)$. This proves that $\hat g$ is reinforcing, too.
\end{proof}

Next, we define for every alternative $x_i\in A$ the set $Q_i=\{v\in\mathbb{Q}^q_{\geq 0}\setminus \{0\}\colon \hat g(v)=x_i\}$ as the subset of $\mathbb{Q}^q_{\geq 0}\setminus \{0\}$ such that $\hat g$ chooses $x_i$ for every point in $Q_i$. We note that $Q_i\cap Q_j=\emptyset$ for all $i\neq j$ as $\hat g$ returns for every point in $\mathbb{Q}^q_{\geq 0}\setminus \{0\}$ only a single alternative and that $\bigcup_{x_i\in A} Q_i=\mathbb{Q}^q_{\geq 0}\setminus \{0\}$. Moreover, $Q_i$ is $\mathbb{Q}$-convex (i.e., for all $v^1,v^2\in Q_i$ and $\lambda\in \mathbb{Q}\cap [0,1]$, it holds that $\lambda v^1+(1-\lambda) v^2\in Q_i$) because $\hat g$ is reinforcing. Next, we let $\bar Q_i$ denote the closure of $Q_i$ with respect to $\mathbb{R}^q$. Using standard arguments from \citet{Youn75a}, it can be shown that $\bar Q_i$ is convex for all $x_i\in A$ and that $\bigcup_{x_i\in A} \bar Q_i=\mathbb{R}^q_{\geq 0}$. We will next prove that the sets $\bar Q_i$ are polytopes. In the following, $uv$ will denote the standard scalar product between two vectors $u, v\in\mathbb{R}^q$, i.e., $uv=\sum_{i=1}^q u_iv_i$. 

\begin{lemma}\label{lem:hyperplane}
    For every alternative $x_i\in A$, the following claims hold: 
    \begin{enumerate}[label=(\arabic*)]
        \item $\bar Q_i$ is fully dimensional.
        \item For every alternative $x_j\in A\setminus \{x_i\}$, there is a non-zero vector $u^{i,j}\in\mathbb{R}^q$ such that $vu^{i,j}\geq 0$ for all $v\in \bar Q_i$ and $vu^{i,j}\leq 0$ for all $v\in \bar Q_j$. 
        \item For every $x_j\in A\setminus \{x_i\}$, let $u^{i,j}$ denote a non-zero vector such that $vu^{i,j}\geq 0$ if $v\in \bar Q_i$ and $vu^{i,j}\leq 0$ if $v\in \bar Q_j$. It holds that $\bar Q_i=\{v\in\mathbb{R}^q_{\geq 0}\colon \forall x_j\in A\setminus \{x_i\}\colon vu^{i,j}\geq 0\}$.
    \end{enumerate}
\end{lemma}
\begin{proof}
    Fix an alternative $x_i$ and consider the corresponding set $\bar Q_i$. We will prove each of the three claims separately.\medskip

    \textbf{Claim (1):} We will first show that $\bar Q_i$ is fully dimensional by studying $f$ in more detail. Thus, let $\theta\in(0,1)^m$ denote the threshold vector such that $f(\mathcal I)=\max_{\ssucc}\{x_j\in A\colon\Pi_\mathit{SP}(\mathcal I, x_j)\geq \theta_j n_{\mathcal I}\}$ for all $\mathcal I\in \mathcal{D}_1^*$; such a vector exists due to \Cref{lem:generalizedmedian}. Moreover, we define $\delta=\min(\theta_i, 1-\theta_i)$ and we choose $w\in\mathbb{N}$ such that $\frac{2(q-1)}{w}<\delta$. We claim that $f(\mathcal I)=x_i$ for all profiles $\mathcal I$ such that each interval $I\in\Lambda\setminus \{\{x_i\}\}$ is reported by at most two voters and at least $w$ voters report the interval $\{x_i\}$. Assume for contradiction that this is not the case, i.e., that $f(\mathcal I)=x_j$ for such a profile $\mathcal I$ and some alternative $x_j\neq x_i$. Next, let $n_X$ denote the number of voters who report the interval $X$ in the profile $\mathcal I$ and recall that $n_{\mathcal I}$ is the total number of voters in $\mathcal I$. We consider the profile $\mathcal{I}'$ where $\sum_{X\in \Lambda\setminus \{\{x_i\}\}} n_X$ voters report $\{x_j\}$ and $w$ voters report $\{x_i\}$. 
    A repeated application of robustness shows that if $f(\mathcal I)=x_j$, then $f(\mathcal{I}')=x_j$. In more detail, for each voter $k$ with $x_j\not\in I_k$, we can first extend his interval to include $x_j$ and then delete the other alternatives. For each of these steps, robustness implies that the winner is not allowed to change. Similarly, if $x_j\in I_k$, we can one after another remove all alternatives but $x_j$ from the voter's interval, and robustness again demands that the outcome does not change. Next, we observe that $\mathcal{I}'\in\mathcal{D}_1^*$, so $f(\mathcal{I}')=\max_{\ssucc}\{x_h\in A\colon\Pi_\mathit{SP}(\mathcal I', x_h)\geq \theta_h n_{\mathcal I'}\}$. Now, if $x_j\ssucc x_i$, it holds that 
    \begin{align*}
        \Pi_\mathit{SP}(\mathcal{I}',x_j)= \sum_{X\in \Lambda\setminus \{\{x_i\}\}} n_X\leq 2(q-1)<\delta w\leq \theta_i n_{\mathcal I'}.
    \end{align*}
    Since $x_j\ssucc x_i$, it holds that $\theta_j\geq \theta_i$, so $f(\mathcal I')\neq x_j$.
    On the other hand, if $x_i\ssucc x_j$, then 
    \begin{align*}
        \Pi_\mathit{SP}(\mathcal I', x_i)=w\geq n_{\mathcal I}(1-\frac{2(q-1)}{n_{\mathcal I}})>n_{\mathcal I}(1-\delta)\geq n_{\mathcal I}(1-(1-\theta_i))=\theta_i n_{\mathcal I}.
    \end{align*}
    Hence, we derive that $f(\mathcal I')\ssucceq x_i$, which means again that $f(\mathcal I')\neq x_j$. Since we have a contradiction in both cases, it follows that $f(\mathcal I)=x_i$ for all profiles $\mathcal I$ where at least $w$ voters report $\{x_i\}$ and every other interval is reported by at most $2$ voters. Since $f(\mathcal I)=\hat g(v(\mathcal I))$, we conclude that $v\in Q_i\subseteq \bar {Q}_x$ for all vectors $v\in\mathbb{Q}^q_{\geq 0}$ such that $v_{\{x_i\}}\geq w$ and $v_I\leq 2$ for all other intervals $I\in\Lambda\setminus\{\{x_i\}\}$. Combined with the convexity of $\bar Q_i$, this shows that this set is indeed fully dimensional.\medskip

    \textbf{Claim (2):} We next let $x_j\in A\setminus \{x_i\}$ denote a second alternative, and we aim to apply the separating hyperplane theorem for convex sets to infer a non-zero vector $u^{i,j}$ such that $vu^{i,j}\geq 0$ for all $v\in\bar Q_i$ and $vu^{i,j}\leq 0$ for all $v\in \bar Q_j$. To this end, we need to show that the interiors of $\bar Q_i$ and $\bar Q_j$ are disjoint, i.e., that $\text{int } \bar Q_i\cap \text{int }\bar Q_j=\emptyset$. Assume for contradiction that this is not the case, i.e., $\text{int } \bar Q_i\cap \text{int } \bar Q_j\neq\emptyset$. 
    Since $\text{int } \bar Q_i$ and  $\text{int } \bar Q_j$ are open and convex, this means also that $\text{int } \bar Q_i\cap \text{int } \bar Q_j\cap \mathbb{Q}^q\neq\emptyset$. 
    We will show that this implies that $Q_i\cap Q_j\neq\emptyset$, which is a contradiction as these sets are disjoint by definition. To this end, we subsequently prove that $\text{int } \bar Q_i \cap \mathbb{Q}^q \subseteq Q_i$ as a symmetric argument shows that $\text{int } \bar Q_j \cap \mathbb{Q}^q_{\geq 0} \subseteq Q_j$. This means that $Q_i\cap Q_j\neq\emptyset$ because
    $\text{int } \bar Q_i\cap \text{int }  \bar Q_j \cap \mathbb{Q}^q\subseteq Q_i\cap Q_j$.
    
    Now, to show that $\text{int } \bar Q_i \cap \mathbb{Q}^q\subseteq Q_i$, we first note that the convex hull of $Q_i$, i.e., $\textit{Conv}(Q_i)$, is a subset of $\bar Q_i$ because the latter set is convex. 
    Since $Q_i\subseteq \textit{Conv}(Q_i)\subseteq \bar Q_i$, this means that $\bar Q_i\subseteq \overline{\textit{Conv}(Q_i)}\subseteq \overline{\bar{Q_i}}$, so $\overline{\textit{Conv}(Q_i)}=\bar Q_i$. 
    Next, it holds for convex sets $Y$ with non-empty interior that $\text{int } Y=\text{int } \bar Y$, so $\text{int } \bar Q_i=\text{int } \overline{\mathit{Conv}(Q_i)}=\text{int } \mathit{Conv}(Q_i)$. Finally, Lemma 1 of \citet{Youn75a} shows that $\textit{Conv}(Q_i)\cap \mathbb{Q}^q$ is the same as $Q_i$ due to the $\mathbb{Q}$-convexity of this set. Hence, we have that $\text{int }\bar Q_i\cap \mathbb{Q}^q=\text{int }\textit{Conv}(Q_i)\cap \mathbb{Q}^q\subseteq \textit{Conv}(Q_i)\cap \mathbb{Q}^q=Q_i$. This yields the desired contradiction.
    
    We next apply the separation theorem for convex sets to infer that there is a non-zero vector $u^{i,j}\in\mathbb{R}^q$ and a constant $c\in\mathbb{R}$ such that $vu^{i,j}> c$ for all $v\in \text{int } \bar Q_i$ and $vu^{i,j}< c$ for all $v\in \text{int } \bar Q_j$. By taking the closure, it follows that $vu^{i,j}\geq c$ for all $v\in \bar Q_i$ and $vu^{i,j}\leq c$ for all $v\in \bar Q_j$. Moreover, since these sets are closed under multiplication with a positive scalar, it is easy to infer that $c$ must be $0$. This completes the proof of this claim.\medskip

    \textbf{Claim (3):} For all $x_j\in A\setminus \{x_i\}$, let $u^{i,j}\in\mathbb{R}^q$ denote a non-zero vector such that $vu^{i,j}\geq 0$ if $v\in \bar Q_i$ and $vu^{i,j}\leq 0$ if $v\in \bar Q_j$. Moreover, we define $S_i=\{v\in\mathbb{R}^q_{\geq 0}\colon \forall x_j\in A\setminus \{x_i\}\colon vu^{i,j}\geq0\}$ to ease notation, and we will show that $\bar Q_i=S_i$ by considering the two set inclusions between these sets. First, by the definition of the vectors $u^{i,j}$, it follows immediately that if $v\in \bar Q_i$, then $vu^{i,j}\geq 0$ for all $x_j\in A\setminus \{x_i\}$, so $\bar Q_i\subseteq S_i$. For the other direction, we first observe that $\text{int } S_i\neq \emptyset$ as $\text{int } \bar Q_i\neq\emptyset$. 
    Now, it holds by the definition of the interior that, if $v\in \text{int } S_i$, then $vu^{i,j}>0$ for all $x_j\in A\setminus \{x_i\}$. By the definition of these vectors, this means that $v\not\in\bar Q_j$ for all $x_j\in A\setminus \{x_i\}$ since this would imply that $vu^{i,j}\leq 0$. Because $\bigcup_{x_k\in A} \bar Q_k=\mathbb{R}^q_{\geq 0}$, we conclude that $v\in \bar Q_i$. Hence $\text{int } S_i\subseteq \bar Q_i$, and taking the closure of both sets implies that $S_i\subseteq \bar Q_i$. 
\end{proof}

Motivated by \Cref{lem:hyperplane}, we will study the vectors $u^{i,j}$ in more detail. To this end, we define the relation $\ssucc^*$ by $x_i\ssucc^* x_j$ if and only if $x_i\ssucc x_j$ and there is no alternative $x_k$ such that $x_i\ssucc x_k\ssucc x_j$. By our assumption that ${\ssucc}=x_1\ssucc x_2\ssucc \dots \ssucc x_m$, this means that $x_i\ssucc^* x_{j}$ if and only if $j=i+1$. For the next lemmas, we denote by $\theta\in (0,1)^m$ the threshold vector inferred in \Cref{lem:generalizedmedian}. We will next show that the robustness of $f$ severely restricts the vectors $u^{i,j}$ for alternatives $x_i$ and $x_j$ with $x_i\ssucc^* x_j$. 

\begin{lemma}\label{lem:hyperplaneAnalysis}
    Fix two alternatives $x_i,x_j\in A$ such that $x_i\ssucc^* x_j$ and let $u^{i,j}\in\mathbb{R}^q$ denote a non-zero vector such that $vu^{i,j}\geq 0$ for all $v\in \bar Q_i$ and $vu^{i,j}\leq 0$ for all $v\in \bar Q_j$. It holds that 
    \begin{enumerate}[label=(\arabic*)]
        \item $\theta_i\cdot u^{i,j}_{\{x_i\}}=-(1-\theta_i)\cdot u^{i,j}_{\{x_j\}}>0$,
        \item $u^{i,j}_{\{x_i\}}\geq u^{i,j}_X\geq u^{i,j}_{\{x_j\}}$ for all intervals $X\subseteq \Lambda$, and
        \item $u^{i,j}_X=u^{i,j}_Y$ for all intervals $X=[\ell, r]$, $Y=[\ell', r']$ such that either (i) $r\ssucceq x_i$ and $r'\ssucceq x_i$, (ii) $x_j\ssucceq \ell$ and $x_j\ssucceq \ell'$, or (iii) $\{x_i,x_j\}\subseteq X$ and $\{x_i, x_j\}\subseteq Y$.
    \end{enumerate}
\end{lemma}
\begin{proof}
    Fix two alternatives $x_i,x_j\in A$ with $x_i\ssucc^* x_j$ and let $u^{i,j}$ denote the non-zero vector derived in \Cref{lem:hyperplane} that satisfies that $vu^{i,j}\geq 0$ for all $v\in \bar Q_i$ and $vu^{i,j}\leq 0$ for all $v\in \bar Q_j$. 
    We will prove each claim separately, but we will always use the same strategy: if the given equation is violated, we will construct a profile $\mathcal I$ such that $v(\mathcal I) u^{i,j}>0$ (or $v(\mathcal I) u^{i,j}<0$) but $f(\mathcal I)=x_j$ (or $f(\mathcal I)=x_i$). This is a contradiction since $f(\mathcal I)=x_i$ implies that $v(\mathcal I)\in \bar Q_i$ and therefore $v(\mathcal I)u^{i,j}\geq 0$. 
    \medskip

    \textbf{Claim (1):} We will first show that $u^{i,j}_{\{x_i\}}>0$. To this end, we first observe that, by Claim (1) of \Cref{lem:hyperplane}, both $\bar Q_i$ and $\bar Q_j$ are fully dimensional. In particular, this implies for $\bar Q_i$ that $\text{int } \bar Q_i\neq \emptyset$. By Claim (3) of \Cref{lem:hyperplane}, this means that there is a vector $v$ such that $vu^{i,j}>0$. This requires that there is an interval $X$ such that $u^{i,j}_X>0$. Moreover, if $u^{i,j}_Y\geq 0$ for all $Y\in \Lambda$, then $\bar Q_j$ would not be fully dimensional because this implies that $v\in \bar Q_j$ only if $v_X=0$. Hence, there is an interval $Z\in \Lambda$ such that $u^{i,j}_Z<0$. 
    
    Now, let $\delta=\min(\theta_i, 1-\theta_i)$ and define $w\in\mathbb{N}$ such that $\frac{1}{w+1}<\delta$. We consider the profile $\mathcal I$ where a single voter reports $Z$ and $w$ voters report $\{x_i\}$, and we will show that $f(\mathcal I)=x_i$. Assume for contradiction that this is not the case and let $x_k=f(\mathcal I)$ denote the chosen alternative. We consider next the profile $\mathcal I'$ derived from $\mathcal I$ by assigning the interval $\{x_k\}$ to the voter who initially reported $Z$. Since this voter does not remove $x_k$ from his interval (but he may add it), robustness implies that $f(\mathcal I')=f(\mathcal I)=x_k$. Next, since $\mathcal I'\in\mathcal{D}_1^*$, it holds that $f(\mathcal I')=\max_{\ssucc}\{x_h\in A\colon \Pi_\mathit{SP}(\mathcal I', x_h)\geq \theta_h n_{\mathcal I'}\}$ by \Cref{lem:generalizedmedian}. However, if $x_k\ssucc x_i$, we compute that 
    \begin{align*}
        \Pi(\mathcal I', x_k)=1<\delta (w+1) \leq\theta_i n_{\mathcal I'}\leq\theta_kn_{\mathcal I'}. 
    \end{align*}
    Here, the last inequality uses that $\theta_k\geq \theta_i$ as $x_k\ssucc x_i$, and the second to last one uses the definition of $\delta$. This contradicts that $f(\mathcal I')=\max_{\ssucc}\{x_h\in A\colon \Pi_\mathit{SP}(\mathcal I', x_h)\geq \theta_h n_{\mathcal I'}\}=x_k$. 
    As the second case, suppose that $x_i\ssucc x_k$. We observe that 
\begin{align*}
        \Pi_\mathit{SP} (\mathcal I', x_i)=w=n_{\mathcal I'}(1-\frac{1}{w+1})>n_{\mathcal I'} (1-\delta)\geq n_{\mathcal I'} (1-(1-\theta_i))\geq \theta_i n_{\mathcal I'}. 
    \end{align*}
    This shows that $\max_{\ssucc}\{x_h\in A\colon \Pi_\mathit{SP}(\mathcal I', x_h)\geq \theta_h n_{\mathcal I'}\}\ssucceq x_i$, which contradicts $f(\mathcal I')=x_k$. Since we have a contradiction in both cases, it follows that the assumption that $f(\mathcal I)=x_k\neq x_i$ is wrong, i.e., it holds that $f(\mathcal I)=x_i$. This implies that $v(\mathcal I)\in \bar Q_i$, so $v(\mathcal I) u^{i,j}=w u^{i,j}_{\{x_i\}} + u^{i,j}_Z\geq 0$. From this, we finally infer that $u^{i,j}_{\{x_i\}}>0$ because $u^{i,j}_Z<0$. 

    Next, we note that $u^{i,j}_{\{x_i\}}>0$ implies that $u^{i,j}_{\{x_j\}}<0$. To see this, we can consider a profile $\mathcal I''$ where a single voter reports $\{x_i\}$ and $w'$ voters report $\{x_j\}$. Just as before, if $w'$ is large enough, then $f(\mathcal I'')=x_j$. This implies that $v(\mathcal I'')\in \bar Q_j$ and thus $v(\mathcal I'') u^{i,j}=u^{i,j}_{\{x_i\}}+w'u^{i,j}_{\{x_j\}}\leq 0$, which is only possible if $u^{i,j}_{\{x_j\}}< 0$.

    Next, we will prove that $\theta_i u^{i,j}_{\{x_i\}}=-(1-\theta_i) u^{i,j}_{\{x_j\}}$ and we assume for contradiction that this is not the case. 
    We subsequently focus on the case that $\theta_i u^{i,j}_{\{x_i\}}<-(1-\theta_i) u^{i,j}_{\{x_j\}}$; the case that $\theta_i u^{i,j}_{\{x_i\}}>-(1-\theta_i) u^{i,j}_{\{x_j\}}$ follows analogously by exchanging the roles of $x_i$ and $x_j$. 
    By reformulating our assumption, we obtain that $\theta_i<\frac{-u^{i,j}_{\{x_j\}}}{u^{i,j}_{\{x_i\}}-u^{i,j}_{\{x_j\}}}$. Since $-u^{i,j}_{\{x_j\}}>0$, there is a value $\lambda\in (\theta_i, \frac{-u^{i,j}_{\{x_j\}}}{u^{i,j}_{\{x_i\}}-u^{i,j}_{\{x_j\}}})\cap \mathbb{Q}\subseteq  (0,1)\cap \mathbb{Q}$. 
    Moreover, there are two integers $w_1,w_2\in\mathbb{N}$ such that $\lambda=\frac{w_1}{w_1+w_2}$. 
    Now, consider the profile $\mathcal I$ such that $w_1$ voters report $\{x_i\}$ and $w_2$ voters report $\{x_j\}$. 
    Since $\mathcal{I}\in\mathcal{D}_1^*$, $\Pi_\mathit{SP}(\mathcal I,x_i)=w_1>\theta_i n_{\mathcal I}$, and $\Pi_\mathit{SP}(\mathcal I, x_k)=0$ for all $x_k$ with $x_k\ssucc x_i$, we derive that $f(\mathcal I)=\max_{\ssucc}\{x_h\in A\colon \Pi_\mathit{SP}(\mathcal I, x_h)\geq \theta_h n_{\mathcal I}\}= x_i$ and thus $v(\mathcal I)\in \bar Q_i$. 
    On the other hand, it holds that 
    \begin{align*}
        v(\mathcal I)u^{i,j}&=w_1 u^{i,j}_{\{x_i\}}+w_2 u^{i,j}_{\{x_j\}}\\
        &=(w_1+w_2)\left(\lambda u^{i,j}_{\{x_i\}}+(1-\lambda) u^{i,j}_{\{x_j\}}\right)\\
        &<(w_1+w_2)\left(\frac{-u^{i,j}_{\{x_j\}}}{u^{i,j}_{\{x_i\}}-u^{i,j}_{\{x_j\}}}\cdot u^{i,j}_{\{x_i\}} +(1+\frac{u^{i,j}_{\{x_j\}}}{u^{i,j}_{\{x_i\}}-u^{i,j}_{\{x_j\}}})\cdot u^{i,j}_{\{x_j\}}\right)\\
        &=(w_1+w_2)\left(\frac{-u^{i,j}_{\{x_j\}}}{u^{i,j}_{\{x_i\}}-u^{i,j}_{\{x_j\}}}\cdot u^{i,j}_{\{x_i\}}
        +\frac{u^{i,j}_{\{x_i\}}}{u^{i,j}_{\{x_i\}}-u^{i,j}_{\{x_j\}}}
        \cdot u^{i,j}_{\{x_j\}}\right)\\
        &=0.
    \end{align*}

    In the third step, we use that $\lambda<\frac{-u^{i,j}_{\{x_j\}}}{u^{i,j}_{\{x_i\}}-u^{i,j}_{\{x_j\}}}$ as well as $u^{i,j}_{\{x_i\}}>0$ and $u^{i,j}_{\{x_j\}}<0$.
    Finally, we note that this inequality contradicts that $v(\mathcal I)\in \bar Q_i$, so the assumption that $\theta_i u^{i,j}_{\{x_i\}}<-(1-\theta_i) u^{i,j}_{\{x_j\}}$ is wrong. Since the case that $\theta_i u^{i,j}_{\{x_i\}}>-(1-\theta_i) u^{i,j}_{\{x_j\}}$ is symmetric, we conclude that $\theta_i u^{i,j}_{\{x_i\}}=-(1-\theta_i) u^{i,j}_{\{x_j\}}>0$. \medskip

    \textbf{Claim (2):} Consider an arbitrary interval $X$ and assume for contradiction that $u^{i,j}_X\not\in[u^{i,j}_{\{x_j\}}, u^{i,j}_{\{x_i\}}]$. We will assume here that $u^{i,j}_X>u^{i,j}_{\{x_i\}}$ since our cases are again symmetric. Moreover, we let $\delta=u^{i,j}_X-u^{i,j}_{\{x_i\}}$ and choose $t\in\mathbb{N}$ such that $t\cdot \delta>u^{i,j}_{\{x_i\}}-u^{i,j}_{\{x_j\}}$. 
    Next, we choose integers $w_1,w_2\in\mathbb{N}$ such that $w_2>w_1\geq t$, $\frac{t}{w_2}<1-\theta_i$, and $\frac{w_1}{w_2}<\theta_i\leq \frac{w_1+1}{w_2}$. Moreover, we consider the profile $\mathcal I$ in which $w_1$ voters report $\{x_i\}$ and $w_2-w_1$ voters report $\{x_j\}$. It holds that $f(\mathcal I)=\max_{\ssucc}\{x_h\in A\colon \Pi_\mathit{SP}(\mathcal I, x_h)\geq \theta_h n_{\mathcal I}\}$ since $\mathcal{I}\in\mathcal{D}_1^*$. We thus conclude that $f(\mathcal I)=x_j$ because $\Pi_\mathit{SP}(\mathcal I,x_k)=0$ for all $x_k$ with $x_k\ssucc x_i$, $\Pi_\mathit{SP}(\mathcal I, x_i)=w_1<\theta_i w_2=\theta_i n_{\mathcal I}$, and $\Pi_\mathit{SP}(\mathcal I, x_j)=n_{\mathcal I}$. Note that we use here that $x_i\ssucc^* x_j$ as otherwise some alternative $x_k$ with $x_i\ssucc x_k\ssucc x_j$ could be chosen.  
    
    Next, let $\mathcal{I}'$ denote the profile where $w_1-t$ voters report $\{x_i\}$, $t$ voters report $X$, and $w_2-w_1$ voters report $\{x_j\}$. We claim that $f(\mathcal{I}')=x_j$. 
    To show this, let $X=[x_\ell,x_r]$. Now, if $x_r\ssucc x_j$, we can transform the interval $\{x_i\}$ to $X$ by sequentially adding and removing alternatives without touching $x_j$ as $x_i\ssucc^* x_j$. Hence, robustness implies that $f(\mathcal I')=x_j$. On the other hand, if $x_j\ssucceq x_r$ and thus $x_i\ssucc x_r$, we first expand the interval $\{x_i\}$ to the right until our $t$ voters report $[x_i, x_r]$. Repeatedly applying robustness during this process shows that these steps can only move the winner to the right, i.e., that an alternative $y$ with $x_j\ssucceq y$ is now chosen. Finally, we transform the intervals $[x_i, x_r]$ into $[x_\ell, x_r]$ by either adding more alternatives to the left of $x_i$ (if $x_\ell\ssucc x_i$) or by deleting alternatives from the interval (if $x_i\ssucc x_\ell$). In the first case, robustness implies that the winner cannot change as the current winner is right of $x_i$. In the second case, robustness only allows the winner to move further to the right. Hence, it follows that $x_j\ssucceq f(\mathcal I')$. 
    
    Finally, assume for contradiction that $f(\mathcal I')=x_k$ for some $x_k$ with $x_j\ssucc x_k$. In this case, we consider the profile $\mathcal I''$ where $w_1-t$ voters report $\{x_i\}$, $t$ voters report $\{x_k\}$, and $w_2-w_1$ voters report $\{x_j\}$. For this profile, robustness from $\mathcal I'$ implies that $f(\mathcal I'')=x_k$. 
    However, $\mathcal{I}''\in\mathcal{D}_1^*$, so $f(\mathcal I'')=\max_{\ssucc}\{x_h\in A\colon \Pi_\mathit{SP}(\mathcal I'', x_h)\geq \theta_h n_{\mathcal I''}\}$. Hence, we compute that $\Pi_\mathit{SP}(\mathcal I'',x_j)=w_1-t+w_2-w_1=n_{\mathcal I''}(1-\frac{t}{n_{\mathcal I''}})>n_{\mathcal I''}(1-(1-\theta_i))= \theta_i n_{\mathcal I''}\geq \theta_j n_{\mathcal I''}$. The strict inequality here uses that $\frac{t}{n_{\mathcal I''}}=\frac{t}{w_2}<1-\theta_i$ and the last inequality that $\theta_i\geq \theta_j$ as $x_i\ssucc x_j$. This contradicts that $x_j\ssucc f(\mathcal I'')$, so we conclude that $f(\mathcal{I}')=x_j$.

    Finally, we compute $vu^{i,j}$ for the vector $v=v(\mathcal{I}')$. 
    \begin{align*}
        vu^{i,j}&=(w_1-t)u^{i,j}_{\{x_i\}}+t u^{i,j}_X + (w_2-w_1)u^{i,j}_{\{x_j\}}\\
        &=w_1u^{i,j}_{\{x_i\}} + (w_2-w_1) u^{i,j}_{\{x_j\}}+t\delta\\
        &>(w_1+1) u^{i,j}_{\{x_i\}} + (w_2-w_1-1) u^{i,j}_{\{x_j\}}\\
        &=w_2\left(\frac{w_1+1}{w_2} u^{i,j}_{\{x_i\}} + (1-\frac{w_1+1}{w_2}) u^{i,j}_{\{x_j\}}\right)\\
        &\geq w_2\left(\theta_i u^{i,j}_{\{x_i\}} + (1-\theta_i) u^{i,j}_{\{x_j\}}\right)\\
        &=0
    \end{align*}

    The first equality here uses the definition of $v$ (resp. $\mathcal I'$), the second one uses that $\delta=u^{i,j}_X-u^{i,j}_{\{x_i\}}$, and the third one that we choose $t$ such that $t\delta>u^{i,j}_{\{x_i\}}-u^{i,j}_{\{x_j\}}$. The fourth line is a simple transformation, and the fifth inequality uses that $\frac{w_1+1}{w_2}\geq\theta_i$ and $u^{i,j}_{\{x_i\}}>0$ as well as $1-\frac{w_1+1}{w_2}\leq1-\theta_i$ and $u^{i,j}_{\{x_j\}}<0$. The last step follows from Claim (1). However, the observation that $vu^{i,j}>0$ contradicts that $f(\mathcal I')=x_j$ as the latter implies that $v\in \bar Q_j$ and thus $vu^{i,j}\leq 0$. This is the desired contradiction and we thus infer that $u^{i,j}_{\{x_i\}}\geq u^{i,j}_X$. Finally, a symmetric argument shows that $u^{i,j}_X\geq u^{i,j}_{\{x_j\}}$, thus completing the proof of Claim (2).\medskip

    \textbf{Claim (3):} Finally, we will show that $u^{i,j}_X=u^{i,j}_Y$ for all intervals $X=[\ell, r]$, $Y=[\ell',r']$ such that either $r\ssucceq x_i$ and $r'\ssucceq x_i$, $x_j\ssucceq \ell$ and $x_j\ssucceq \ell'$, or $\{x_i, x_j\}\subseteq X\cap Y$. We focus here on the last case, i.e., we assume that $\{x,y\}\subseteq X\cap Y$, and note that all three cases follow from analogous arguments. Moreover, we suppose that $Y\setminus X=\{x_k\}$, i.e., $Y$ arises from $X$ by adding one more alternative $x_k$. This assumption is without loss of generality, because for all intervals $X',Y'$ with $\{x,y\}\subseteq X'\cap Y'$, we can transform $X'$ to $Y'$ by one after another adding and deleting alternatives. 

    Now, assume for contradiction that $u^{i,j}_X\neq u^{i,j}_Y$ and first consider the case that $u^{i,j}_X<u^{i,j}_Y$. We define $\delta=u^{i,j}_Y-u^{i,j}_X$ and let $t\in\mathbb{N}$ denote an integer such that $\delta t>2u^{i,j}_{\{x_i\}}$. 
    Moreover, let $w_1,w_2\in\mathbb{N}$ denote integers such that $\frac{t}{w_1+w_2+t}<\min(\theta_i, 1-\theta_i)$ and $w_1u^{i,j}_{\{x_i\}}+w_2 u^{i,j}_{\{x_j\}}+tu^{i,j}_X<0<w_1u^{i,j}_{\{x_i\}}+w_2 u^{i,j}_{\{x_j\}}+tu^{i,j}_Y$. 
    Such integers exist because we can first set $w_2$ to an arbitrarily large number such that $tu^{i,j}_X+w_2 u^{i,j}_{\{x_j\}}<0$ and then choose $w_1$ such that $-u^{i,j}_{\{x_i\}}\leq w_1u^{i,j}_{\{x_i\}}+w_2 u^{i,j}_{\{x_j\}}+tu^{i,j}_X<0$. 
    Next, let $\mathcal{I}^X$ (resp. $\mathcal{I}^Y$) denote the profile where $w_1$ voters report $\{x_i\}$, $w_2$ voters report $\{x_j\}$, and $t$ voters report $X$ (resp. $Y$), and let $v^X=v(\mathcal{I}^X)$ and $v^Y=v(\mathcal{I}^Y)$. We first observe that, by construction, $v^Xu^{i,j}<0$, which means that $v^X\not\in \bar Q_i$ and hence $f(\mathcal{I}^X)\neq x_i$. Similarly, $v^Yu^{i,j}>0$ and hence $v^Y\not\in \bar Q_j$ and $f(\mathcal{I}^Y)\neq x_j$. 

    Next, we will show that $f(\mathcal{I}^X)\in \{x_i,x_j\}$ and $f(\mathcal{I}^Y)\in \{x_i,x_j\}$. Since the argument for both profiles is symmetric, we assume for contradiction that $f(\mathcal{I}^X)=x_k\not\in \{x_i, x_j\}$. In this case, let $\mathcal{\hat I}^X$ denote the profile where all voters reporting $X$ change their interval to $\{x_k\}$. By repeatedly applying robustness, we infer that $f(\mathcal{\hat I}^X)=x_k$. On the other hand, $\mathcal{\hat I}^X\in\mathcal{D}_1^*$, which implies that $f(\mathcal{\hat I}^X)=\max_{\ssucc}\{x_h\in A\colon \Pi_\mathit{SP}(\mathcal{\hat I}^X, x_h)\geq \theta_h n_{\mathcal{\hat I}^X}\}$. Now, if $x_k\ssucc x_i$, this results in a contradiction as $\Pi_\mathit{SP}(\mathcal{\hat I}^X, x_k)=t<\theta_i (w_1+w_2+t)\leq \theta_k n_{\mathcal{\hat I}^X}$. For the last inequality, we recall that $\theta_k\geq \theta_i$ if $x_k\ssucc x_i$ and that $n_{\mathcal{\hat I}^X}=w_1+w_2+t$. By contrast, if $x_j\ssucc x_k$, we derive a contradiction because $\Pi_\mathit{SP}(\mathcal{\hat I}^X, x_j)=w_1+w_2=n_{\mathcal{\hat I}^X}(1-\frac{t}{w_1+w_2+t})>n_{\mathcal{\hat I}^X}(1-(1-\theta_i))\geq \theta_jn_{\mathcal{\hat I}^X}$. Since $x_i\ssucc^* x_j$, this shows that the assumption that $f(\mathcal{I}^X)\not\in \{x_i,x_j\}$ is wrong. Furthermore, an analogous argument proves that $f(\mathcal{I}^Y)\in \{x_i,x_j\}$. Combined with our previous insights, this means that $f(\mathcal{I}^X)=x_j$ and $f(\mathcal{I}^Y)=x_i$. However, robustness rules out such a deviation because $\{x_i, x_j\}\cap X=\{x_i, x_j\}\cap Y$, i.e., the membership of $x_i$ and $x_j$ in $X$ an $Y$ does not change but this is necessary if the winner changes from $x_i$ to $x_j$. Hence, our assumption that $u^{i,j}_X<u^{i,j}_Y$ must have been wrong. 
    
    As second case suppose that $u^{i,j}_X>u^{i,j}_Y$. In this case, we define $\delta=u^{i,j}_X-u^{i,j}_Y$ and let $t\in\mathbb{N}$ again denote an integer such that $\delta t>2u^{i,j}_{\{x_i\}}$. Moreover, we choose two integers $w_1,w_2\in\mathbb{N}$ such that $\frac{t}{w_1+w_2+t}<\min(\theta_i, 1-\theta_i)$ and $w_1u^{i,j}_{\{x_i\}}+w_2 u^{i,j}_{\{x_j\}}+tu^{i,j}_X>0>w_1u^{i,j}_{\{x_i\}}+w_2 u^{i,j}_{\{x_j\}}+tu^{i,j}_Y$, and define the profiles $\mathcal{I}^X$ and $\mathcal{I}^Y$ as before. Analogous arguments as before show that $f(\mathcal{I}^X)=x_i$ and $f(\mathcal{I}^Y)=x_j$. However, as $Y\setminus X=\{x_k\}$ and $x_k\not\in\{x_i,x_j\}$, this contradicts robustness, so $u^{i,j}_X>u^{i,j}_Y$ is not possible. This means that $u^{i,j}_X=u^{i,j}_Y$. Finally, we note that we never used the fact that $\{x_i, x_j\}\subseteq X\cap Y$, but only that robustness rules out that the winner changes from $x_i$ to $x_j$ for the given modification. Since this holds for all three cases of this claim, it is straightforward to extend the analysis to the remaining cases.
\end{proof}

Motivated by Claim (1) of \Cref{lem:hyperplaneAnalysis}, we will assume from now on that $u^{i,j}_{\{x_i\}}=(1-\theta_i)$ and $u^{i,j}_{\{x_j\}}=-\theta_i$ for all $x_i,x_j\in A$ with $x_i\ssucc^* x_j$. This is without loss of generality because we can scale the vector $u^{i,j}$ arbitrarily and it still separates $\bar Q_i$ from $\bar Q_j$. We will next use our insights to severely simplify the representation of the sets $\bar Q_i$. 

\begin{lemma}\label{lem:hyperplaneOrdering}
    The following claims are true: 
    \begin{enumerate}[label=(\arabic*)]
        \item $\bar Q_1=\{v\in \mathbb{R}^q\colon vu^{1,2}\geq 0\}$.
        \item $\bar Q_i=\{v\in \mathbb{R}^q\colon vu^{i-1,i}\leq 0\land vu^{i,i+1}\geq 0\}$ for all $i\in \{2,\dots, m-1\}$.
        \item $\bar Q_m=\{v\in\mathbb{R}^q\colon vu^{m-1,m}\leq 0\}$.
    \end{enumerate}
\end{lemma}
\begin{proof}
    For proving this lemma, we will first show an auxiliary claim: for all alternatives $x_i,x_j,x_k\in A$ with $x_i\ssucc^* x_j\ssucc^* x_k$, the vectors $u^{i,j}$ and $u^{j,k}$ given by \Cref{lem:hyperplane,lem:hyperplaneAnalysis}, and all vectors $v\in\mathbb{R}^q_{\geq 0}$, it holds that $vu^{i,j}\geq 0$ implies that $vu^{j,k}\geq 0$. In a second step, we prove the lemma.\medskip

    \textbf{Step 1:} Let $x_i,x_j,x_k\in A$ denote alternatives such that $x_i\ssucc^* x_j\ssucc^* x_k$ and assume for contradiction that there is a vector $v\in\mathbb{R}^q_{\geq 0}$ such that $vu^{i,j}\geq 0$ and $vu^{j,k}<0$. Now, if such a vector $v$ exists, there is also a vector $v'$ such that $v'u^{i,j}>0$ and $vu^{j,k}<0$. In more detail, since $u^{i,j}_{\{x_i\}}>0$, we can derive $v'$ from $v$ by marginally increasing $v_{\{x_i\}}$. This shows that the set $\{x\in\mathbb{R}^q_{\geq 0} \colon xu^{i,j}>0\land xu^{j,k}<0\}$ is non-empty, so there also is a vector $v^0\in\mathbb{Q}^q_{\geq 0}\setminus \{0\}$ such that $v^0 u^{i,j}>0$ and $v^0 u^{j,k}<0$. 

    We will next simplify the vector $v^0$ by employing the insights of \Cref{lem:hyperplaneAnalysis}. To this end, we first define the vector $v^1$ by $v_{\{x_j\}}^1=0$ and $v^1_X=v^0_X$ for all intervals $X\in \Lambda \setminus \{\{x_j\}\}$. Since $u^{i,j}_{\{x_j\}}<0$ and $u^{j,k}_{\{x_j\}}>0$ by Claim (1) of \Cref{lem:hyperplaneAnalysis}, it holds for this vector that $v^1u^{i,j}\geq v^0u^{i,j}>0$ and that $v^1u^{j,k}\leq v^0u^{j,k}<0$. 
    
    Next, let $\Lambda_1=\{I\in \Lambda \colon I\subseteq [x_1, x_i]\}$ denote the set of intervals that contain only alternatives weakly left of $x_i$, $\Lambda_2=\{I\in \Lambda \colon I\subseteq [x_k, x_m]\}$ denote the set of intervals that are weakly right of $x_k$, and $\Lambda_3=\{I\in \Lambda \colon \{x_i, x_j, x_k\}\subseteq I\}$ denote the set of intervals that contain $x_i$, $x_j$, and $x_k$. By Claim (3) of \Cref{lem:hyperplaneAnalysis}, we have that \emph{(i)} $u^{i,j}_X=u^{i,j}_{\{x_i\}}$ and $u^{j,k}_X=u^{j,k}_{\{x_i\}}$ for all $X\in \Lambda_1$, \emph{(ii)} $u^{i,j}_X=u^{i,j}_{\{x_k\}}$ and $u^{j,k}_X=u^{j,k}_{\{x_k\}}$ for all $X\in \Lambda_2$, and
        \emph{(iii)} $u^{i,j}_X=u^{i,j}_{\{x_i, x_j, x_k\}}$ and $u^{j,k}_X=u^{j,k}_{\{x_i, x_j, x_k\}}$ for all $X\in \Lambda_3$.
    We hence define the vector $v^2$ by \emph{(i)} $v^2_{\{x_i\}}=\sum_{X\in \Lambda_1} v^1_X$ and $v^2_{X}=0$ for all $X\in \Lambda_1\setminus \{\{x_i\}\}$, \emph{(ii)} $v^2_{\{x_k\}}=\sum_{X\in \Lambda_2} v^1_X$ and $v^2_{X}=0$ for all $X\in \Lambda_2\setminus \{\{x_k\}\}$, \emph{(iii)} $v^2_{\{x_i, x_j, x_k\}}=\sum_{X\in \Lambda_3} v^1_X$ and $v^2_{X}=0$ for all $X\in \Lambda_3\setminus \{\{x_i, x_j, x_k\}\}$, and \emph{(iv)} $v^2_X=v^1_X$ for all $X\in \Lambda\setminus (\Lambda_1\cup\Lambda_2\cup\Lambda_3)$. By our previous insights, it holds that $v^2u^{i,j}=v^1u^{i,j}>0$ and $v^2u^{j,k}=v^1u^{j,k}<0$. 

    For our third modification, let $\Lambda_4=\{I\in \Lambda \colon \{x_i, x_j\}\subseteq I, x_k\not \in I\}$ denote the intervals that contain $x_i$ and $x_j$ but not $x_k$, and let $\Lambda_5=\{I\in \Lambda \colon \{x_j, x_k\}\subseteq I, x_i\not \in I\}$ denote the intervals that contain $x_j$ and $x_k$ but not $x_i$. By Claim (2) of \Cref{lem:hyperplaneAnalysis}, it holds that $u^{i,j}_{\{x_i\}}\geq u^{i,j}_X$ for all $X\in \Lambda_4$ and $u^{j,k}_{\{x_k\}}\leq u^{j,k}_X$ for all $X\in \Lambda_5$. 
    Moreover, Claim (3) of this lemma shows that $u^{j,k}_{\{x_i\}}= u^{j,k}_X$ for all $X\in \Lambda_4$ and $u^{i,j}_{\{x_k\}}=u^{i,j}_X$ for all $X\in\Lambda_5$. 
    We now define our final vector $v^3$: \emph{(i)} $v^3_{\{x_i\}}=v^2_{\{x_i\}}+\sum_{X\in\Lambda_4} v^2_X$ and $v^3_X=0$ for all $X\in \Lambda_4$, \emph{(ii)} $v^3_{\{x_k\}}=v^2_{\{x_k\}}+\sum_{X\in\Lambda_5} v^2_X$ and $v^3_X=0$ for all $X\in \Lambda_5$, and \emph{(iii)} $v^3_X=v^2_X$ for all $X\in \Lambda\setminus (\Lambda_4\cup\Lambda_5\cup \{\{x_i\}, \{x_k\}\})$. Based on our insights from \Cref{lem:hyperplaneAnalysis}, it holds that $v^3u^{i,j}\geq v^2u^{i,j}>0$ and $v^3u^{j,k}\leq v^2u^{j,k}<0$. Moreover, it can be checked that $\Lambda=\{\{x_j\}\}\cup\bigcup_{\ell\in \{1,\dots, 5\}}\Lambda_\ell$, so we have by construction that $v^3_X=0$ for all $X\not\in \{\{x_i\}, \{x_k\}, \{x_i, x_j, x_k\}\}$. Finally, we note that $v^3\in\mathbb{Q}^q_{\geq 0}$ since $v^0\in\mathbb{Q}^q_{\geq 0}$.

    Because $v^3\in\mathbb{Q}^q_{\geq 0}$ (and $v^3\neq 0$ as $v^3u^{i,j}>0$), there is a scalar $\lambda\in\mathbb{N}$ such that $\lambda v^3\in\mathbb{N}_0^q\setminus \{0\}$. Since $\lambda v^3 u^{i,j}>0$, $\lambda v^3 u^{j,k}<0$ and $v^3_X=0$ for all $X\not\in \{\{x_i\}, \{x_k\}, \{x_i, x_j, x_k\}\}$, there are integers $w_1, w_2,t\in\mathbb{N}_0$ such that $w_1 u^{i,j}_{\{x_i\}}+w_2 u^{i,j}_{\{x_k\}}+t u^{i,j}_{\{x_i,x_j,x_k\}}>0$ and $w_1 u^{j, k}_{\{x_i\}}+w_2 u^{j, k}_{\{x_k\}}+t u^{j, k}_{\{x_i,x_j,x_k\}}<0$. 
    Now, let $\mathcal I$ denote the profile where $w_1$ voters report $\{x_i\}$, $w_2$ voters report $\{x_k\}$, and $t$ voters report $\{x_i,x_j,x_k\}$, and let $v^*=v(\mathcal I)$ denote the corresponding vector. First, it is easy to see that $f(\mathcal I)\in \{x_i,x_j,x_k\}$. Indeed, if $f(\mathcal I)\not\in \{x_i,x_j, x_k\}$, then all our voter can deviate to report, e.g., $\{x_i\}$ and robustness implies that the outcome is not allowed to change. However, for the resulting profile $\mathcal{\bar I}$, unanimity requires that $f(\mathcal{\bar I})=x_i$, a contradiction. 
    Next, since $v^*u^{i,j}>0$ and $v^*u^{j, k}<0$, we conclude that $f(\mathcal I)\in \{x_i, x_k\}$. If $f(\mathcal I)=x_k$, we consider the profile $\mathcal{I}'$ derived from $\mathcal I$ by changing the intervals of the $t$ voters who report $\{x_i,x_j,x_k\}$ to $\{x_i,x_j\}$ and the intervals of the $w_2$ voters reporting $\{x_k\}$ to $\{x_j\}$. The conjunction of unanimity and robustness implies that $f(\mathcal{I}')=x_j$. On the other hand, Claim (3) of \Cref{lem:hyperplaneAnalysis} shows that 
    \begin{align*}
        v(\mathcal{I}')u^{i,j}&=w_1 u^{i,j}_{\{x_i\}}+w_2 u^{i,j}_{\{x_j\}}+tu^{i,j}_{\{x_i,x_j\}}
        =w_1 u^{i,j}_{\{x_i\}}+w_2 u^{i,j}_{\{x_k\}}+tu^{i,j}_{\{x_i,x_j,x_k\}}
        >0.
    \end{align*} 

    This implies that $f(\mathcal{I}')\neq x_j$. However, there is no feasible choice left for this profile, so the assumption that $f(\mathcal I)=x_k$ must have been wrong.
    
    Conversely, if $f(\mathcal I)=x_i$, we derive a contradiction by considering the profile $\mathcal{I}''$ where $w_1$ voters report $\{x_j\}$, $w_2$ voters report $\{x_k\}$, and $t$ voters report $\{x_j, x_k\}$. In particular, unanimity and robustness imply for this profile that $f(\mathcal{I}'')=x_j$ but $v(\mathcal I'')u^{j,k}<0$, thus yielding the desired contradiction. Since we have a contradiction in both cases, we finally conclude that if $vu^{i,j}\geq0$ for some vector $v\in\mathbb{R}^q_{\geq 0}$, then $vu^{j,k}\geq0$.\medskip

    \textbf{Step 2:} Next, we will prove the lemma. To this end, fix an alternative $x_i\in A$, let $S_i=\{v\in \mathbb{R}^q_{\geq 0}\colon \forall x_j\in A\setminus \{x_i\}\colon vu^{i,j}\geq 0\}$ and $T_i=\{v\in\mathbb{R}^q_{\geq 0}\colon vu^{i-1,i}\leq 0 \land vu^{i, i+1}\geq 0\}$. Note that, for $T_i$, we define the vectors $u^{0,1}$ (if $i=1$) and $u^{m,m+1}$ (if $i=m$) by $u^{0,1}_X=u^{m,m+1}_X=0$ for all $X\in\Lambda$. First, by \Cref{lem:hyperplane}, it holds that $\bar Q_i=S_i$, so it suffices to show that $S_i=T_i$. To this end, we first note that we can suppose that $u^{i,i-1}=-u^{i-1,i}$ because Claim (3) of \Cref{lem:hyperplane} allows to replace the vector $u^{i, i-1}$ with any non-zero vector $u\in\mathbb{R}^q$ such that $vu\geq 0$ if $v \in \bar Q_i$ and $vu\leq 0$ if $v \in \bar Q_{i-1}$. Since $-u^{i-1, i}$ satisfies this condition, we derive that $T_i=\{v\in\mathbb{R}^q_{\geq 0}\colon vu^{i,i-1}\geq 0\land vu^{i,i+1}\geq 0\}$. By this insight, it is clear that $S_i\subseteq T_i$ because we only remove constraints to infer $T_i$ from $S_i$. 

    Now, assume for contradiction that there is a point $v\in T_i\setminus S_i$. Since $v\in T_i$, we have that $vu^{i,i-1}\geq 0$ and $vu^{i, i+1}\geq 0$. On the other hand, because $v\not\in S_i$, there is an index $k\not\in \{i-1,i,i+1\}$ such that $vu^{i,k}<0$. Next, let $v'$ denote the vector such that $v'_{\{x_i\}}=1$ and $v'_X=0$ for all $X\in\Lambda\setminus \{\{x_i\}\}$. Moreover, we define $v^*=v+\epsilon v'$, where $\epsilon>0$ is so small that $v^*u^{i, k}<0$ still holds. By Claim (1) of \Cref{lem:hyperplaneAnalysis}, we have that $v^*u^{i, i-1}>0$ and $v^*u^{i, i+1}>0$. By Step 1, we derive from $v^*u^{i, i+1}> 0$ that $v^*u^{i+1, i+2}\geq 0$, too. Further, if $v^*u^{i+1, i+2}=0$, we could marginally increase the value of $v^*_{\{x_{i+2}\}}$ to construct a vector $\bar v$ with $\bar v u^{i, i+1}>0$ and $\bar v u^{i+1, i+2}<0$, which contradicts Step 1. Hence, we derive from $v^*u^{i, i+1}> 0$ also that $v^*u^{i+1, i+2}> 0$. Now, by repeatedly applying this reasoning, we conclude that $v^*u^{j,j+1}>0$ for all $j\in \{i+1,\dots, m-1\}$, which means that $v^*\not\in \bar Q_{j+1}$ for all such $j\in \{i, \dots, m-1\}$. Next, we observe that $v^*u^{i,i-1}>0$ means that $v^*u^{i-1,i}<0$. By the contraposition of Step 1, we infer that if $v^*u^{i-1, i}<0$, then $v^*u^{i-2, i-1}<0$. By repeating this argument, it follows that 
    for all $j\in \{2,\dots, i\}$ that $v^*u^{j-1,j}<0$, so $v^*\not\in \bar Q_{j-1}$. Finally, since $v^*u^{i,k}<0$, we also have that $v^*\not\in \bar Q_i$. However, this means that $v^*\not\in \bar Q_j$ for all $j\in \{1,\dots, m\}$. This contradicts that $\bigcup_{x_j\in A} \bar Q_j=\mathbb{R}^q_{\geq 0}$ (which is implied by the basic insight that $\bigcup_{x_j\in A} Q_j=\mathbb{Q}^q_{\geq 0}\setminus \{0\}$). Hence, we have a contradiction, so there is no point $v\in T_i\setminus S_i$. This shows that $T_i\subseteq S_i$, which completes the proof of the lemma. 
\end{proof}

Based on our observations so far, we will now define a weight vector $\alpha$ such that its induced collective position function $\Pi_\alpha$ satisfies that $\Pi_\alpha(v,x_i)\geq \theta_i \sum_{X\in\Lambda} v_X$ if and only if $vu^{i,i+1}\geq 0$. We extend here the definition of collective position functions from interval profiles to vectors $v\in \mathbb{R}^q_{\geq 0}$ by letting $\Pi_\alpha(v,x_i)=\sum_{X\in\Lambda} v_X \pi_\alpha(X,x_i)$. Recall for the subsequent lemma that $u^{i,i+1}_{\{x_i\}}=1-\theta_i$ and $u^{i,i+1}_{\{x_{i+1}\}}=-\theta_i$ for all $i\in \{1,\dots, m-1\}$.

\begin{lemma}\label{lem:Scoring}
    There is a weight vector $\alpha\in[0,1]^m$ and a collective position function $\Pi_\alpha$ such that $\Pi_\alpha(v, x_i)=vu^{i, i+1} + \theta_i \sum_{X\in \Lambda} v_X$ for all $v\in \mathbb{R}^q_{\geq 0}$ and $x_i\in A\setminus \{x_m\}$.
\end{lemma}
\begin{proof}
We define the weight vector $\alpha=(\alpha_1,\dots, \alpha_m)$ by $\alpha_i=u^{i, i+1}_{\{x_i, x_{i+1}\}}+\theta_i$ for all $i\in \{1,\dots, m-1\}$ and $\alpha_m=1$. We moreover note that the value $\alpha_m$ does not matter as $\pi_\alpha(X,x_m)=1$ for all intervals $X\in\Lambda$. First, since $-\theta_i\leq u^{i, i+1}_{\{x_i, x_{i+1}\}}\leq 1-\theta$ by Claim (2) of \Cref{lem:hyperplaneAnalysis}, it holds that $\alpha_i\in [0,1]$. Next, let $v\in\mathbb{R}^q_{\geq 0}$ denote an arbitrary vector and fix an alternative $x_i\neq x_m$. We partition the set of intervals $\Lambda$ with respect to $x_i$: the set $L=\{X\in\Lambda\colon X\subseteq [x_1,x_i]\}$ contains all intervals that are (weakly) left of $x_i$, $M=\{X\in\Lambda\colon \{x_i, x_{i+1}\}\subseteq X\}$ is the set of intervals containing both $x_{i}$ and $x_{i+1}$, and $R=\{X\in\Lambda\colon X\subseteq [x_{i+1},x_m]\}$ are the intervals that are (weakly) right of $x_{i+1}$. By Claim (3) in \Cref{lem:hyperplaneAnalysis}, we have that $u^{i, {i+1}}_X=u^{i,i+1}_{\{x_i\}}=1-\theta_i$ for all $X\in L$, $u^{i, {i+1}}_X=u^{i,i+1}_{\{x_i, x_{i+1}\}}=\alpha_i-\theta_i$ for all $X\in M$, and $u^{i, {i+1}}_X=u^{i,i+1}_{\{x_{i+1}\}}=-\theta_i$ for all $X\in R$. Moreover, for the individual position function induced by $\alpha$, it holds that $\pi_\alpha(X, x_i)=1$ if $X\in L$, $\pi_\alpha(X, x_i)=\alpha_i$ if $X\in M$, and $\pi_\alpha(X, x_i)=0$ if $X\in R$. We thus compute that
\begin{align*}
    \Pi_\alpha(v,x_i)&=\sum_{X\in L} v_X + \alpha_i \sum_{X\in M} v_X\\
    &=(1-\theta_i)\sum_{X\in L} v_X + u^{i,i+1}_{\{x_i, x_{i+1}\}} \sum_{x\in M}v^X -\theta_i \sum_{X\in R} v_X +\theta_i\sum_{X\in \Lambda} v_X\\
    &=vu^{i,i+1}_{\{x_i, x_{i+1}\}} + \theta_i \sum_{X\in \Lambda} v_X.
\end{align*}

This completes the proof of this lemma. 
\end{proof}

We are finally ready to prove \Cref{thm:characterization}.

\robustCharacterization*
\begin{proof}
We have shown the direction from left to right in \Cref{lem:allaxioms}, so we focus here on the converse. 
    Thus, let $f$ denote a voting rule on $\dom$ that satisfies our five axioms. Now, by \Cref{lem:MoulinD1,lem:generalizedmedian}, there is a threshold vector $\theta\in(0,1)^m$ such that $\theta_1\geq\theta_2\geq\dots\geq\theta_m$ and $f(\mathcal I)=\max_{\ssucc}\{x_i\in A\colon \Pi_\mathit{SP}(\mathcal I, x_i)\geq \theta_i n_{\mathcal I}\}$ for all profiles $\mathcal{I}\in\mathcal{D}_1^*$. 
    Next, we note that we can represent interval profiles $\mathcal I$ as vectors $v\in\mathbb{N}_0^q\setminus \{0\}$, where the entry $v_i$ states how often the $i$-th interval is submitted. 
    Moreover, there is a (unique) function $g:\mathbb{N}_0^q\setminus \{0\}\rightarrow A$ such that $f(\mathcal I)=g(v(\mathcal I))$ for all $\mathcal{I}\in\dom$. In \Cref{lem:representation}, we extend this function to the domain $\mathbb{Q}^q_{\geq 0}\setminus \{0\}$, i.e., we show that there is a function $\hat g:\mathbb{Q}^q_{\geq 0}\setminus \{0\}\rightarrow A$ that is reinforcing and satisfies that $f(\mathcal I)=\hat g(v(\mathcal I))$ for all $\mathcal{I}\in\dom$. 
    Based on this function, we define the sets $Q_i=\{v\in\mathbb{Q}^q_{\geq 0}\setminus \{0\}\colon \hat g(v)=x\}$ and let $\bar Q_i$ denote the closure of $Q_i$ with respect to $\mathbb{R}^q$. In a sequence of lemmas (\Cref{lem:hyperplane,lem:hyperplaneAnalysis,lem:hyperplaneOrdering}), we derive that there are non-zero vectors $u^{1,2}, u^{2,3}, \dots, u^{m-1,m}\in\mathbb{R}^q$ such that $\bar Q_i=\{v\in\mathbb{R}^q_{\geq 0}\colon vu^{i-1,i}\leq 0\land vu^{i, i+1}\geq 0\}$ for all $i\in \{1,\dots, m\}$ (where $u^{0,1}=u^{m,m+1}=0$ for notational simplicity). Finally, we show in \Cref{lem:Scoring} that there is a weight vector $\alpha$ such that the corresponding collective position function $\Pi_\alpha$ satisfies for all $i\in \{1,\dots, m-1\}$ and $v\in\mathbb{R}^q_{\geq 0}$ that $\Pi_\alpha(v,x_i)=vu^{i, {i+1}}+\theta_i\sum_{X\in \Lambda} v_X$. We derive from this that $\bar Q_i=\{v\in\mathbb{R}^q_{\geq 0}\colon \Pi_\alpha(v,x_{i-1})\leq \theta_{i-1}\cdot \sum_{X\in\Lambda} v_X\land \Pi_\alpha(v,x_i)\geq \theta_i\cdot \sum_{X\in\Lambda} v_X\}$ for all $i\in \{1,\dots, m\}$ (where we define $\theta_0=\Pi_\alpha(\mathcal I, x_0)=0$ for notational simplicity). Since $\Pi_\alpha(\mathcal I, x_i)=\Pi_\alpha(v(\mathcal I), x_i)$ and $\sum_{X\in \Lambda} v(\mathcal I)_X=n_{\mathcal I}$, we derive the following equation:
    \begin{align*}
        f(\mathcal I)&=\hat g(v(\mathcal I))\\
        &\in \{x_i\in A\colon v(\mathcal I)\in Q_i\}\\
        &\subseteq \{x_i\in A\colon v(\mathcal I)\in \bar Q_i\}\\
        &=\{x_i\in A\colon \Pi_\alpha(\mathcal I,x_{i-1})\leq \theta_{i-1} n_{\mathcal I} \land \Pi_\alpha(\mathcal I,x_i)\geq \theta_i n_{\mathcal I}\}.
    \end{align*}

    Now, we define $O(\mathcal I)=\{x_i\in A\colon \Pi_\alpha(\mathcal I,x_{i-1})\leq \theta_{i-1} n_{\mathcal I} \land \Pi_\alpha(\mathcal I,x_i)\geq \theta_i n_{\mathcal I}\}$ as the set of possible winners of $f$ at the profile $\mathcal I$ and we note that $\max_{\ssucc} O(\mathcal I)=\max_{\ssucc} \{x_i\in A\colon \Pi_{\alpha}(\mathcal I, x_i)\geq \theta_i n_{\mathcal I}\}$. To see this, let $x_j=\max_{\ssucc} O(\mathcal I)$. By definition, we have that $\Pi_\alpha(\mathcal I,x_j)\geq \theta_j n_{\mathcal I}$, so $x_j\in \{x_i\in A\colon \Pi_{\alpha}(\mathcal I, x_i)\geq \theta_i n_{\mathcal I}\}$. This proves that $\max_{\ssucc} \{x_i\in A\colon \Pi_{\alpha}(\mathcal I, x_i)\geq \theta_i n_{\mathcal I}\}\ssucceq \max_{\ssucc} O(\mathcal I)$. Next, let $x_j=\max_{\ssucc} \{x_i\in A\colon \Pi_{\alpha}(\mathcal I, x_i)\geq \theta_i n_{\mathcal I}\}$. If $x_j=x_1$, then $x_j\in O(\mathcal I)$ as the condition on $\theta_0$ is trivial. Otherwise, it holds $x_{j-1}\not \in \{x_i\in A\colon \Pi_{\alpha}(\mathcal I, x_i)\geq \theta_i n_{\mathcal I}\}$, so $\Pi_{\alpha}(\mathcal I, x_{j-1})< \theta_{j-1} n_{\mathcal I}$. This proves that $x_j\in O(\mathcal I)$ and we thus conclude that $\max_{\ssucc} O(\mathcal I)\ssucceq \max_{\ssucc} \{x_i\in A\colon \Pi_{\alpha}(\mathcal I, x_i)\geq \theta_i n_{\mathcal I}\}$. Combining these two observations gives the desired equality. 

    Based on our last insight, we will next show that $f(\mathcal I)=\max_{\ssucc} O(\mathcal I)$ for all profiles $\mathcal I\in \dom$. To this end, we assume for contradiction that there is a profile $\mathcal I$ such that $f(\mathcal I)=x_j\neq x_i=\max_{\ssucc} O(\mathcal I)$. 
    Because $f(\mathcal I)\in O(\mathcal I)$, this means that $x_i\ssucc x_j$. Next, we partition the voters in $N_\mathcal{I}$ into three sets: $L=\{k\in N_{\mathcal I}\colon I_k\subseteq [x_1, x_{j-1}]\}$ contains all voters whose interval is fully left of $x_j$, $M=\{k\in N_{\mathcal I}\colon \{x_i, x_j\}\subseteq I_k\}$ contains all voters who report both $x_i$ and $x_j$, and $R=N_{\mathcal I}\setminus (L\cup M)$ contains all voters whose interval does not contain $x_i$ but an alternative that is weakly right of $x_j$. Now, consider the profile $\mathcal I^1$ where all voters in $L$ report $\{x_i\}$, all voters in $M$ report $[x_i, x_j]$, and all voters in $R$ report $\{x_j\}$. Repeatedly applying robustness shows that $f(\mathcal I^1)=x_j$ because we can transform $\mathcal I$ to $\mathcal I^1$ without removing $x_j$ from the interval of any voter. 
    
    Next, we assume that $j\geq i+2$; otherwise, we can skip the following step. In this case, we consider the profile $\mathcal I^2$ where all voters in $L$ report $\{x_i\}$, all voters in $M$ report $[x_i, x_{j-1}]$, and all voters in $R$ report $\{x_{j-1}\}$. Using robustness from $\mathcal I^1$, we infer that $f(\mathcal I^2)\in \{x_{j-1}, x_j\}$. Moreover, if $f(\mathcal I^2)=x_j$, our voters can deviate to, e.g., unanimously report $\{x_i\}$. Since none of these modifications touches on $x_j$, this alternative has to remain the winner by robustness, but unanimity postulates that $x_i$ is now chosen. This contradiction proves that $f(\mathcal I^2)=x_{j-1}$. Furthermore, by repeating this argument, we derive a profile $\mathcal I^*$ such that all voters in $L$ report $\{x_i\}$, all voters in $M$ report $\{x_i, x_{i+1}\}$, all voters in $R$ report $\{x_{i+1}\}$, and $f(\mathcal I^*)=x_{i+1}$. 

    Next, we compute that $\Pi_\alpha(\mathcal I^*, x_i)=|L|+\alpha_i |M|\geq \Pi_\alpha(\mathcal I, x_i)\geq \theta_i n_{\mathcal I^*}$ because $\pi_\alpha(I_k, x_i)\leq 1$ for all $k\in L$, $\pi_\alpha(I_k, x_i)\leq \alpha_i$ for all $k\in M$ (as $x_j\in I_k$), and $\pi_\alpha(I_k, x_i)=0$ for all $k\in R$ (as all these voters report intervals fully right of $x_i$). On the other hand, since $f(\mathcal{\bar I})\in O(\mathcal{\bar I})$ for all interval profiles $\mathcal{\bar I}$ and $f(\mathcal I^*)=x_{i+1}$, we conclude that $\Pi_\alpha(\mathcal I^*, x_{i})\leq \theta_i n_{\mathcal I^*}$. This proves that $\Pi_\alpha(\mathcal I^*, x_i)=\theta_i n_{\mathcal I^*}$. Now, let $\mathcal I'$ denote the profile where a single voter report $\{x_i\}$. By unanimity, we have that $f(\mathcal I')=x_i$. In turn, right-biased continuity implies that there must be a $\lambda\in\mathbb{N}$ such that $f(\lambda \mathcal I^* + I')=x_{i+1}$. However, it holds for every $\lambda\in\mathbb{N}$ that $\Pi_{\alpha}(\lambda \mathcal I^*+ I', x_i)=\lambda \theta_i n_{\mathcal I^*}+1>\theta_i (\lambda n_{\mathcal I^*} +1)=\theta_i n_{\lambda \mathcal I^*+ I'}$. This shows that $x_{i+1}\not\in O(\lambda \mathcal I^*+ I')$ because the membership of $x_{i+1}$ in this set requires that $\Pi_{\alpha}(\lambda \mathcal I^*+ I', x_i)\leq \theta_i n_{\lambda \mathcal I^*+ I'}$. This is the desired contradiction, so we conclude that $f(\mathcal I)=\max_{\ssucc} O(\mathcal I)=\max_{\ssucc} \{x_i\in A\colon \Pi_\alpha(\mathcal I, x_i)\geq \theta_i n_{\mathcal I}\}$ for all profiles $\mathcal I\in\dom$. Hence, $f$ is the voting rule induced by the weight vector $\alpha$ and the threshold vector $\theta$. Finally, since $f$ is robust, \Cref{lem:robust} shows that these vectors must be compatible. This proves that $f$ is indeed the position-threshold rule defined by $\alpha$ and $\theta$. 
\end{proof}


\begin{thebibliography}{55}
\providecommand{\natexlab}[1]{#1}
\providecommand{\url}[1]{\texttt{#1}}
\expandafter\ifx\csname urlstyle\endcsname\relax
  \providecommand{\doi}[1]{doi: #1}\else
  \providecommand{\doi}{doi: \begingroup \urlstyle{rm}\Url}\fi

\bibitem[Arrow(1951)]{Arro51a}
K.~J. Arrow.
\newblock \emph{Social Choice and Individual Values}.
\newblock New Haven: Cowles Foundation, 1st edition, 1951.
\newblock 2nd edition 1963.

\bibitem[Arrow et~al.(2011)Arrow, Sen, and Suzumura]{ASS11a}
K.~J. Arrow, A.~Sen, and K.~Suzumura, editors.
\newblock \emph{Handbook of Social Choice and Welfare}, volume~2.
\newblock North-Holland, 2011.

\bibitem[Austen-Smith and Banks(1999)]{AuBa00a}
D.~Austen-Smith and J.~S. Banks.
\newblock \emph{Positive Political Theory I: {Collective Preference}}.
\newblock University of Michigan Press, 1999.

\bibitem[Barber{\`a} et~al.(1993)Barber{\`a}, Gul, and Stacchetti]{BGS93a}
S.~Barber{\`a}, F.~Gul, and E.~Stacchetti.
\newblock Generalized median voter schemes and commitees.
\newblock \emph{Journal of Economic Theory}, 61:\penalty0 262--289, 1993.

\bibitem[Berga(1998)]{Berg98a}
D.~Berga.
\newblock Strategy-proofness and single-plateaued preferences.
\newblock \emph{Mathematical Social Sciences}, 35:\penalty0 105--120, 1998.

\bibitem[Berga and Moreno(2009)]{BeMo09a}
D.~Berga and B.~Moreno.
\newblock Strategic requirements with indifference: single-peaked versus single-plateaued preferences.
\newblock \emph{Social Choice and Welfare}, 32:\penalty0 275--298, 2009.

\bibitem[Black(1948)]{Blac48a}
D.~Black.
\newblock On the rationale of group decision-making.
\newblock \emph{Journal of Political Economy}, 56\penalty0 (1):\penalty0 23--34, 1948.

\bibitem[Border and Jordan(1983)]{BoJo83a}
K.~C. Border and J.~S. Jordan.
\newblock Straightforward elections, unanimity and phantom voters.
\newblock \emph{Review of Economic Studies}, 50\penalty0 (1):\penalty0 153--170, 1983.

\bibitem[Brams and Fishburn(2007)]{BrFi07c}
S.~J. Brams and P.~C. Fishburn.
\newblock \emph{Approval Voting}.
\newblock Springer-Verlag, 2nd edition, 2007.

\bibitem[Brandl and Peters(2022)]{BrPe19a}
F.~Brandl and D.~Peters.
\newblock Approval voting under dichotomous preferences: {A} catalogue of characterizations.
\newblock \emph{Journal of Economic Theory}, 205, 2022.

\bibitem[Brandl et~al.(2016)Brandl, Brandt, and Seedig]{Bran13a}
F.~Brandl, F.~Brandt, and H.~G. Seedig.
\newblock Consistent probabilistic social choice.
\newblock \emph{Econometrica}, 84\penalty0 (5):\penalty0 1839--1880, 2016.

\bibitem[Brandt et~al.(2016)Brandt, Conitzer, Endriss, Lang, and Procaccia]{BCE+15a}
F.~Brandt, V.~Conitzer, U.~Endriss, J.~Lang, and A.~D. Procaccia.
\newblock Introduction to computational social choice.
\newblock In F.~Brandt, V.~Conitzer, U.~Endriss, J.~Lang, and A.~D. Procaccia, editors, \emph{Handbook of Computational Social Choice}, chapter~1. Cambridge University Press, 2016.

\bibitem[Chan et~al.(2021)Chan, Filos-Ratsikas, Li, Li, and Wang]{CFLL+21a}
H.~Chan, A.~Filos-Ratsikas, B.~Li, M.~Li, and C.~Wang.
\newblock Mechanism design for facility location problems: a survey.
\newblock In \emph{Proceedings of the 30th International Joint Conference on Artificial Intelligence (IJCAI)}, pages 4356--4365, 2021.

\bibitem[Chatterji and Mass\'o(2018)]{ChMa18a}
S.~Chatterji and J.~Mass\'o.
\newblock On strategy-proofness and the salience of single-peakedness.
\newblock \emph{International Economic Review}, 59:\penalty0 163--189, 2018.

\bibitem[Chatterji and Zeng(2023)]{ChZe21a}
S.~Chatterji and H.~Zeng.
\newblock A taxonomy of non-dictatorial unidimensional domains.
\newblock \emph{Games and Economic Behavior}, 137:\penalty0 228--269, 2023.

\bibitem[Chatterji et~al.(2013)Chatterji, Sanver, and Sen]{CSS13a}
S.~Chatterji, M.~R. Sanver, and A.~Sen.
\newblock On domains that admit well-behaved strategy-proof social choice functions.
\newblock \emph{Journal of Economic Theory}, 148\penalty0 (3):\penalty0 1050--1073, 2013.

\bibitem[Chatterji et~al.(2016)Chatterji, Sen, and Zeng]{CSZ16a}
S.~Chatterji, A.~Sen, and H.~Zeng.
\newblock A characterization of single-peaked preferences via random social choice functions.
\newblock \emph{Theoretical Economics}, 11:\penalty0 711--733, 2016.

\bibitem[Ching(1997)]{Chin97a}
S.~Ching.
\newblock Strategy-proofness and ``median voters''.
\newblock \emph{International Journal of Game Theory}, 26\penalty0 (4):\penalty0 473--490, 1997.

\bibitem[Ehlers and Storcken(2008)]{EhSt08a}
L.~Ehlers and T.~Storcken.
\newblock Arrow's possibility theorem for one-dimensional single-peaked preferences.
\newblock \emph{Games and Economic Behavior}, 64\penalty0 (2):\penalty0 533--547, 2008.

\bibitem[Ehlers et~al.(2002)Ehlers, Peters, and Storcken]{EPS02a}
L.~Ehlers, H.~Peters, and T.~Storcken.
\newblock Strategy-proof probabilistic decision schemes for one-dimensional single-peaked preferences.
\newblock \emph{Journal of Economic Theory}, 105\penalty0 (2):\penalty0 408--434, 2002.

\bibitem[Elkind and Lackner(2015)]{ElLa15a}
E.~Elkind and M.~Lackner.
\newblock Structure in dichotomous preferences.
\newblock In \emph{Proceedings of the 24th International Joint Conference on Artificial Intelligence (IJCAI)}, pages 2019--2025, 2015.

\bibitem[Elkind et~al.(2022)Elkind, Li, and Zhou]{ELZ22a}
E.~Elkind, M.~Li, and H.~Zhou.
\newblock Facility location with approval preferences: Strategyproofness and fairness.
\newblock In \emph{Proceedings of the 21st International Conference on Autonomous Agents and Multiagent Systems (AAMAS)}, pages 391--399, 2022.

\bibitem[Endriss(2013)]{Endr13a}
U.~Endriss.
\newblock Sincerity and manipulation under approval voting.
\newblock \emph{Theory and Decision}, 74\penalty0 (4):\penalty0 335--355, 2013.

\bibitem[Endriss et~al.(2022)Endriss, Novaro, and Terzopoulou]{ENT22a}
U.~Endriss, A.~Novaro, and Z.~Terzopoulou.
\newblock Representation matters: Characterisation and impossibility results for interval aggregation.
\newblock In \emph{Proceedings of the 31st International Joint Conference on Artificial Intelligence (IJCAI)}, pages 286--292, 2022.

\bibitem[Farfel and Conitzer(2011)]{FaCo11a}
J.~Farfel and V.~Conitzer.
\newblock Aggregating value ranges: preference elicitation and truthfulness.
\newblock \emph{Autonomous Agents and Multi-Agent Systems}, 22:\penalty0 127--150, 2011.

\bibitem[Feldman et~al.(2016)Feldman, Fiat, and Golomb]{FFG16a}
M.~Feldman, A.~Fiat, and I.~Golomb.
\newblock On voting and facility location.
\newblock In \emph{Proceedings of the 17th ACM Conference on Economics and Computation (ACM-EC)}, pages 269--286, 2016.

\bibitem[Fishburn(1978)]{Fish78d}
P.~C. Fishburn.
\newblock Axioms for approval voting: {D}irect proof.
\newblock \emph{Journal of Economic Theory}, 19\penalty0 (1):\penalty0 180--185, 1978.

\bibitem[Gibbard(1973)]{Gibb73a}
A.~Gibbard.
\newblock Manipulation of voting schemes: {A} general result.
\newblock \emph{Econometrica}, 41\penalty0 (4):\penalty0 587--601, 1973.

\bibitem[Gibbard(1977)]{Gibb77a}
A.~Gibbard.
\newblock Manipulation of schemes that mix voting with chance.
\newblock \emph{Econometrica}, 45\penalty0 (3):\penalty0 665--681, 1977.

\bibitem[Jennings et~al.(2024)Jennings, Laraki, Puppe, and Varloot]{JLPV24a}
A.~B. Jennings, R.~Laraki, C.~Puppe, and E.~M. Varloot.
\newblock New characterizations of strategy-proofness under sincle-peakedness.
\newblock \emph{Mathematical Programming}, 203:\penalty0 207--238, 2024.

\bibitem[Klaus and Protopapas(2020)]{KlPr20a}
B.~Klaus and P.~Protopapas.
\newblock On strategy-proofness and single-peakedness: median-voting over intervals.
\newblock \emph{International Journal Game Theory}, 49:\penalty0 1059--1080, 2020.

\bibitem[Lackner and Skowron(2021)]{LaSk21a}
M.~Lackner and P.~Skowron.
\newblock Consistent approval-based multi-winner rules.
\newblock \emph{Journal of Economic Theory}, 192:\penalty0 105173, 2021.

\bibitem[Lederer(2024)]{Lede23a}
P.~Lederer.
\newblock Bivariate scoring rules: Unifying the characterizations of positional scoring rules and {K}emeny's rule.
\newblock \emph{Journal of Economic Theory}, 218:\penalty0 105836, 2024.

\bibitem[Maskin(1999)]{Mask99a}
E.~Maskin.
\newblock Nash equilibrium and welfare optimality.
\newblock \emph{Review of Economic Studies}, 66\penalty0 (26):\penalty0 23--38, 1999.

\bibitem[Moulin(1980)]{Moul80a}
H.~Moulin.
\newblock On strategy-proofness and single peakedness.
\newblock \emph{Public Choice}, 35\penalty0 (4):\penalty0 437--455, 1980.

\bibitem[Moulin(1984{\natexlab{a}})]{Moul84a}
H.~Moulin.
\newblock Lecture notes on the theory of voting.
\newblock {M}imeo, 1984{\natexlab{a}}.

\bibitem[Moulin(1984{\natexlab{b}})]{Moul84c}
H.~Moulin.
\newblock Generalized condorcet-winners for single peaked and single-plateau preferences.
\newblock \emph{Social Choice and Welfare}, 1\penalty0 (2):\penalty0 127--147, 1984{\natexlab{b}}.

\bibitem[Moulin(1988)]{Moul88b}
H.~Moulin.
\newblock Condorcet's principle implies the no show paradox.
\newblock \emph{Journal of Economic Theory}, 45\penalty0 (1):\penalty0 53--64, 1988.

\bibitem[Muto and Sato(2017)]{MuSa17a}
N.~Muto and S.~Sato.
\newblock An impossibility under bounded response of social choice functions.
\newblock \emph{Games and Economic Behavior}, 106:\penalty0 1--25, 2017.

\bibitem[Myerson(1995)]{Myer95b}
R.~B. Myerson.
\newblock Axiomatic derivation of scoring rules without the ordering assumption.
\newblock \emph{Social Choice and Welfare}, 12\penalty0 (1):\penalty0 59--74, 1995.

\bibitem[Niemi(1984)]{Niem84a}
R.~G. Niemi.
\newblock The problem of strategic behavior under approval voting.
\newblock \emph{The American Political Science Review}, 78\penalty0 (4):\penalty0 952--958, 1984.

\bibitem[Peters et~al.(2014)Peters, Roy, Sen, and Storcken]{PRSS14a}
H.~Peters, S.~Roy, A.~Sen, and T.~Storcken.
\newblock Probabilistic strategy-proof rules over single-peaked domains.
\newblock \emph{Journal of Mathematical Economics}, 52:\penalty0 123--127, 2014.

\bibitem[Procaccia and Tennenholtz(2013)]{PrTe13a}
A.~D. Procaccia and M.~Tennenholtz.
\newblock Approximate mechanism design without money.
\newblock \emph{ACM Transactions on Economics and Computation}, 1\penalty0 (4), 2013.

\bibitem[Puppe(2018)]{Pupp16a}
C.~Puppe.
\newblock The single-peaked domain revisited: {A} simple global characterization.
\newblock \emph{Journal of Economic Theory}, 176:\penalty0 55--80, 2018.

\bibitem[Pycia and {\"U}nver(2015)]{PyUn15a}
M.~Pycia and M.~U. {\"U}nver.
\newblock Decomposing random mechanisms.
\newblock \emph{Journal of Mathematical Economics}, 61:\penalty0 21--33, 2015.

\bibitem[Reffgen(2015)]{Reff15a}
A.~Reffgen.
\newblock Strategy-proof social choice on multiple and multi-dimensional single-peaked domains.
\newblock \emph{Journal of Economic Theory}, 157:\penalty0 349--383, 2015.

\bibitem[Saijo(1987)]{Saij87a}
T.~Saijo.
\newblock On constant maskin monotonic soical choice functions.
\newblock \emph{Journal of Economic Theory}, 42:\penalty0 382--386, 1987.

\bibitem[Satterthwaite(1975)]{Satt75a}
M.~A. Satterthwaite.
\newblock Strategy-proofness and {A}rrow's conditions: {E}xistence and correspondence theorems for voting procedures and social welfare functions.
\newblock \emph{Journal of Economic Theory}, 10\penalty0 (2):\penalty0 187--217, 1975.

\bibitem[Smith(1973)]{Smit73a}
J.~H. Smith.
\newblock Aggregation of preferences with variable electorate.
\newblock \emph{Econometrica}, 41\penalty0 (6):\penalty0 1027--1041, 1973.

\bibitem[Sprumont(1995)]{Spru95a}
Y.~Sprumont.
\newblock Strategyproof collective choice in economic and political enviornments.
\newblock \emph{Canadian Journal of Economics}, 28\penalty0 (1):\penalty0 68--107, 1995.

\bibitem[Weymark(2011)]{Weym11a}
J.~A. Weymark.
\newblock A unified approach to strategy-proofness for single-peaked preferences.
\newblock \emph{SERIEs}, 2:\penalty0 529--550, 2011.

\bibitem[Young(1975)]{Youn75a}
H.~P. Young.
\newblock Social choice scoring functions.
\newblock \emph{SIAM Journal on Applied Mathematics}, 28\penalty0 (4):\penalty0 824--838, 1975.

\bibitem[Young and Levenglick(1978)]{YoLe78a}
H.~P. Young and A.~B. Levenglick.
\newblock A consistent extension of {C}ondorcet's election principle.
\newblock \emph{SIAM Journal on Applied Mathematics}, 35\penalty0 (2):\penalty0 285--300, 1978.

\bibitem[Zhou et~al.(2023)Zhou, Zhang, Mei, and Li]{ZZML23a}
H.~Zhou, G.~Zhang, L.~Mei, and M.~Li.
\newblock Facility location games with thresholds.
\newblock In \emph{Proceedings of the 22nd International Conference on Autonomous Agents and Multiagent Systems (AAMAS)}, pages 2170--2178, 2023.

\bibitem[Zhou(1991)]{Zhou91b}
L.~Zhou.
\newblock Impossibility of strategy-proof mechanisms in economies with pure public goods.
\newblock \emph{The Review of Economic Studies}, 58\penalty0 (1):\penalty0 107--119, 1991.

\end{thebibliography}
\end{document}